\documentclass[a4paper,
11pt,
openright]{article} % Per avere margini destro e sinistro diversi

\usepackage[english]{babel} % per lingua. TeXnic warning per misteriosi motivi...
\usepackage[utf8]{inputenc}    % smart input of funny chars
\usepackage[T1]{fontenc}       % also for the font encoding
\usepackage{amsmath,amssymb,amsthm,amsfonts,amscd,color,enumerate,eucal,latexsym,mathrsfs,mathtools,multirow,cancel,tikz,ulem}
\usepackage{epsfig}  % ciao
\usepackage{simplewick}
\usepackage[multiuser]{fixme}

\usepackage[
final=true,                    % treat as final
plainpages=false,              % whatever that means
pdfpagelabels=true,            % strange options
pdfencoding=auto,              % getting worse
unicode=true,                  % unicode is always nice
hypertexnames=true,            % needs both to get index *and*
naturalnames=true              % crossrefs correct
]{hyperref}

\usetikzlibrary{matrix,calc}

\FXRegisterAuthor{nd}{annd}{nicolo}
\FXRegisterAuthor{sw}{answ}{stefan}

\newcommand{\supp}{\operatorname{supp}}
\newcommand{\KMS}{\operatorname{KMS}}
\newcommand{\kms}{\operatorname{kms}}
\newcommand{\defterm}[1]{\textbf{#1}}

\definecolor{NiColor}{RGB}{77,77,255}

\newtheoremstyle{TheoremStyle}% <name>
{3pt}% <Space above>
{3pt}% <Space below>
{}% <Body font>
{}% <Indent amount>
{\bf}%{\itshape}% <Theorem head font>
{:}% <Punctuation after theorem head>
{.5em}% <Space after theorem head>
{}% <Theorem head spec (can be left empty, meaning 'normal')>

\theoremstyle{TheoremStyle}

\newtheorem{theorem}{Theorem}
\newtheorem{corollary}[theorem]{Corollary}
\newtheorem{proposition}[theorem]{Proposition}
\newtheorem{lemma}[theorem]{Lemma}
\newtheorem{definition}[theorem]{Definition}
\newtheorem{remark}[theorem]{Remark}
\newtheorem{example}[theorem]{Example} %[theorem]{Example}

\title{Classical KMS Functionals and Phase Transitions in Poisson Geometry}

\author{\textbf{Nicolò Drago}\thanks{\texttt{nicolo.drago@unitn.it}},\\[0.2cm]
  Universit\`a di Trento \\
  Dipartimento di Matematica \\
  Via Sommarive 14 \\
  I-38123 Povo \\
  Italy
  \\[0.5cm]
  \textbf{and}
  \\[0.5cm]
  \textbf{Stefan Waldmann}\thanks{\texttt{stefan.waldmann@uni-wuerzburg.de}},\\[0.2cm]
  Julius Maximilian University of Würzburg \\
  Department of Mathematics \\
  Chair of Mathematics X (Mathematical Physics) \\
  Emil-Fischer-Straße 31 \\
  97074 Würzburg \\
  Germany
}

\date{July 2021}

\begin{document}
\maketitle
\begin{abstract}
    In this paper we study the convex cone of not necessarily smooth
    measures satisfying the classical KMS condition within the context
    of Poisson geometry.  We discuss the general properties of KMS
    measures and their relation with the underlying Poisson geometry in
    analogy to Weinstein's seminal work in the smooth case.  Moreover,
    by generalizing results from the symplectic case, we focus on the
    case of $b$-Poisson manifolds, where we provide an almost complete
    characterization of the convex cone of KMS measures.
\end{abstract}

\tableofcontents

%
% Introduction
%

\section{Introduction}

The quantum KMS condition has been introduced in the context of
quantum statistical mechanics to describe states which are in thermal
equilibrium at a given temperature \cite{Kubo-57, MS59}.  Its
connection with $C^*$-dynamical systems has been studied extensively
\cite{Bratteli-Robinson-87,Bratteli-Robinson-97} and nowadays it is an
established milestone in the Tomita-Takesaki theory.

The classical version of the KMS condition has been introduced in
\cite{Katz-67} and studied in the context of infinite-dimensional
classical systems \cite{Aizemann-Gallavotti-Goldstein-Lebowitz-76,
  Aizenman-Goldstein-Gruber-Lebowitz-Martin-77,
  Gallavotti-Pulvirenti-75, Gallavotti-Pulvirenti-76}.  Abstractly,
this condition can be naturally formulated within the context of
Poisson geometry.  Given a finite-dimensional Poisson manifold
$(M,\Pi)$ together with a vector field $X\in\Gamma(TM)$, a linear,
positive functional
$\varphi\colon C^\infty_{\mathrm{c}}(M)\to\mathbb{C}$ is called
$(X,\beta)$-KMS for $\beta>0$ if
\begin{align*}
    \varphi(\{f,g\})
    =
    \beta\varphi(X(f)g)
    \qquad
    \forall f,g\in C^\infty_{\mathrm{c}}(M)\,.
\end{align*}
From the point of view of Poisson geometry, KMS functionals generalize
$\Pi$-traces, which are included when $X=0$.  The link between KMS
functionals and the underlying geometry of $(M,\Pi)$ has been
described in \cite{Weinstein-97}.  Indeed, when $\varphi$ is induced
by a smooth density ---that is, $\varphi(f) = \int_M f\mu$ with $\mu$
being a strictly positive density on $M$--- the vector field $X$ is
forced to be Poisson and its cohomology class defines a Poisson
invariant known as the modular class.

The classical KMS condition has also been investigated in the
symplectic case by \cite{Basart-Flato-Lichnerowicz-Sternheimer-84,
  Basart-Lichnewicz-85} where the KMS condition for an Hamiltonian
vector field $X=\{\cdot,h\}$ has been interpreted from the point of
view of conformal symplectic geometry.  These results have been
corroborated in \cite{Bordermann-Hartmann-Romer-Waldmann-98}, where
classical KMS functionals have been completely classified in the
symplectic framework for all possible choices of Poisson vector fields $X$. While being
classically of interest already, in that work also the quantum
situation has been clarified yielding a complete classification of KMS
states for formal star products on symplectic manifolds.

In this paper we wish to push ahead these investigations.  Our main
goal is to show that, varying $X$, the set of KMS functionals
$\KMS(X,\beta)$ captures information about the
symplectic foliation of $M$.  This is clear for functionals $\varphi$
induced by smooth densities ---where $X$ is forced to lie in the
modular class--- however, more information is achievable when
considering more general functionals.  The positivity requirement on
$\varphi$ plays a major r\^ole, as it implies that $\varphi$ is a
distribution of order $0$, \textit{i.e.} a probability measure, thus enforcing a minimal requirement of
continuity.

Generally speaking the set $\KMS(X,\beta$) is a convex cone
which heavily depends on the choice of $X$ and $\beta$.  Per
definition, a \defterm{phase transition} occurs whenever different
choices for $X, \beta$ produce non-isomorphic cones.  For example, in
\cite{Bordermann-Hartmann-Romer-Waldmann-98} it is proved that
$\KMS(X,\beta$) is non-empty if and only if $X$ is
Hamiltonian, in which case the cone of KMS functionals is isomorphic
to the one of $\Pi$-traces (which can be shown to be a singleton up to
multiplicative constant). We will recall the precise result later in
Remark~\ref{Rmk: KMS for symplectic manifolds}.  In the general case
we expect phase transitions to reflect the geometry of the underlying
Poisson structure.

Given the large variety of Poisson manifolds, a general description of
$\KMS(X,\beta$) for all $(M, \Pi)$ is probably too ambitious.
In this paper we shall consider the particular class of $b$-Poisson
manifolds \cite{Nest-Tsygan-96}.  These are even dimensional Poisson
manifolds whose associated Poisson tensor $\Pi$ is, roughly speaking,
singular on a codimension $1$ submanifold.  These structures are
degenerate in a very controlled way and they recently attracted a lot
of interest \cite{Braddell-Kiesenhofer-Miranda-18,
  Braddell-Kiesenhofer-Miranda-20, Cavalcanti-17,
  Guillemin-Miranda-Pires-14, Guillemin-Miranda-Pires-Scott-17,
  Guillemin-Miranda-Weitsman-19, Marcut-Torres-14a, Marcut-Torres-14b,
  Miranda-Oms-Salas-20}.

Our main result is Theorem~\ref{Thm: KMS functional space for b-Poisson manifolds}, which provides an almost complete classification of $\KMS(X,\beta$) for all choices of a Poisson vector field $X$ and $\beta>0$.
A similar result concerning $\Pi$-traces is not available, however, we prove a uniqueness result for $\Pi$-traces which are invariant under the modular class in Theorem~\ref{Thm: KMS functional space for cosymplectic manifolds}.
The classification of KMS functionals for $b$-Poisson manifolds profits of an analogous classification obtained for cosymplectic manifolds ---\textit{cf.} Theorem \ref{Thm: KMS functional space for cosymplectic manifolds}.
Within the latter case the classification is complete for compact cosymplectic manifolds with one (hence all) proper leaf.

These results corroborate the geometrical intuition that
$\KMS(X,\beta$) provides information about the geometry of
$(M,\Pi)$.  This correspondence is not completely settled, therefore,
in Section~\ref{Sec: examples} we provide several non-trivial examples
where $\KMS(X,\beta$) can be computed explicitly.  This
allows us to build our intuition for future investigations.

The paper is organized as follows: In Section~\ref{Sec: Main
  properties of classical KMS functionals} we introduce and discuss
the main properties of KMS functionals with a particular focus on
symplectic and cosymplectic manifolds.  Section~\ref{Sec: KMS convex
  cone for b-Poisson manifolds} is devoted to the study of the convex
cone of KMS functionals for the case of a $b$-Poisson manifold.
Section~\ref{Sec: examples} presents two non-trivial examples of
regular Poisson manifolds for which the set of KMS functionals can be
computed explicitly.

%
% Main Properties of Classical KMS Functionals
%

\section{Main Properties of Classical KMS Functionals}
\label{Sec: Main properties of classical KMS functionals}

Let $(M,\Pi)$ be a Poisson manifold of dimension $\dim M=m$.  For the
sake of simplicity we shall always assume that $M$ is orientable,
however, we point out that this is a priori not necessary for stating
Definition \ref{Def: classical KMS functional}.  In what follows $X_g$
will denote the Hamiltonian vector field associated with
$g\in C^\infty(M)$, determined by $X_g(f)=\{f,g\}$ for all
$f\in C^\infty(M)$.  We shall denote by
$\sharp\colon\Gamma(T^*M)\to\Gamma(TM)$ the anchor map
$\alpha^\sharp(f):=\Pi(\mathrm{d}f,\alpha)$.

For a relatively compact subset $U\subseteq M$ we denote by $1_U$ a
smooth compactly supported function $1_U\in C^\infty_{\mathrm{c}}(M)$
such that $1_U|_U=1$.  Similarly for $f\in C^\infty_{\mathrm{c}}(M)$
we set $1_f:=1_{\supp(f)}$.  For a smooth vector field
$X\in\Gamma(TM)$ we shall denote by $\Psi^X$ the associated flow.
\begin{definition}
    \label{Def: classical KMS functional}%
    Let $X\in\Gamma(TM)$ and $\beta>0$.  A linear positive functional
    $\varphi\colon C^\infty_{\mathrm{c}}(M)\to\mathbb{C}$ is called
    \defterm{$(X, \beta)$-KMS functional} if
    \begin{align}
        \label{Eq: classical KMS condition}
        \varphi(\{f,g\})
        =
        \beta\varphi(gX(f))\,,
        \qquad
        \forall f,g\in C^\infty_{\mathrm{c}}(M)\,.
    \end{align}
    The set of $(X, \beta)$-KMS functionals will be denoted by
    $\KMS(X,\beta)$.
\end{definition}
\begin{remark}
    \label{Rmk: on definition of KMS functional}%
    According to Definition \ref{Def: classical KMS functional}, a
    classical KMS functional is not necessarily normalized, therefore
    $\KMS(X,\beta)$ is a convex cone.  Positivity of $\varphi$
    entails that $\varphi$ is a distribution of order $0$, that is,
    for all compact $K\subseteq M$ there exists $C_K>0$ such that
    \begin{align*}
        |\varphi(f)|\leq
        C_K\sup_K|f|\,,\qquad
        \forall f\in C^\infty_{\mathrm{c}}(K)\,.
    \end{align*}
    In particular $\varphi$ admits a unique extension
    $\mu_\varphi\in\mathcal{M}_+(M)$ to a positive measure over
    $M$.
\end{remark}

In the limit case $X=0$ Equation~\eqref{Eq: classical KMS condition}
identifies positive \defterm{Poisson traces} which we shall denote by
$\KMS(M,\Pi)$.
Moreover, for all $\beta> 0$ we have
$\KMS(X,\beta)=\KMS(\beta X,1)$ while, for all
$\beta,\beta'> 0$ it holds
$\varphi\in\KMS(X,\beta) \cap \KMS(X,\beta')$ if and only
if $\varphi\in\KMS(M,\Pi)$.  This may happen whenever
$\supp(\varphi)\cap\supp(X) = \emptyset$ as shown by the following
example.
\begin{example}
    Let $M = \mathbb{R}\times\mathbb{T}$ with Poisson structure
    $\Pi = x\partial_x\wedge\partial_\theta$ and let $X:=x\partial_x$.
    Then any $\varphi\in\KMS(M,\Pi)$ such that
    $\supp(\varphi)\subseteq\{x=0\}$ defines an element
    $\varphi \in \KMS(x\partial_x,\beta)$ for all $\beta>0$.  In
    fact, $\KMS(x\partial_x,\beta)$ is exactly made by these
    functionals as a consequence of Theorem \ref{Thm: KMS functionals
      for vartheta Hamiltonian in Mz}.
\end{example}
\begin{remark}[Symplectic manifolds]\label{Rmk: KMS for symplectic manifolds}
    For a symplectic manifold the convex cone $\KMS(X,\beta)$
    have been classified for all Poisson vector field $X$,
    \textit{cf.} \cite[Thm. 8]{Bordermann-Hartmann-Romer-Waldmann-98}.
    In particular if $(M,\Pi)$ is a connected symplectic manifold and
    $[X]\in H_\Pi^1(M)$ denotes the Poisson cohomology class of $X$,
    then
    \begin{align}
        \label{eq:KMSSymplecticCase}
        \KMS(X,\beta)
        \simeq
        \begin{cases}
            [0,+\infty)
            & \textrm{for } [X]=[0]
            \\
            \{0\}
            & \textrm{for } [X]\neq[0]\,,
        \end{cases}
    \end{align}
    where $\simeq$ denotes isomorphism of convex cone.  In particular
    if $X=X_h$ we have
    \begin{align}
        \label{Eq: Gibbs functionals}
        \KMS(X_h,\beta)
        =\left\lbrace
            cI_{\mu_\beta}
            \; \Big|\;
            \mu_\beta=e^{-\beta h}\mu_0\,,\; c>0
        \right\rbrace\,,\qquad
        I_\mu(f):=\int_M f\mu\,,
    \end{align}
	where $\mu_0$ is the symplectic volume form of $M$.
\end{remark}

%
% Invariance and Support Properties
%

\subsection{Invariance and Support Properties}

The following proposition collects useful properties of KMS
functionals.  In particular we prove Equation~\eqref{Eq: global
  classical KMS condition}, which is a global counterpart of the
infinitesimal classical KMS condition \eqref{Eq: classical KMS
  condition} -- \textit{cf.}
\cite{Aizemann-Gallavotti-Goldstein-Lebowitz-76} for a different
though similar in spirit dynamical condition.  Such condition
generalizes the invariance of $\Pi$-traces under Hamiltonian flows and
implies in particular a support constrain on symplectic leaves.
\begin{proposition}
    \label{Prop: properties of KMS functionals}%
    Let $\varphi\in\KMS(X,\beta)$ as per Definition~\ref{Def:
      classical KMS functional}.  Then:
    \begin{enumerate}
    \item\label{Item: KMS X-invariance} $\varphi$ is invariant under
        the flow of $X$, that is
        \begin{align}
            \label{Eq: KMS X-invariance}
            \varphi\circ\Psi^X=\varphi\,.
        \end{align}
    \item\label{Item: global classical KMS condition} For all complete
        Hamiltonian vector fields $X_g$ it holds
        \begin{align}
            \label{Eq: global classical KMS condition}
            \varphi(f)
            =
            \varphi\bigg(
            \exp\bigg[
            \beta\int_0^t[\Psi_s^g]^*X(g)\mathrm{d}s
            \bigg](\Psi_t^g)^*f
            \bigg)\,,
            \qquad
            \forall t\in\mathbb{R}\,.
        \end{align}
    \item\label{Item: KMS support properties} If $M$ is compact then
        $x \in \supp(\varphi)$ implies $L_x \subseteq \supp(\varphi)$,
        where $L_x$ is the symplectic leaf through $x$.  In particular
        \begin{align*}
            \supp(\varphi)=\bigcup_{x\in\supp(\varphi)}L_x\,.
        \end{align*}
    \end{enumerate}
\end{proposition}
\begin{proof}
	\noindent
    \begin{description}
    \item[\ref{Item: KMS X-invariance}] For all
        $f\in C^\infty_{\mathrm{c}}(M)$ let
        $1_f\in C^\infty_{\mathrm{c}}(M)$ be such that
        $1_f|_{\supp(f)}=1$.  A direct computation leads to
        \begin{align*}
            \beta\varphi(X(f))
            =\beta\varphi(1_fX(f))
            \stackrel{\eqref{Eq: classical KMS condition}}{=}
            \varphi(\{1_f,f\})
            =0\,.
        \end{align*}
        If follow that $\varphi\circ X=0$ which entails the invariance
        property $\varphi\circ\Psi^X=\varphi$.

    \item[\ref{Item: global classical KMS condition}]
        For all $f\in C^\infty_{\mathrm{c}}(M)$ and $g\in C^\infty(M)$ a direct computation
        leads to
        \begin{align*}
            \varphi(X_g(f))
            =\varphi(\{f,g\})
            =\varphi(\{f,1_fg\})
            \stackrel{\eqref{Eq: classical KMS condition}}{=}
            \beta\varphi(gX(f))
            \stackrel{\eqref{Eq: KMS X-invariance}}{=}-\beta\varphi(X(g)f)\,.
        \end{align*}
        The previous equality entails $\varphi\circ D_{g,X}=0$ where
        $D_{g,X}$ is the first order differential operator
        $D_{g,X}:=X_g+\beta X(g)$.  Integrating this equality leads to
        Equation~\eqref{Eq: global classical KMS condition}.

    \item[\ref{Item: KMS support properties}] Finally we prove that
        $L_x\subset\supp(\varphi)$ if $x\in\supp(\varphi)$.  We recall
        that for any $y\in L_x$ there exist $n\in\mathbb{N}$,
        $t_1, \ldots, t_n\in\mathbb{R}$ and
        $g_1, \ldots, g_n\in C^\infty(M)$ such that
        $y=\Psi^{g_1}_{t_1}\circ \cdots \circ \Psi^{g_n}_{t_n}(x)$.
        Therefore it is enough to show that
        $\Psi^g_{-t}(x)\in\supp(\varphi)$ if $x\in\supp(\varphi)$.

		By definition, $x \in \supp(\varphi)$ if for all open neighbourhood $U_x$ of $x$ there exists a positive function
		$f\in C^\infty_{\mathrm{c}}(U_x)$ such that $\varphi(f)>0$.
		
		Let now $U_{t,x}$ be an arbitrary open neighbourhood of $\Psi_{-t}^g(x)$.
		By possibly shrinking $U_{t,x}$ we may assume without loss of generality that $U_{t,x}=\Psi^g_{-t}(U_x)$ for an open neighbourhood $U_x$ of $x$.
		Since by assumption $x\in\operatorname{supp}(\varphi)$, there exists a positive function $f\in C^\infty_{\mathrm{c}}(U_x)$ such that $\varphi(f)>0$.
		Let now
		\begin{equation*}
			f_t
			:=
			\exp\left(
			\beta\int_0^t[\Psi^g_s]^*X(g)\mathrm{d}s
			\right)[\Psi^g_t]^*f
			\in
			C^\infty_{\mathrm{c}}(M).
		\end{equation*}
		It follows that $f_t\in C^\infty_{\mathrm{c}}(U_{t,x})$ and
		$\varphi(f_t)=\varphi(f)>0$ because of Equation~\eqref{Eq: global classical KMS condition}.
		Therefore
		$\Psi_{-t}(x)\in\supp(\varphi)$.
    \end{description}
\end{proof}
\begin{remark}
    We stress that the completeness of $X_g$ in Item \ref{Item: global
      classical KMS condition} is assumed only for the sake of
    simplicity: If $X$ is not complete Equation \eqref{Eq: global
      classical KMS condition} still holds true on small open subsets
    of $M$ provided $t$ is chosen so that $\Psi_s^g$ exists for all
    $s \in [0,t]$.  A similar comment applies concerning the
    compactness of $M$ in Item \ref{Item: KMS support properties}.
\end{remark}

%
%  Stability under Hamiltonian Perturbation and Relation with
%  Weinstein's Modular Class
%

\subsection{Stability under Hamiltonian Perturbation and Relation with Weinstein's Modular Class}

Definition~\ref{Def: classical KMS functional} does not require any
particular assumption on $X\in\Gamma(TM)$.  However, the forthcoming
Lemma~\ref{Lem: existence of faithful KMS implies that X is Poisson}
shows that under mild conditions the vector field $X$ has to be a
Poisson vector field.  We recall that a positive functional
$\varphi\colon C^\infty_{\mathrm{c}}(M)\to\mathbb{C}$ is called
\defterm{faithful} if, for all positive
$f\in C^\infty_{\mathrm{c}}(M)$, $\varphi(f)=0$ implies $f=0$.
\begin{lemma}
    \label{Lem: existence of faithful KMS implies that X is Poisson}%
    Let $\varphi\in\KMS(X,\beta)$ be faithful.  Then $X$ is
    Poisson.
\end{lemma}
\begin{proof}
    Let $\varphi\in\KMS(X,\beta)$ and let us consider
    $f,g,h\in C^\infty_{\mathrm{c}}(M)$.  A direct inspection leads to
    \begin{align*}
        \beta \varphi(gX(\{f,h\}))
        &=
        \varphi(\{\{f,h\},g\}) \\
        &=
        -\varphi(\{\{h,g\},f\})
        -\varphi(\{\{g,f\},h\})
        \\
        &=
        -\beta \varphi(fX(\{h,g\}))
        -\beta \varphi(hX(\{g,f\}))
        \\
        &=
        \beta \varphi(X(f)\{h,g\})
        +\beta \varphi(X(h)\{g,f\})
        \\
        &=
        \beta \varphi(\{h,X(f)g\})
        -\beta \varphi(g\{h,X(f)\})
        \\
        &\quad+
        \beta \varphi(\{gX(h),f\})
        -\beta \varphi(g\{X(h),f\})\,,
    \end{align*}
    where we exploited the KMS condition \eqref{Eq: classical KMS
      condition}.  Moreover, the KMS condition also implies
    $\beta \varphi(\{h,X(f)g\})+\beta \varphi(\{gX(h),f\})=0$, so that
    we end up with
    \begin{align*}
        \varphi(g(X(\{f,h\})-\{X(f),h\}-\{f,X(h)\}))=0\,.
    \end{align*}
    If now $\varphi$ is also faithful, then the arbitrariness of
    $g\in C^\infty_{\mathrm{c}}(M)$ entails
    $X(\{f,h\})=\{X(f),h\}+\{f,X(h)\}$, that is, $X$ is Poisson.
\end{proof}

On compact symplectic manifolds, Proposition \ref{Prop: properties of
  KMS functionals} entails $\supp(\varphi)=M$ for all non-zero
$\varphi\in\KMS(X,\beta)$, no matter the choice of $X$ and
$\beta$.  If follows that $\varphi$ is faithful and therefore $X$ is
Poisson because of Lemma \ref{Lem: existence of faithful KMS implies
  that X is Poisson}.
This entails in particular that $\KMS(X,\beta)=\{0\}$ if $X$ is
not Poisson.

Choosing a Poisson vector field $X$ provides a nice interplay with the underlying Poisson structure.
In this respect Lemma \ref{Lem: existence of faithful KMS implies that X is Poisson} forces $X$ to be Poisson in the presence of a faithful KMS functional $\varphi\in\KMS(X,\beta)$.
However, we stress that in general our analysis will not be limited to faithful KMS functionals.

Lemma \ref{Lem: stability of KMS} shows that, for a Poisson vector field $X$, the set $\KMS(X,\beta)$ only depends on the Poisson cohomology class of $X$.
\begin{lemma}
    \label{Lem: stability of KMS}%
    Let $X\in\Gamma(TM)$ be a Poisson vector field.  Then for all
    $\lambda\in C^\infty(M)$ and $\beta> 0$
    \begin{equation}
        \label{eq:PerturbalteKMSByHamiltonian}
        \KMS(X,\beta)
        \ni \varphi \mapsto \varphi_\lambda
        \in
        \KMS(X-X_\lambda,\beta)\,,
        \qquad
        \varphi_\lambda(f):=\varphi(e^{\beta\lambda}f)\,,
    \end{equation}
    is an isomorphism of convex cones.  In particular if $X=X_h$ is
    Hamiltonian, where $h\in C^\infty(M)$, then for all $\beta> 0$
    \begin{align*}
        \KMS(X_h,\beta)
        \simeq
        \KMS(M,\Pi)\,.
    \end{align*}
\end{lemma}
\begin{proof}
    Let $\varphi\in\KMS(X,\beta)$ and set
    $\varphi_\lambda(f):=\varphi(e^{\beta \lambda}f)$ for all
    $f\in C^\infty_{\mathrm{c}}(M)$.  We claim that
    $\varphi_\lambda\in\KMS(X-X_\lambda,\beta)$.  Indeed
    $\varphi_\lambda$ is linear and positive, while for all
    $f,g\in C^\infty_{\mathrm{c}}(M)$ we have
    \begin{align*}
	\varphi_\lambda(\{f,g\})
        =
	\varphi(e^{\beta\lambda}\{f,g\})
        =
	\varphi(\{f,ge^{\beta\lambda}\}-g\{f,e^{\beta\lambda}\})
        =
	\beta\varphi(ge^{\beta\lambda}X(f))
	-\beta\varphi(ge^{\beta\lambda}X_\lambda(f))\,.
    \end{align*}
    By direct inspection, the map
    $\KMS(X,\beta) \ni \varphi \mapsto \varphi_\lambda \in
    \KMS(X-X_\lambda,\beta)$ is an isomorphism of convex cones.
\end{proof}
\begin{remark}\label{Rmk: stability of KMS}
    Equation~\eqref{eq:PerturbalteKMSByHamiltonian} in Lemma~\ref{Lem:
      stability of KMS} can be interpreted as a stability property of
    the set of KMS functionals.  Indeed, let
    $f \in C^\infty_{\mathrm{c}}(M)$ and let $X\in\Gamma(TM)$ be a
    Poisson vector field: in this situation the vector field $X+X_f$
    is interpreted as a small perturbation of the dynamics associated
    with $X$.  Lemma~\ref{Lem: stability of KMS} ensures that for any
    $\varphi\in\KMS(X,\beta)$ there exists a perturbed KMS
    functional $\varphi_f\in\KMS(X+X_f,\beta)$.  In this sense
    we may say that $\KMS(X,\beta)$ is stable under small (in
    the sense of being inner) perturbations.  This feature is the
    classical analogue of the stability property for quantum KMS
    functionals in the $C^*$-algebraic setting \cite{Araki-73,
      Haag-Kastler-Trych-Pohlmeyer-74, Haag-Hugenholtz-Winnik-67} --
    see also \cite{Drago-Faldino-Pinamonti-18, Drago-Gerard-17,
      Drago-Hack-Pinamonti-17, Fredenhagen-Lindner-14} for similar
    results in the $*$-algebraic setting.
\end{remark}

Lemma~\ref{Lem: existence of faithful KMS implies that X is Poisson}
and Lemma~\ref{Lem: stability of KMS} suggest to restrict our
investigation to the case of $(X, \beta)$-KMS functionals for Poisson
vector fields only.
\begin{definition}
    \label{Def: KMS functional space}%
    Let $X\in TM$ be a Poisson vector field.  Under the identification
    proved in Lemma~\ref{Lem: stability of KMS} we denote by
    $\KMS([X],\beta)$ the set of classical KMS functionals
    associated with $[X]\in H_\Pi^1(M)$, where $H_\Pi^1(M)$ denotes
    the first Poisson cohomology group.  We will also let
    $\KMS(M,\Pi) = \KMS([0],\beta)$ denote the set of positive
    $\Pi$-traces over $M$.
\end{definition}

We now show that, at least for a specific choice of
$[X]\in H_\Pi^1(M)$ and $\beta>0$ one has
$\KMS([X],\beta)\neq\{0\}$, \textit{cf.} Remark~\ref{Rmk:
  properties of the modular class}.  This also gives an indication on
the connection between the convex cone of KMS functionals and the
underlying Poisson structure.
\begin{definition}
    \label{Def: regular classical KMS functional}%
    A classical KMS functional $\varphi$ is called \defterm{regular}
    if there exists a volume form $\mu\in \Omega^m(M)$ such that
    \begin{align}\label{Eq: regularity condition for KMS functionals}
        \varphi(f)
        =
        I_\mu(f):=\int_M f\mu
        \qquad
        \forall f\in C^\infty_{\mathrm{c}}(M)\,.
    \end{align}
\end{definition}

If $\mu\in\Omega^m(M)$ is a volume form and
$f,g\in C^\infty_{\mathrm{c}}(M)$ we have
\begin{align}
    \label{Eq: KMS identities for smooth density}%
    I_\mu(\{f,g\})=
    -\int_M\mathcal{L}_{X_f}(g)\mu=
    \int_Mg\mathcal{L}_{X_f}(\mu)=
    \int_M g\operatorname{div}_\mu(X_f)\mu\,.
\end{align}
This leads to the following result:
\begin{remark}[Weinstein \cite{Weinstein-97}]
    \label{Rmk: properties of the modular class}%
    Let $\mu\in \Omega^m(M)$ be a volume form and let
    \begin{align*}
        Y_\mu(f)
        :=\operatorname{div}_\mu(X_f)
        \qquad\forall f\in C^\infty(M)\,.
    \end{align*}
    Then $Y_\mu$ is a Poisson vector field and
    $I_\mu\in\KMS([Y_\mu],1)$.  Moreover the Poisson cohomology
    class $[Y_\mu]$ does not depend on $\mu$.  This Poisson cohomology
    class is called the \defterm{modular class} of $(M,\Pi)$ and it is
    denoted with $[Y_\Pi]$.  A Poisson manifold with vanishing modular
    class is called \defterm{unimodular}.
\end{remark}

Equation~\eqref{Eq: KMS identities for smooth density} entails in
particular that regular $(X, \beta)$-KMS functionals exist if and only
if $\beta X\in[Y_\Pi]$.
\begin{remark}
    \label{Rmk: KMS and homology}%
    If $M$ is compact, any Poisson trace can be thought as a positive
    linear functional on the first homology group
    $H^\Pi_1(M):=\dfrac{C^\infty(M)}{\{C^\infty(M),C^\infty(M)\}}$.  A
    similar interpretation is possible for $(X,\beta)$-KMS
    functionals.  For that we have to introduce the complex
    \begin{align*}
        \{0\}
        \stackrel{\delta_\beta}{\longleftarrow}
        \Omega^0(M)
        \stackrel{\delta_\beta}{\longleftarrow}
        \Omega^1(M)
        \stackrel{\delta_\beta}{\longleftarrow}
        \cdots\stackrel{\delta_\beta}{\longleftarrow}
        \Omega^m(M)\,,
    \end{align*}
    where
    $\delta_\beta:=\mathrm{d}\iota_\Pi-\iota_\Pi\mathrm{d}-\beta\iota_X$.
    A direct inspection shows that $\delta_\beta^2=0$, therefore we
    can define the $(X,\beta)$-homology group
    $H^\beta_\bullet(M):=\dfrac{\ker\delta_\beta}{\operatorname{im}\delta_\beta}$.
    Moreover a direct computation leads to
    \begin{align*}
        \delta_\beta(g\mathrm{d}f)=\{f,g\}-\beta gX(f)\,,
    \end{align*}
    which shows that positive linear functionals on $H^\beta_0(M)$
    satisfies the classical KMS condition \eqref{Eq: classical KMS
      condition}.

    For the case of a unimodular Poisson manifold, \textit{i.e.}
    $[Y_\Pi]=[0]$, Poisson homology and cohomology are isomorphic
    \cite[Prop. 4.18]{LaurentGengoux-Pichereau-Vanhaeckie-13}.  In
    fact, $H^k_\Pi(M)\ni[X] \mapsto [\iota_X\mu]\in H^\Pi_{m-k}(M)$ is
    an isomorphism. Here $\mu$ is any volume form such that $Y_\mu=0$.
    Similarly we may define a $(X,\beta)$-Poisson cohomology
    $H_\beta^\bullet(M)$ such that $[X]\to[\iota_X\mu]$ descends to an
    isomorphism $H_\beta^k(M)\simeq H^\beta_{m-k}(M)$.  To this avail
    it suffices to consider the complex
    \begin{align*}
        C^\infty(M)
        \stackrel{\mathrm{d}_\beta}{\longrightarrow}
        \Gamma(TM)
        \stackrel{\mathrm{d}_\beta}{\longrightarrow}
        \Gamma(\wedge^2TM)
        \stackrel{\mathrm{d}_\beta}{\longrightarrow}
        \cdots
        \stackrel{\mathrm{d}_\beta}{\longrightarrow}
        \Gamma(\wedge^mTM)\,,
    \end{align*}
    where $\mathrm{d}_\beta=[\Pi,\;]-\beta X\wedge$.  A direct
    inspection shows that $\mathrm{d}_\beta^2=0$, therefore
    $H_\beta^\bullet(M) :=
    \dfrac{\ker\mathrm{d}_\beta}{\operatorname{im}\mathrm{d}_\beta}$
    are well-defined.  The proof that
    $H_\beta^k(M)\ni[X] \mapsto [\iota_X\mu]\in H^\beta_{m-k}(M)$ is a
    well-defined isomorphism follows the one presented in
    \cite[Prop. 4.18]{LaurentGengoux-Pichereau-Vanhaeckie-13}.
\end{remark}

%
% Extremal KMS Functionals
%

\subsection{Extremal KMS Functionals}

By the very definition, the set of KMS functionals
$\KMS([X],\beta)$ is a convex cone.  For a compact manifold $M$
we can normalize all KMS functionals by setting
$\widehat{\varphi}(f):=\varphi(1)^{-1}\varphi(f)$. Such normalized
positive functionals are referred to as \defterm{states}.  In this way
we are reduced to consider the convex set of KMS states which we shall
denote by $\kms([X],\beta)$.

We recall that for a convex subset $C$ of a real vector space $V$
a point $x \in C$ is called \defterm{extremal} if for all
$x_1, x_2 \in C$ such that $x=\lambda_1x_1+\lambda_2x_2$ for
$\lambda_1,\lambda_2\in(0,1)$ with $\lambda_1+\lambda_2=1$, we have
$x=x_1=x_2$.  In what follows we shall denote by
$\kms_0([X],\beta)\subseteq\kms([X],\beta)$ the set of
extremal KMS states.

A standard application of the Krein-Milman Theorem entails that the
convex hull of extremal KMS states generates all KMS states:
\begin{align*}
    \kms([X],\beta)
    =
    \overline{\operatorname{conv}}(\kms_0([X],\beta))\,,
\end{align*}
where the bar denotes closure with respect to the weak* topology of
$\mathcal{M}_+(M)$.

Indeed Equation~\eqref{Eq: classical KMS condition} entails that
$\kms([X],\beta)$ is a weakly*-closed subset of the
(weakly*-compact) unit ball of the Banach space $C(M)'$.  Thus
$\kms([X],\beta)$ is weakly*-compact and the Krein-Milman
theorem applies.  Therefore, whenever $M$ is compact we may focus on
$\kms_0([X],\beta)$.

Extremal KMS states are interpreted as pure thermodynamical phases of
the system under investigation and it is natural to conjecture that
extremal states have to be supported on the leaves of $M$. Indeed, for
leaves that embed nicely into $M$ we have the following statement.
\begin{proposition}
    \label{Prop: extremal KMS supported on a proper leaf is extremal}%
    Let $M$ be compact, $[X]\in H_\Pi^1(M)$ and let $L\subset M$ be a
    leaf which is a smooth embedded submanifold of $M$.  Then any
    $\varphi\in\kms([X],\beta)$ such that
    $\supp(\varphi)\subseteq L$ is extremal.
\end{proposition}
\begin{proof}
    Let $\varphi_1,\varphi_2\in\kms([X],\beta)$ be such that
    $\varphi=\lambda_1\varphi_1+\lambda_2\varphi_2$, where
    $\lambda_1,\lambda_2\in(0,1)$ are such that
    $\lambda_1+\lambda_2=1$.
    For every function $f \ge 0$ with $\supp(f) \subseteq M
    \setminus L$ we have
    \begin{align*}
        0
        =
        \varphi(f)
        =
        \lambda_1\varphi_1(f)+\lambda_2\varphi_2(f)\,.
    \end{align*}
    Positivity of $\varphi_1$ and $\varphi_2$ entails that
    $\varphi_1(f) = \varphi_2(f) = 0$, too. As non-negative functions
    are sufficient to test the support of positive functionals, we
    conclude that
    $\supp(\varphi_1) \cup \supp(\varphi_2) \subseteq \overline{L} = L$ as
    $L$ is proper. Therefore we have $\varphi = \varphi_1 = \varphi_2$
    on account of Remark~\ref{Rmk: KMS for symplectic manifolds}.
\end{proof}

Even if $M$ is not compact the notion of extremal KMS functionals
still make sense.  In fact we shall call
$\varphi\in\KMS([X],\beta)$ \defterm{extremal} whenever
$\varphi=\varphi_1+\varphi_2$ for
$\varphi_1,\varphi_2\in\KMS([X],\beta)$ implies that there
exists $\lambda_1,\lambda_2>0$ such that $\varphi_1=\lambda_1\varphi$,
$\varphi_2=\lambda_2\varphi$.

With this definition at hand one may try to apply a suitable
generalization of the Krein-Milman theorem \cite{Klee-57} to the
convex cone made of KMS functionals.  However, at first glance this
seems not directly applicable and we shall postpone this discussion to
a future investigation.
\begin{remark}
    Notice that the proof of Proposition~\ref{Prop: extremal KMS
      supported on a proper leaf is extremal} generalizes to the
    non-compact case.  In fact, if $L \subset M$ is a leaf which is a
    smooth embedded submanifold of $M$ and $\varphi\in\KMS([X],\beta)$ is
    such that $\supp(\varphi)\subseteq L$, then $\varphi$ is extremal.
\end{remark}
\begin{example}
    For the degenerate case of a trivial Poisson manifold $(M, \Pi)$,
    $\Pi=0$, we find that, for any $X\in\Gamma(TM)$ and
    $\varphi\in\KMS(X,\beta)$ it holds
    \begin{align*}
        \supp(\varphi)\subseteq\{x\in M\,|\, X|_x=0\}\,.
    \end{align*}
    Indeed for any $x\in M$ such that $X|_x\neq 0$ we can find a
    neighbourhood $U$ of $x$ and $f\in C^\infty_{\mathrm{c}}(M)$ with
    $U\subset\supp(f)$ and such that $X(f)|_U=1$.  Applying the KMS
    condition we obtain, for all $g\in C^\infty_{\mathrm{c}}(U)$,
    \begin{align*}
        \beta\varphi(g)=\beta\varphi(X(f)g)=\varphi(\{f,g\})=0\,.
    \end{align*}
    In particular $\varphi\in\KMS(X,\beta)$ is an extremal point
    in $\KMS(X,\beta)$ if and only if
    \begin{align*}
        \varphi\in\{\delta_x\,|\,x\in M\,,\,X|_x=0\}\,,
    \end{align*}
    where $\delta_x(f):=f(x)$ is the Dirac delta measure centred at
    $x\in M$.
\end{example}

%
% Cosymplectic Manifolds
%

\subsection{Cosymplectic Manifolds}\label{Sec: cosymplectic manifolds}

Before starting with the study of the convex set of KMS functionals
for $b$-Poisson manifolds, we shall consider the case of a
cosymplectic manifold $(M,\Pi)$.  The latter is a particular example
of a Poisson manifold whose foliation is regular with codimension $1$.
Moreover, cosymplectic manifolds naturally arise in the study of
$b$-Poisson manifold, see also Remark~\ref{Rmk: results on b-Poisson
  manifolds}.

We shall now recall some basic facts on cosymplectic manifolds, see
\cite{CappellettiMontano-DeNicola-Yudin-13} and references therein for
an extensive review.  A smooth manifold $M$ is called
\defterm{cosymplectic} if $\dim M=2n+1$ and there exists
$\eta\in\Omega^1(M)$ and $\omega\in\Omega^2(M)$ such that
\begin{align}
    \label{Eq: cosymplectic conditions}
    \mathrm{d}\eta=0\,,\qquad
    \mathrm{d}\omega=0\,,\qquad
    \eta\wedge\omega^{\wedge n}>0\,.
\end{align}

Any cosymplectic manifold $(M,\eta,\omega)$ admits a Poisson manifold
structure defined as follows.  One first observes that the map
\begin{align}
    \label{Eq: isomorphism for cosymplectic manifolds}%
    \flat\colon\Gamma(TM)\to\Omega^1(M)\,,
    \qquad
    \flat(X):=\iota_X\omega+\eta(X)\eta\,,
\end{align}
is an isomorphism of vector spaces.  For later convenience we recall
that the vector field $\xi:=\flat^{-1}(\eta)\in\Gamma(TM)$ is called
\defterm{Reeb vector field}.  The Poisson tensor
$\Pi\in\Gamma(\wedge^2TM)$ is then defined by
\begin{align}
    \label{Eq: Poisson tensor for cosymplectic manifolds}%
    \Pi(\alpha,\sigma)
    :=
    \omega(\flat^{-1}(\alpha),\flat^{-1}(\sigma))\,,
\end{align}
for all $\alpha,\sigma\in\Omega^1(M)$.  The following remark
recollects all relevant results for the forthcoming discussion
\cite{CappellettiMontano-DeNicola-Yudin-13}.
\begin{remark}\label{Rmk: results on cosymplectic manifolds}
    Let $(M,\eta,\omega)$ be a cosymplectic manifold.  Then:
    \begin{enumerate}
    \item The anchor map $\sharp\colon\Omega^1(M)\to\Gamma(TM)$
        associated with the Poisson tensor $\Pi$ defined in
        Equation~\eqref{Eq: Poisson tensor for cosymplectic manifolds}
        is given by
        \begin{align*}
            \alpha^\sharp
            =
            -\iota_{\flat^{-1}(\alpha)}\omega
            =
            \flat^{-1}(\alpha)-\alpha(\xi)\xi\,.
        \end{align*}
        Moreover, the symplectic foliation coincides with the
        foliation $x \mapsto \ker\eta_x$ induced by
        $\eta\in\Omega^1(M)$ and the symplectic form $\omega_L$
        associated with any leaf $L$ is given by
        $\omega_L=\iota_L^*\omega$.
    \item For all $k\in\mathbb{N}\cup\{0\}$, let
        $\Omega^k_\eta(M) :=
        \dfrac{\Omega^k(M)}{\eta\wedge\Omega^{k-1}(M)}$.  Then the
        differential descends to a well-defined linear map
        $\mathrm{d}_k\colon\Omega^k_\eta(M)\to\Omega_\eta^{k+1}(M)$ and we
        shall denote by
        $H_\eta^k(M):=\dfrac{\ker\mathrm{d}_k}{\operatorname{im}\mathrm{d}_{k-1}}$
        the associated cohomology group.  Moreover, there is an
        isomorphism of vector spaces
        \begin{align}
            \label{Eq: characterization of cosymplectic Poisson cohomology}%
            H_\Pi^1(M)\simeq H^1_\eta(M)\oplus H^0_\eta(M)\,.
        \end{align}
        In particular any $[X]\in H_\Pi^1(M)$ can be written as
        $[X] = [\alpha^\sharp + f\xi]$ where $[\alpha]\in H_\eta^1(M)$
        and $f\in H_\eta^0(M)$.
    \end{enumerate}
\end{remark}

We now address the problem of describing the convex cone $\KMS([\sigma^\sharp+f\xi],\beta)$ where $f\in H_\eta^0(M)$ and $[\sigma]\in H_\eta^1(M)$.
We begin
with a result which holds for any regular codimension $1$ Poisson
manifold.
\begin{lemma}
	\label{Lem: non-existence criterion for regular codim 1 Poisson manifolds}
	Let $(M,\Pi)$ be a codimension $1$ regular Poisson manifold.
	Let $x\in M$ and $[X]\in H_\Pi^1(M)$ be such that one (hence all) $X\in[X]$ is transverse at $x$ to the symplectic leaf $L_x$.
	Then for all $\varphi\in\KMS([X],\beta)$ we have $x\notin\supp\varphi$.
\end{lemma}
\begin{proof}
	We shall prove that for all $\varphi\in\KMS(X,\beta)$ there exists a neighbourhood $U$ of $x$ such that $\varphi|_{C^\infty_{\mathrm{c}}(U)}=0$.
	For that let $L_x\subseteq M$ be the leaf of $M$ passing through $x$.
	Let $K\subseteq L_x$ be a compact subset of $L_x$ containing $x$ and such that $X_y$ is transverse to $L_x$ for all $y\in K$.
	Let $K_\varepsilon$ be an $\varepsilon$-tubular neighbourhood of $K$ built out of the
	flow $\Psi^{X}$ of $X$ -- that is
	$(-\varepsilon,\varepsilon)\times K\ni(t,x)\mapsto\Psi^{X}_t(x)\in
	K_\varepsilon$ is a diffeomorphism.
	Notice that, since $X$ is transversal to $L_x$ on $K$, $\Psi^X_t(K)$ is a (compact subset of a) leaf for all $t\in(-\varepsilon,\varepsilon)$.
	In particular this entails that the locally defined $1$-form $\mathrm{d}t$ vanishes on all	Hamiltonian vector fields.

	For any $g\in C^\infty_{\mathrm{c}}(K_\varepsilon)$ let	$1_g\in C^\infty_{\mathrm{c}}(K_\varepsilon)$ be such that $1_g|_{\supp(g)}=1$ and set $f(\Psi^{X}_t(x)):=t\,1_g(\Psi^{X}_t(x))$.
	Then $f\in C^\infty_{\mathrm{c}}(M)$ and $X(f)|_{\supp(g)}=1$.
	The classical KMS condition implies that
	\begin{align*}
		\beta\varphi(g)
		=\beta\varphi(gX(f))
		\stackrel{\eqref{Eq: classical KMS condition}}{=}\beta\varphi(\{f,g\})
		=\beta\varphi(X_g(t))
		=0\,,
	\end{align*}
	where we used that $X_g(f)=X_g(t1_g)=X_g(t)=0$.
	The arbitrariness of $g\in C^\infty_{\mathrm{c}}(K_\varepsilon)$ and $K$ ensures that $\varphi=0$.
\end{proof}

Lemma~\ref{Lem: non-existence criterion for regular codim 1 Poisson manifolds} entails that any $\varphi\in\KMS([\sigma^\sharp+f\xi],\beta)$ is supported away from $\supp f$.
Thus we are lead to consider the case $f=0$, for which we have the following theorem.
\begin{theorem}
    \label{Thm: KMS functional space for cosymplectic manifolds}%
    Let $(M,\eta,\omega)$ be a cosymplectic manifold.  Moreover, let
    $[\sigma]\in H_\eta^1(M)$ as well as $\xi:=\flat^{-1}(\eta)$ and
    let $\mu:=\eta\wedge\omega^{\wedge n}$.  Then
    \begin{enumerate}
    \item\label{Item: uniqueness of Reeb-invariant Poisson trace}
    setting $\varphi_\mu:=I_\mu$ we have that
        $\varphi_\mu\in\KMS(M,\Pi)$ as well as
        $\varphi_\mu\circ\xi=0$.  In particular $M$ is unimodular.
        Moreover any $\varphi\in\KMS(M,\Pi)$ which is invariant under
        $\xi$ is a multiple of $\varphi_\mu$.
    \item\label{Item: KMS for compact cosymplectic manifold}
    if $M$ is compact and there exists a proper leaf
        $L_0 \hookrightarrow M$ then $\varphi\in\kms_0([\sigma^\sharp],\beta)$ implies that
        $\supp\varphi\subseteq L$ where $L$ is a leaf such that
        $\iota_L^*\sigma$ is exact.  In particular
        \begin{align*}
            \kms_0([\sigma^\sharp],\beta)
            \simeq
            \left\{
                I_{\iota_L^*\omega^{\wedge n}}
                \; \Big| \;
                L\textrm{ leaf }\,,\,\iota_L^*\sigma\textrm{ is exact}
            \right\}\,.
        \end{align*}
        where $\iota_L^*\omega^{\wedge n}$ is the Liouville form
        associated with the leaf $L$.  In particular the support of
        $\varphi\in\kms([\sigma^\sharp],\beta)$ is contained in the union of
        the leaves $L$ such that $\iota_L^*\sigma$ is exact.
    \end{enumerate}
\end{theorem}
\begin{proof}
	\noindent
    \begin{description}
    \item[\ref{Item: uniqueness of Reeb-invariant Poisson trace}] We
        first observe that, for all $f\in C^\infty(M)$,
        \begin{align*}
            \mathcal{L}_{X_f}(\mu)
            =\mathcal{L}_{X_f}(\eta)\wedge\omega^{\wedge n}
            -\eta\wedge\mathcal{L}_{X_f}(\omega^{\wedge n})
            =\mathrm{d}\iota_{X_f}(\eta)\wedge\omega^{\wedge n}
            -\eta\wedge\mathrm{d}\iota_{X_f}\omega^{\wedge n}
            =0\,,
        \end{align*}
        where we used $\eta(X_f)=0$ as well as
        $\iota_{X_f}\omega=\mathrm{d}f$.  This shows that $I_\mu$
        induces a Poisson trace -- in fact, it also shows that
        $[Y_\Pi]=[0]$, \textit{i.e.} $M$ is unimodular.  Moreover we
        have
        \begin{align*}
            \mathcal{L}_\xi(\mu)
            =\mathcal{L}_\xi(\eta)\wedge\omega^{\wedge n}
            +\eta\wedge\mathcal{L}_\xi(\omega^{\wedge n})
            =0\,.
        \end{align*}
        which entails that $\varphi_\mu$ is $\xi$-invariant.

        Conversely, let $\varphi\in\KMS(M,\Pi)$ be $\xi$-invariant.
        We consider the symplectification
        \cite{CappellettiMontano-DeNicola-Yudin-13} of
        $(M,\eta,\omega)$, that is, the symplectic manifold
        $(\widehat{M},\widehat{\omega})$ given by
        \begin{align*}
            \widehat{M}=M\times\mathbb{T}\,,
            \qquad
            \widehat{\omega}=\omega+\eta\wedge\mathrm{d}\theta\,,
        \end{align*}
        where $\theta$ is the usual quasi-global angle coordinate on
        the circle.  Let $\widehat{\varphi}$ be the functional
        \begin{align*}
            \widehat{\varphi}\colon
            C^\infty_{\mathrm{c}}(\widehat{M})\to\mathbb{C}\,,
            \qquad
            \widehat{\varphi}(f)
            :=
            \frac{1}{2\pi}\int_0^{2\pi}\varphi(f_\theta)\mathrm{d}\theta\,,
        \end{align*}
        where $f_\theta\in C^\infty_{\mathrm{c}}(M)$ denotes the
        restriction of $f\in C^\infty_{\mathrm{c}}(\widehat{M})$ to
        $M\times\{e^{i\theta}\}$.  A direct inspection leads to
        \begin{align*}
            \widehat{\varphi}(\{f,g\}_{\widehat{M}})
            &=\frac{1}{2\pi}\int_0^{2\pi}\left[
                \cancel{\varphi(\{f_\theta,g_\theta\}_M)}
                +\varphi(\xi(f_\theta)\partial_\theta g_\theta)
                -\varphi(\partial_\theta f_\theta \xi(g_\theta))
            \right]\mathrm{d}\theta
            \\
            &=
            \frac{1}{2\pi}\int_0^{2\pi}
            \frac{\mathrm{d}}{\mathrm{d}\theta}\varphi(\xi(f_\theta)g_\theta)
            \mathrm{d}\theta
            =
            0\,,
        \end{align*}
        where we used the $\xi$-invariance to obtain the equality
        $\varphi(\partial_\theta f_\theta
        \xi(g_\theta))=-\varphi(\partial_\theta\xi(f_\theta)g_\theta)$.
        It follows that
        $\widehat{\varphi}\in\KMS(\widehat{M},\widehat{\Pi})$: on
        account of Remark \ref{Rmk: KMS for symplectic manifolds}
        there exists $c>0$ such that
        \begin{align*}
            \widehat{\varphi}(f)
            =
            \frac{c}{2\pi} \int_0^{2\pi}
            f\mathrm{d}\theta\wedge\eta\wedge\omega^{\wedge n}\,,
            \qquad
            \forall f\in C^\infty_{\mathrm{c}}(\widehat{M})\,.
        \end{align*}
        Finally $\widehat{\varphi}(f)=\varphi(f)$ for all
        $f\in C^\infty_{\mathrm{c}}(M)\subset
        C^\infty_{\mathrm{c}}(\widehat{M})$, therefore
        $\varphi = c\varphi_\mu$.

    \item[\ref{Item: KMS for compact cosymplectic manifold}] As we
        assumed $M$ to be compact, we are reduced to consider the
        convex set $\kms([\sigma^\sharp],\beta)$ of KMS states.  We shall now
        prove that an extremal KMS state
        $\varphi\in\kms_0([\sigma^\sharp],\beta)$ is necessarily supported on
        a leaf $L$.  Indeed, since $M$ admits a proper leaf $L_0$, all
        leaves are proper and diffeomorphic to $L_0$, the
        diffeomorphism being given by the flow $\Phi^\xi$ associated
        with $\xi$, see \cite[Thm. 13]{Guillemin-Miranda-Pires-11}.
        By Proposition~\ref{Prop: extremal KMS supported on a proper leaf is extremal} any KMS state supported on a
        leaf is necessarily extremal.

        Conversely if $\varphi\in\kms([\sigma^\sharp],\beta)$ is extremal it
        has to be supported on a leaf.  Indeed, let $t,s\in\mathbb{R}$
        be such that the leaves $L_t=\Phi^\xi_t(L_0)$ and
        $L_s=\Phi^\xi_s(L_0)$ are such that
        $L_t\cap\supp(\varphi)\neq\emptyset$,
        $L_s\cap\supp(\varphi)\neq\emptyset$. By
        Proposition~\ref{Prop: properties of KMS functionals} we
        necessarily have $L_t\cup L_s\subseteq\supp(\varphi)$.  We
        shall now built $\varphi_1\in\kms([\sigma^\sharp],\beta)$ such that
        $\lambda\varphi_1\leq\varphi$ for $0<\lambda<1$.
        This shows that $\varphi$ is not extremal as
        \begin{align*}
            \varphi=\lambda\varphi_1
            +(1-\lambda)\varphi_2\,,
            \qquad
            \varphi_2:=\frac{1}{1-\lambda}(\varphi-\lambda\varphi_1)\,.
        \end{align*}
        Let $\chi\in C^\infty_{\mathrm{c}}(\mathbb{R})$ be such that
        $0\leq\chi\leq 1$, $\chi(t)=1$ and $s\notin\supp\chi$.  Let
        $h_\chi\in C^\infty(M)$ be defined by
        $h(\Phi^\xi_\tau(x))=\chi(\tau)$ where $x\in L_0$ and
        $\tau\in\mathbb{R}$.  Notice that $h_\chi$ is a well-defined
        function provided $\supp\chi$ is sufficiently small.
    	Moreover, per construction
        $h_\chi\in C^\infty_{\mathrm{c}}(M)\cap H_\Pi^0(M)$.

        We then define
        $\varphi_\chi\colon C^\infty_{\mathrm{c}}(M)\to\mathbb{C}$ by
        $\varphi_\chi(f):=\varphi(h_\chi)^{-1}\varphi(h_\chi f)$.  It
        follows that $\varphi_\chi$ is a state: In fact,
        $\varphi_\chi\in\kms([\sigma^\sharp],\beta)$ as
        \begin{align*}
            \varphi_\chi(\{f,g\})
            =\varphi(\{f,h_\chi g\})
            =\beta\varphi_\chi(g\sigma^\sharp(f))\,.
        \end{align*}
        Finally $\varphi_\chi\leq\lambda\varphi$ for
        $\lambda:=\varphi(h)$.

        Therefore $\varphi\in\kms_0([\sigma^\sharp],\beta)$ is necessarily
        supported on a leaf $L$: By Remark \ref{Rmk: KMS for
          symplectic manifolds} this implies that $\iota_L^*\sigma$ is
        exact and that $\varphi=I_{\mu_\beta}$ where
        $\mu_\beta=e^{-\beta H_L}\iota_L^*\omega^{\wedge n}$, where
        $H_L\in C^\infty(L)$ is such that
        $\mathrm{d}H_L=\iota_L^*\sigma$ while
        $\iota_L^*\omega^{\wedge n}$ is the associated Liouville form.
        This shows that $\kms_0([\sigma^\sharp],\beta)$ can be identified
        with the collection
        \begin{align*}
            \left\{
                I_{\iota_L^*\omega^{\wedge n}}
                \; \Big| \;
                L\textrm{ leaf }\,,\,\iota_L^*\sigma\textrm{ is exact}
            \right\}\,.
        \end{align*}
        The last statement follows from the equality
        $\kms([\sigma^\sharp],\beta)=\overline{\kms_0}([\sigma^\sharp],\beta)$.
    \end{description}
\end{proof}

Theorem~ \ref{Thm: KMS functional space for cosymplectic manifolds}
shows that the convex cone of KMS functionals associated with a
cosymplectic structure is essentially made by Poisson traces with
support constraints.  Remarkably, among all Poisson traces, there
exists a preferred one, which is invariant under the flow of the
vector field $\xi$.
The result in item \ref{Item: KMS for compact cosymplectic manifold} is obtained under the assumption of compactness of $M$ together with the existence of a proper leaf $L$.
The example discussed in Section \ref{Sec: 3-torus regular Poisson manifolds} shows that the last assumption (the existence of a proper leaf) cannot be dropped.

\begin{example}
    \label{Ex: KMS functionals for plane Poisson manifold}%
    We shall discuss the convex cone of KMS functionals for the
    non-compact cosymplectic manifold defined by
    \begin{align*}
        M:=\mathbb{R}^2\times\mathbb{T}\,,\qquad
        \eta:=\mathrm{d}x\,,\qquad
        \omega=\mathrm{d}y\wedge\mathrm{d}\theta\,,
    \end{align*}
    where $\theta\in[0,2\pi)$.  By direct inspection we have
    \begin{align}\label{Eq: plane manifold Poisson cohomology}
        H_\Pi^1(M)\simeq C^\infty(\mathbb{R})^2\,,\qquad
        [X]=[a\partial_x
        +b\partial_y]\,,
    \end{align}
    where $a, b \in C^\infty(\mathbb{R})$ only depend on $x$.  We
    shall now discuss the convex cone $\KMS([a\partial_x+b\partial_y],\beta)$
    for all $a,b\in C^\infty(\mathbb{R})$.  Let
    $\varphi\in\KMS([a\partial_x+b\partial_y],\beta)$.  Whenever $x\in\mathbb{R}$ is such
    that $a(x)\neq 0$ we can find a neighbourhood $U$ of $M$ such that
    $X$ is transversal to the leaves in $U$: Lemma~\ref{Lem:
      non-existence criterion for regular codim 1 Poisson manifolds}
    entails $\supp(\varphi)\cap\supp(a)=\emptyset$.  Henceforth we shall
    assume $a=0$.

    Similarly, let $x\in\mathbb{R}$ be such that $b(x)>0$ and let
    $U=(x-\varepsilon,x+\varepsilon)\times\mathbb{R}\times\mathbb{T}$
    be a neighbourhood of $M$ such that $b>0$.  A direct computation
    shows that, for all $g\in C^\infty_{\mathrm{c}}(U)$,
    \begin{align*}
        \beta\varphi(b g)
        =\beta\varphi(gX(y))
        =\varphi(\{y,g\})
        =-\varphi(\partial_\theta g)\,.
    \end{align*}
    Since $b$ and $U$ are invariant under the flow $\Psi$
    associated with $\partial_\theta$ we find
    \begin{align*}
        \varphi(\Psi_t^*g)
        =
        \varphi(e^{-\beta b t}g)\,,
    \end{align*}
    which is nothing but Equation~\eqref{Eq: global classical KMS condition}.  Since $e^{-\beta bt}\Psi^*_{-t}g\to 0$ as
    $t\to+\infty$ in $C^\infty_{\mathrm{c}}(M)$ we find
    \begin{align*}
        \varphi(g)
        =\lim_{t\to+\infty}
        \varphi(e^{-\beta b t}\Psi^*_{-t}g)
        =0\,,
    \end{align*}
    since the measure $\varphi$ is continuous for locally uniform
    convergence.  This arguments shows that
    $\varphi\in\KMS([a\partial_x+b\partial_y],\beta)$ is supported outside
    $\supp(a)\cup\supp(b)$.  If follows that $\KMS([a\partial_x+b\partial_y],\beta)$ embeds
    in the convex cone $\KMS(M,\Pi)$ made by $\Pi$-traces.

    We shall now consider $\varphi\in\KMS(M,\Pi)$.  We first observe
    that $\varphi$ is invariant under the flow of $\partial_y$.
    Indeed, let $\chi_1$, $\chi_2$ be a partition of unity
    subordinated to the cover
    $[0,2\pi)=[0,\frac{3}{2}\pi)\cup(\pi,2\pi)$.
    Then for all $f\in C^\infty_{\mathrm{c}}(M)$ we have
    \begin{align*}
        \varphi(\partial_yf)
        =\varphi(\partial_y(\chi_1 f))
        +\varphi(\partial_y(\chi_2 f))
        =\varphi(\{\theta 1_f,\chi_1f\})
        +\varphi(\{\theta 1_f,\chi_2f\})
        =0\,,
    \end{align*}
    where $1_f\in C^\infty_{\mathrm{c}}(M)$ is such that
    $1_f|_{\supp(f)}=1$.

    Let $\mu_\varphi\in\mathcal{M}_+(M)$ be the positive measure on
    $M$ associated with the natural extension of $\varphi$ on $C(M)$,
    \textit{cf.}  Remark~\ref{Rmk: on definition of KMS functional}.
    Let $\chi\in C^\infty_{\mathrm{c}}(\mathbb{R})$ be such that
    $\chi\geq0$ and $\int\chi(y)\mathrm{d}y=1$ and for all
    $I\subseteq\mathbb{R}$ we define
    \begin{align*}
        \nu_\varphi(I)
        :=\int_M \varrho_I(x)\chi(y)\mathrm{d}\mu_{\varphi}(x,y,\theta)
        =\varphi(\varrho_I\chi)\,,
    \end{align*}
    where $\varrho_I$ is the characteristic function over $I$.  It
    follows that $\nu_\varphi$ defines a positive Borel measure on
    $\mathbb{R}$.  Notice that $\nu_\varphi$ does not depend on
    $\chi$: Indeed, if
    $\chi,\hat{\chi}\in C^\infty_{\mathrm{c}}(\mathbb{R})$ are such
    that $\int[\chi-\hat{\chi}](y)\mathrm{d}y=0$ then
    $\chi-\hat{\chi}=g_y$ for $g\in C^\infty_{\mathrm{c}}(\mathbb{R})$
    and
    \begin{align*}
        \varphi(\varrho_I\chi)
        -\varphi(\varrho_I\hat{\chi})
        =\varphi(\partial_y(\varrho_Ig))=0\,.
    \end{align*}
    Finally, for all $f\in C^\infty_{\mathrm{c}}(\mathbb{R})$ we have
    \begin{align*}
        f
        &=f-\int_0^{2\pi}f\mathrm{d}\theta
        +\int_0^{2\pi}f\mathrm{d}\theta
        -\chi\int_0^{2\pi}\int_{\mathbb{R}}f(\cdot,\theta,y)\mathrm{d}y\mathrm{d}\theta
        +\chi\int_0^{2\pi}\int_{\mathbb{R}}f(\cdot,\theta,y)\mathrm{d}y\mathrm{d}\theta
        \\
        &=
        \partial_{\theta}g
        +\partial_yh
        +\chi\int_0^{2\pi}\int_{\mathbb{R}}f(\cdot,\theta,y)\mathrm{d}y\mathrm{d}\theta\,,
    \end{align*}
    where $g,h\in C^\infty_{\mathrm{c}}(M)$.  The $\Pi$-trace property
    of $\varphi$ implies that
    \begin{align*}
        \varphi(f)
        =\int_Mf(x,\theta,y)\mathrm{d}\nu_\varphi(x)\mathrm{d}\theta\mathrm{d}y\,.
    \end{align*}
    The correspondence $\varphi\to\nu_\varphi$ is thus 1-1, therefore,
    $\KMS(M,\Pi)\simeq\mathcal{M}_+(\mathbb{R})$, where
    $\mathcal{M}_+(\mathbb{R})$ denotes the convex cone of positive
    Borel measure over $\mathbb{R}$.

    Summing up, the KMS convex cones for $(M,\Pi)$ satisfy
    \begin{subequations}
        \begin{align}
            \label{eq:KMSPi}
            \KMS(M,\Pi)
            &\simeq\mathcal{M}_+(\mathbb{R})\,,\\
            \label{eq:KMSPiabbeta}
            \KMS([a\partial_x+b\partial_y],\beta)
            &\simeq
            \{
            \nu\in\mathcal{M}_+(\mathbb{R})
            \,|\,
            \supp(\nu)\cap(\supp(a)\cup\supp(b))=\emptyset\}\,.
        \end{align}
    \end{subequations}
\end{example}

%
% KMS Convex Cone for $b$-Poisson Manifolds
%

\section{KMS Convex Cone for $b$-Poisson Manifolds}
\label{Sec: KMS convex cone for b-Poisson manifolds}

In this section we discuss the convex cone of KMS functionals for the
case of a $b$-Poisson manifold \cite{Melrose-93, Nest-Tsygan-96}.  To
this avail we recall the basic definitions and properties of
$b$-Poisson manifolds whose proofs can be found in the literature
\cite{Braddell-Kiesenhofer-Miranda-18,
  Braddell-Kiesenhofer-Miranda-20, Cavalcanti-17,
  Guillemin-Miranda-Pires-14, Guillemin-Miranda-Pires-Scott-17,
  Guillemin-Miranda-Weitsman-19, Marcut-Torres-14a, Marcut-Torres-14b,
  Miranda-Oms-Salas-20}.

%
% $b$-Poisson Manifolds
%

\subsection{$b$-Poisson Manifolds}
\label{Sec: b-Poisson manifolds}

To fix our notation we need to recall some basic facts about
$b$-Poisson manifolds in this preliminary section
\cite{Guillemin-Miranda-Pires-14,Marcut-Torres-14b}.  A Poisson
manifold $(M,\Pi)$ is called \defterm{$b$-Poisson} if $\dim M=2n$ and
$\Pi^{\wedge n}$ is transversal to the zero section in
$\Gamma(\wedge^{2n}TM)$.  In what follows we shall denote by
$Z:=\{x\in M\;|\; \Pi^{\wedge n}|_x=0\}$ the zero locus of
$\Pi^{\wedge n}$.  A \defterm{$Z$-defining function} is a smooth
function $\zeta\in C^\infty(M)$ such that
$Z=\{x\in M\;|\;\zeta(x)=0\}$ and $\mathrm{d}\zeta|_Z\neq 0$.

\begin{remark}\label{Rmk: global Z-definition function}%
    Since we are assuming that $M$ is orientable we can always find a
    $Z$-defining function $\zeta\in C^\infty(M)$.  Indeed, for a given
    volume form $\mu\in\Omega^{2n}(M)$ we can set
    $\zeta_\mu:=\iota_{\Pi^{\wedge n}}\mu$.  Moreover, given a
    $Z$-defining function $\zeta_0$, all other $Z$-defining functions
    are of the form $h\zeta_0$ for a non-vanishing $h\in C^\infty(M)$.
    Notice in particular that $\zeta_\mu Y_\mu=-X_{\zeta_\mu}$.
\end{remark}

As $b$-Poisson manifolds are best understood in the context of
$b$-geometry we recall a few basic fact of $b$-geometry introduced in
\cite{Melrose-93, Nest-Tsygan-96}.  The \defterm{$b$-tangent bundle}
relative to the pair $(M, Z)$ is the unique (up to vector bundle
isomorphisms) vector bundle $^bTM$ whose sections are isomorphic to
the subspace $\Gamma(TM)_Z$ of vector fields in $M$ which are tangent
to $Z$, that is,
\begin{align*}
	\Gamma(^bTM)
    \simeq\Gamma(TM)_Z
    :=\{X\in\Gamma(TM)\;|\;X|_Z\in\Gamma(TZ)\}\,.
\end{align*}
In what follows we shall identify $\Gamma(^bTM)$ and $\Gamma(TM)_Z$.

Existence and uniqueness of $^bTM$ is ensured by the fact that, on a
coordinate neighbourhood $U$, we have the local frame
$\{z\partial_z,\partial_{x^1},\ldots\partial_{x^{2n-1}}\}$ spanning
the sections $\Gamma(TM)_Z$, where $z$ is a defining function for
$Z\cap U$.

By definition it also follows that $^bT_xM= T_xM$ for all
$x\in M\setminus Z$.  For points $x\in Z$ we have
$^bT_xM=T_xZ\oplus\operatorname{sp}w_x$, where $w_x$ is the canonical
non-vanishing section of the line bundle $^bTM|_Z\to TZ$ \textit{cf.}
\cite[Prop. 4]{Guillemin-Miranda-Pires-14}.  Roughly speaking, the
latter bundle is induced by the restriction map
$\Gamma(^bTM)\to\Gamma(TZ)$, while $w$ is defined by considering a
$Z$-defining function $\zeta$ and setting $w:=\zeta Z_\zeta$ where
$Z_\zeta\in\Gamma(TM)$ is any vector field such that $Z_\zeta(\zeta)=1$ -- the definition does not depend on the chosen
$\zeta$.

By duality one can introduce the \defterm{$b$-cotangent bundle} $^bT^*M$.  As
for $^bTM$, we have that
\begin{align}
    \label{Eq: decomposition of b-cotangent bundle}
    ^bT^*_xM
    =
    \begin{cases}
        T^*_xM
        & \textrm{for } x\in M\setminus Z
        \\
        T^*_xZ\oplus\operatorname{sp}\nu_w
        & \textrm{for } x\in Z.
    \end{cases}
\end{align}
Here $\nu_w\in\Gamma(^bT^*M)$ is defined by the requirement that
$\nu_w(w) = 1$.  In what follows we shall adopt the notation
$\nu_w=\dfrac{\mathrm{d}\zeta}{\zeta}$ once a $Z$-defining function
$\zeta$ has been fixed.

Let $^b\Omega^k(M):=\Gamma(\wedge^k\,^bT^*M)$ be the space of
$b$-forms of degree $k$.  Notice that any $\eta\in\, ^b\Omega^k(M)$
can be decomposed as
\begin{align}
    \label{Eq: decomposition of b-forms}
    \eta=\alpha\wedge\frac{\mathrm{d}\zeta}{\zeta}+\sigma\,,
\end{align}
where $\alpha\in\Omega^{k-1}(M)$ and $\sigma\in\Omega^k(M)$ while
$\zeta$ is a $Z$-defining function.  In Equation~\eqref{Eq:
  decomposition of b-forms} we interpreted elements in
$\Omega^\bullet(M)$ as $b$-forms $^b\Omega^\bullet(M)$ according to
the following convention
\begin{align*}
    \alpha_x
    =
    \begin{cases}
        \alpha_x
        & \textrm{for } x\in M\setminus Z
        \\
        \iota_Z^*\alpha_x
        & \textrm{for } x\in Z\,,
    \end{cases}
\end{align*}
being $\iota_Z\colon Z\hookrightarrow M$.  This convention is
consistent with Equation~\eqref{Eq: decomposition of b-cotangent
  bundle}.
\begin{remark}
    \label{Rmk: on decomposition of b-forms}%
    Notice that the forms $\alpha$ and $\sigma$ appearing in
    Equation~\eqref{Eq: decomposition of b-forms} are not unique:
    $\alpha \mapsto \alpha+h\mathrm{d}\zeta$ does not affect
    Equation~\eqref{Eq: decomposition of b-forms} as well as the
    simultaneous replacement $\alpha \mapsto \alpha+\zeta\alpha'$,
    $\sigma \mapsto \sigma-\alpha'\wedge\mathrm{d}\zeta$, where
    $h\in C^\infty(M)$ and $\alpha'\in\Omega^{k-1}(M)$.  Finally if
    $\zeta_\lambda=\lambda\zeta$ is another $Z$-defining form -- where
    $\lambda\in C^\infty(M)$ is no-where vanishing -- then the form
    $\sigma$ in Equation~\eqref{Eq: decomposition of b-forms} gets
    shifted by $\mathrm{d}\log|\lambda|$.  Nevertheless, on account of
    Equation~\eqref{Eq: decomposition of b-cotangent bundle}, the
    decomposition appearing in Equation~\eqref{Eq: decomposition of
      b-forms} is unique on $Z$ -- \textit{cf.}
    \cite[Prop. 5]{Guillemin-Miranda-Pires-14}.
\end{remark}

On account of Equation~\eqref{Eq: decomposition of b-forms} we can
extend the differential $\mathrm{d}$ to $^b\Omega^\bullet(M)$.  Recall
from \cite{Melrose-93} that the \defterm{$b$-de Rham differential}
$^b\mathrm{d}\colon\,^b\Omega^{\bullet}(M)\to\,^b\Omega^{\bullet+1}(M)$
is defined by
\begin{align}
    \label{Eq: b-de Rham differential}
    ^b\mathrm{d}\eta
    :=\mathrm{d}\alpha\wedge\frac{\mathrm{d}\zeta}{\zeta}
    +\mathrm{d}\sigma\,,
\end{align}
for all $\eta\in\,^b\Omega^k(M)$.  The \defterm{$b$-de Rham complex} is the
chain complex
\begin{align*}
    0
    \stackrel{^b\mathrm{d}}{\longrightarrow}
    \,^b\Omega^1(M)
    \stackrel{^b\mathrm{d}}{\longrightarrow}
    \cdots
    \stackrel{^b\mathrm{d}}{\longrightarrow}
    \,^b\Omega^{2n}(M)
    \stackrel{^b\mathrm{d}}{\longrightarrow}
    0\,.
\end{align*}
The associated cohomology groups $^bH^\bullet(M)$ are called
\defterm{$b$-de Rham cohomology group}.

\begin{remark}[\cite{Melrose-93}]
    Notice that the definition of $^b\mathrm{d}$ does not depend on
    the chosen $Z$-defining function $\zeta$ -- \textit{cf.} Remark
    \ref{Rmk: on decomposition of b-forms}.  Moreover, since the forms
    $\alpha$ and $\sigma$ appearing in Equation~\eqref{Eq:
      decomposition of b-forms} are not unique, the identity
    $^b\mathrm{d}\eta=0$ does not imply $\mathrm{d}\alpha=0$ and
    $\mathrm{d}\sigma=0$ in general.  However, one can exploit the
    non-uniqueness of $\alpha,\sigma$ to prove that, for a suitable
    choice of a $Z$-defining function $\zeta$, the forms
    $\alpha,\sigma$ can be chosen to be closed.  This leads to the
    following result \cite[Prop. 2.49]{Melrose-93}: for the
    $b$-cohomology groups one has
    \begin{align}
        \label{Eq: characterization of b-de Rham cohomology}
        ^bH^\bullet(M)
        \simeq
        H^\bullet(M)\oplus
        H^{\bullet-1}(Z)\,,
    \end{align}
    where $H^\bullet(M)$, $H^\bullet(Z)$ denote the standard de-Rham
    cohomology groups of $M$ and $Z$ respectively.
\end{remark}

We now discuss the relation between $b$-Poisson manifolds and
$b$-geometry -- \textit{cf.} \cite{Guillemin-Miranda-Pires-14} for an
extensive treatment.  If $(M,\Pi)$ is a $b$-Poisson manifold, then
$\Pi\in\Gamma(\wedge^2\,^bTM)$.  Therefore, though $\Pi$ is degenerate
at $Z$ as an element of $\Gamma(\wedge^2 TM)$, it is non-degenerate as
a bivector field in the $b$-tangent bundle -- this is actually the
main reason for considering the $b$-geometrical setting.  It follows
that there exists a unique non-degenerate closed $b$-form
$\omega \in { }^b\Omega^2(M)$ such that $\iota_\Pi\omega=1$.
Manifolds $M$ equipped with such a $b$-form $\omega$ are usually
called \defterm{$b$-symplectic manifolds} -- \textit{cf.}
\cite[Def. 8]{Guillemin-Miranda-Pires-14}. We collect some of the
properties needed later
\cite{Guillemin-Miranda-Pires-14,Marcut-Torres-14b}:
\begin{remark}
    \label{Rmk: results on b-Poisson manifolds}%
    Let $(M,\Pi)$ be a $b$-Poisson manifold.  Let
    $\omega\in { }^b\Omega^2(M)$ be the unique non-degenerate closed
    $b$-form such that $\iota_\Pi\omega=1$.
    \begin{enumerate}
    \item For all $x\in M$ there exists a coordinate neighbourhood $U$
        of $x$ such that
        \begin{align}
            \Pi\big|_U
            =z\partial_z\wedge\partial_t
            +\sum_{\ell=1}^{n-1}\partial_{x^\ell}\wedge\partial_{y^\ell}\,,
            \qquad
            \omega\big|_U
            =\frac{\mathrm{d}z}{z}\wedge\mathrm{d}t
            +\sum_{\ell=1}^{n-1}\mathrm{d}x^\ell\wedge\mathrm{d}y^\ell\,,
        \end{align}
        where $z$ is a defining function for $Z\cap U$.
    \item Let
        $\omega=\alpha\wedge\frac{\mathrm{d}\zeta}{\zeta}+\sigma$ be
        the decomposition of $\omega$ according to Equation~\eqref{Eq:
          decomposition of b-forms}.  Then
        $\alpha_Z:=\iota_Z^*\alpha\in\Omega^1(Z)$ is uniquely defined
        and closed.  Moreover it is possible to choose $\alpha$ and
        $\sigma$ so that $(Z,\alpha_Z,\sigma_Z)$ is a cosymplectic
        manifold, where $\sigma_Z := \iota_Z^*\sigma\in\Omega^2(Z)$.
        In fact the inclusion $(Z, \iota_Z)\hookrightarrow(M,\Pi)$ is
        a Poisson morphism.  Finally, any representative $X\in[Y_\Pi]$
        of the modular class is tangent to $Z$ and $X|_Z$ is the Reeb
        vector field associated with $(Z,\alpha_Z,\sigma_Z)$ --
        \textit{cf.}  Equation~\eqref{Eq: isomorphism for cosymplectic
          manifolds}.
    \item For the Poisson cohomology $H_\Pi^\bullet(M)$ of $(M,\Pi)$
        one has
        \begin{align}
            \label{Eq: characterization of b-Poisson cohomology}
            H_\Pi^\bullet(M)
            \simeq\; ^bH^\bullet(M)
            \simeq
            H^\bullet(M)
            \oplus H^{\bullet-1}(Z)\,.
        \end{align}
        In particular an element $[X]\in H_\Pi^1(M)$ is always of the
        form $X=\Theta^\sharp$ for a closed $b$-form
        $\Theta\in {}^bH^1(M)$.  Under this identification the modular
        class $[Y_\Pi]$ is associated with the $b$-de Rham cohomology
        class $\left[\dfrac{\mathrm{d}\zeta}{\zeta}\right]$ for a given choice of
        a $Z$-defining function $\zeta$.
    \end{enumerate}
\end{remark}

%
% An Explicit Example
%

\subsection{An Explicit Example}
\label{Sec: toy example b-cohomology class}

To make the forthcoming discussion more concrete, we will consider the
following toy example.  Let $(M,\Pi)$ be the $b$-Poisson manifold
defined by
\begin{align}\label{Ex: toy b-Poisson manifold}
    M:=\mathbb{T}\times\mathbb{R}\,,
    \qquad
    \Pi:=\sin(\theta)\partial_\theta\wedge\partial_y\,,
\end{align}
where $\theta$ denotes the quasi-global angular coordinate on
$\mathbb{T}$ as before.  Notice that
\begin{align*}
	Z
	=(\{0\}\times\mathbb{R})
	\cup(\{\pi\}\times\mathbb{R})
	=:Z_0\cup Z_\pi\,,
\end{align*}
so that this is actually an example of a $b$-Poisson manifold with
non-connected $Z$.

It is instructive to compute $H_\Pi^1(M)$ without exploiting
Remark~\ref{Rmk: results on b-Poisson manifolds}.  A vector field
$X=a\partial_\theta+b\partial_y$ where $a,b\in C^\infty(M)$ is Poisson
if and only if
\begin{align*}
    \cos(\theta)a-\sin(\theta)a_\theta+\sin(\theta)b_y=0\,,
\end{align*}
where the sub-indexes $_\theta$ and $_y$ are short notations for
$\partial_\theta,\partial_y$.  Evaluating the previous equation for
$\theta=0$ and $\theta=\pi$ we find that $a(0,y)=a(\pi,y)=0$,
therefore $a(\theta,y)=\sin(\theta)\alpha(\theta,y)$ for
$\alpha\in C^\infty(M)$.  We thus obtain
\begin{align*}
    -\sin(\theta)^2\alpha_\theta+\sin(\theta)b_y=0\,.
\end{align*}
Differentiation in $\theta$ and evaluation at $\theta=0$ and
$\theta=\pi$ leads to $b_y(0,y)=b_y(\pi,y)=0$, that is, $b_0:=b(0,y)$
and $b_\pi:=b(\pi,y)$ are constant.  If $b_0=b_\pi=0$ we have
$b(\theta,y)=\sin(\theta)\beta(\theta,y)$ for $\beta\in C^\infty(M)$,
moreover,
\begin{align*}
    -\alpha_\theta+\beta_y=0\,,
\end{align*}
which implies that $\Theta:=\alpha\mathrm{d}y+\beta\mathrm{d}\theta$ is
closed.  It follows that $[X]=[\Theta^\sharp]$ for
$[\Theta]\in\operatorname{sp}[\mathrm{d}\theta]$.  For the general
case let
\begin{align}
    \label{Eq: psi0, psipi definition}
    \psi_0(\theta):=\frac{1}{2}(1+\cos(\theta))
    =\cos(\theta/2)^2\,,
    \qquad
    \psi_\pi(\theta):=\frac{1}{2}(1-\cos(\theta))
    =\sin(\theta/2)^2\,.
\end{align}
Notice that $\psi_0\partial_y$ and $\psi_\pi\partial_y$ are Poisson
vector fields.  Moreover,
$[\psi_0\partial_y]\neq[0]\neq[\psi_\pi\partial_y]$: Indeed
$\psi_0\partial_y=X_f$ would imply
\begin{align*}
    f_y=0\,,
    \qquad
    \psi_0=-\sin(\theta)f_\theta\,.
\end{align*}
However the left-hand side of the last equation is $1$ for $\theta=0$
while the right-hand side gives $0$.  A similar argument applies for
$\psi_\pi\partial_y$.  Since $\psi_0(0)=\psi_\pi(\pi)=1$ as well as
$\psi_0(\pi)=\psi_\pi(0)=0$, we have
\begin{align*}
    [X]
    &=[\sin(\theta)\alpha\partial_\theta+(b-b_0\psi_0-b_\pi\psi_\pi)\partial_y]
    +[b_0\psi_0\partial_y]
    +[b_\pi\psi_\pi\partial_y]
    \\&=[c\mathrm{d}\theta^\sharp]
    +[b_0\psi_0\partial_y]
    +[b_\pi\psi_\pi\partial_y]\,,
\end{align*}
for $c\in\mathbb{R}$.  It then follows that
$H_\Pi^1(M)\ni X\to(b_0,b_\pi,c)\in\mathbb{R}^3$ realizes the general
isomorphism in Equation~\eqref{Eq: characterization of b-Poisson
  cohomology}.

%
% KMS Convex Cone
%

\subsection{KMS Convex Cone}
\label{Sec: KMS convex cone}

We now investigate the convex cone of KMS functionals
$\KMS([X],\beta)$ for all $[X]\in H_\Pi^1(M)$ and $\beta> 0$.
\begin{remark}
    \label{Rmk: Z-connectedness assumption}%
    For the easy of the presentation we shall assume that both $M,Z$
    are connected.  The non-connected case would not spoil the results
    of our analysis -- \textit{cf.} Example~\ref{Ex: toy example of
      KMS for vartheta vanishing on Mz} -- however we point out the
    main differences in a series of Remarks~\ref{Rmk: Z-connectedness
      assumption for Lemma}, \ref{Rmk: Z-connectedness assumption for
      Theorem}, and \ref{Rmk: Z-connectedness assumption for
      extremality}.  Notice that if $M,Z$ are connected then
    $M_Z:=M\setminus Z$ has only two connected components
    $M_Z=M_Z^+\cup M_Z^-$.  Once a $Z$-defining function has been
    fixed -- \textit{cf.} Remark \ref{Rmk: global Z-definition
      function} -- these can be realized as
    $M_Z^\pm=\{\pm\zeta>0\}$.
\end{remark}

According to Remark~\ref{Rmk: results on b-Poisson manifolds} we
parametrize $[X]=[\Theta^\sharp]\in H_\Pi^1(M)$ with
$[\Theta]\in\,^bH^1(M)$.  Moreover, we shall consider a defining
function $\zeta$ such that
\begin{align}
    \label{Eq: decomposition of closed b-form of degree 1}
    \Theta=-a\frac{\mathrm{d}\zeta}{\zeta}+\vartheta\,,
    \qquad
    a\in\mathbb{R}\,,
    [\vartheta]\in H^1(M)\,.
\end{align}
(The minus sign is inserted so that
$a\dfrac{\mathrm{d}\zeta}{\zeta}^\sharp\in[aY_\Pi]$.)  Notice that any
other defining function would not change the $b$-de Rham class of
$[\Theta]$.
In what follows $\mu$ will denote the volume form associated with the $Z$-defining function $\zeta$.

As a first step we will consider the restriction of
$\varphi\in\KMS([\Theta^\sharp],\beta)$ to either the symplectic leaf
$M_Z:=M\setminus Z$ or the cosymplectic manifold $Z$.  This provides a
strong constraint on the chosen $\Theta$.

\begin{lemma}
    \label{Lem: non-existence of KMS functional supported on Z with dynamics containing modular class}%
    Let $\varphi\colon C^\infty_{\mathrm{c}}(M)\to\mathbb{C}$ be a
    (non-necessarily positive) distribution of order $0$ satisfying
    Equation~\eqref{Eq: classical KMS condition} for
    $X=\Theta^\sharp$, where $\Theta\in\,^b\Omega^1(M)$ decomposes as
    in Equation~\eqref{Eq: decomposition of closed b-form of degree
      1}.  If $\supp(\varphi)\subseteq Z$ and $a\neq 0$, then
    $\varphi=0$.
\end{lemma}
\begin{proof}
    Since $\varphi$ is of order $0$ and $\supp(\varphi)\subseteq Z$,
    it follows that there exists a (non necessarily positive)
    distribution $\psi\colon C^\infty_{\mathrm{c}}(Z)\to\mathbb{C}$ of
    order $0$ such that $\varphi(f)=\psi(f|_Z)$.  (Actually $\psi$ is
    defined by $\psi(g)=\varphi(\widehat{g})$ where
    $\widehat{g}\in C^\infty_{\mathrm{c}}(M)$ is any function such
    that $\widehat{g}|_Z=g$.)  Moreover, $\psi$ satisfies the KMS
    condition \eqref{Eq: classical KMS condition}
    \begin{align*}
	\psi(\{f,g\}_Z)=\beta\psi(g\Theta^\sharp|_Z(f))\,,
        \qquad
        \forall f,g\in C^\infty_{\mathrm{c}}(Z)\,,
    \end{align*}
    where $\{\,,\,\}_Z$ is the Poisson structure induced on $Z$.
    Notice that $\Theta^\sharp|_Z\in\Gamma(TZ)$ because
    $\Theta^\sharp\in\Gamma(^bTM)$. Moreover, $\Theta^\sharp|_Z$ is transverse to
    the symplectic foliation of $Z$, because $a\neq 0$.  Thus
    Lemma~\ref{Lem: non-existence criterion for regular codim 1
      Poisson manifolds} applies and we find $\psi=0$.
\end{proof}

Lemma~\ref{Lem: non-existence of KMS functional supported on Z with
  dynamics containing modular class} shows that there cannot be KMS
functionals supported on $Z$ for dynamics $\Theta^\sharp$ where
$a\neq 0$ in Equation~\eqref{Eq: decomposition of closed b-form of
  degree 1}.  Similarly the following lemma implies that there cannot
be KMS functionals supported in $M_Z$ whenever
$[\vartheta|_{M_Z}]\neq [0]$.
\begin{lemma}
    \label{Lem: constraint on the form of KMS on M minus Z}%
    Let $[\Theta]\in\,^bH^1(M)$ and let $a\in\mathbb{R}$ and
    $[\vartheta]\in H^1(M)$ be as in Equation~\eqref{Eq: decomposition of closed b-form of degree 1}.
  	For $\varphi\in\KMS([\Theta^\sharp],\beta)$, let $\varphi_{M_Z}$ be the
    restriction of $\varphi$ to $C^\infty_{\mathrm{c}}(M_Z)$.  Then:
    \begin{enumerate}
    \item\label{Item: varphiMZ vanishes if varthetaMZ is not trivial}
        If $[\vartheta|_{M_Z}]\neq [0]$, then $\varphi_{M_Z}=0$.
    \item\label{Item: varphiMZ formula if varthetaMZ is trivial}
    	If $\vartheta|_{M_Z}=\mathrm{d}H$ for $H\in C^\infty(M_Z)$,
        then there exist two positive constants $c_\pm\in[0,+\infty)$
        such that
        \begin{align}
            \label{Eq: constrained form of KMS on M minus Z}
            \varphi_{M_Z}(f)
            =
            \sum_\pm c_\pm\int_{M_Z^\pm} fe^{-\beta H_\pm} |\zeta|^{a\beta-1}
            \mu\,,
            \qquad
            \forall f\in C^\infty_{\mathrm{c}}(M_Z)\,.
        \end{align}
        where $M_Z^\pm$ are the two connected components of $M_Z$ --
        \textit{cf.} Remark \ref{Rmk: Z-connectedness assumption} --
        and $H_\pm:=H|_{M_Z^\pm}$ while $\mu\in\Omega^{2n}(M)$ is the
        volume form associated to the chosen defining function
        $\zeta$.
    \end{enumerate}
\end{lemma}
\begin{proof}
    For all $\varphi\in\KMS(\Theta^\sharp,\beta)$ we have that
    $\varphi_{M_Z}$ is a $(\Theta^\sharp|_{M_Z},\beta)$-KMS on $M_Z$.
    Since $(M_Z,\Pi|_{M_Z})$ is a symplectic manifold, Remark
    \ref{Rmk: KMS for symplectic manifolds} entails that, if
    $\varphi_{M_Z}\neq 0$, then $[\Theta|_{M_Z}]=[0]$.  As
    $a\dfrac{\mathrm{d}\zeta}{\zeta}|_{M_Z}=a\mathrm{d}\log|\zeta|$ we
    find $[0]=[\Theta|_{M_Z}]=[\vartheta|_{M_Z}]$.  This proves item
    \ref{Item: varphiMZ vanishes if varthetaMZ is not trivial}.
    Equation~\eqref{Eq: constrained form of KMS on M minus Z} is a
    direct application of
    Remark \ref{Rmk: KMS for symplectic manifolds} -- \textit{cf.}
    Equation \eqref{Eq: Gibbs functionals}.
\end{proof}
\begin{remark}
    \label{Rmk: Z-connectedness assumption for Lemma}%
    It is worth commenting what would happen to Lemma~\ref{Lem:
      constraint on the form of KMS on M minus Z} by allowing $Z$ to
    have many connected components.  Let $\{Z_k\}_k$ be the (possibly
    countably infinite) collection of connected components of $Z$ and
    let $M_{Z,c}$ be the collection of all connected components of
    $M_Z$.  For each $k$, let $Z_k^{\varepsilon_k}$ be a
    $\varepsilon_k$-tubular neighbourhood of $Z_k$ such that
    $Z_k^{\varepsilon_k}\cap Z_j^{\varepsilon_j}=\emptyset$ for at
    most finitely many $k,j$.  Moreover, since $Z_k$ has codimension
    1, we may choose $\varepsilon_k$ small enough so that
    $Z_k^{\varepsilon_k}\setminus Z$ is made by two connected
    components, say $Z_k^{\varepsilon_k,+}$ and
    $Z_k^{\varepsilon_k,-}$.  We can then find a locally defined
    $Z_k$-defining function $\zeta_k$ ---\textit{cf.} Section
    \ref{Sec: b-Poisson manifolds}--- with the property that
    $\zeta_k|_{(Z_k^{\varepsilon_k}\setminus
      Z_k^{\varepsilon_k/2})\cap Z_k^{\varepsilon_k,\pm}}=\pm 1$.

    Though $\zeta_k$ is only locally defined around $Z_k$, the
    $b$-forms $\dfrac{\mathrm{d}\zeta_k}{\zeta_k}$ are globally
    well-defined $b$-forms,
    $\dfrac{\mathrm{d}\zeta_k}{\zeta_k}\in\,^b\Omega^1(M)$ ---as a
    matter of fact $\dfrac{\mathrm{d}\zeta_k}{\zeta_k}$ vanishes
    outside $Z_k^{\varepsilon_k}$.  Such $b$-forms provide a
    representation of any $\Theta\in[\Theta]\in\,^bH^1(M)$ given by
    \begin{align*}
    	\Theta=\vartheta-\sum_k a_k\frac{\mathrm{d}\zeta_k}{\zeta_k}\,,
    \end{align*}
    where $a_k\in\mathbb{R}$ and $[\vartheta]\in H^1(M)$ ---notice
    that the sum over $k$ is locally finite.  Finally, notice that
    $|\zeta_k|$ extends to a globally defined continuous function such
    that $|\zeta_k|\Big|_{M\setminus Z_k^{\varepsilon_k/2}}=1$.

    With these preliminaries the results of Lemma~\ref{Lem: constraint
      on the form of KMS on M minus Z} would then be modified as
    follows.  If $[\vartheta|_{M_Z}]\neq[0]$ then $\varphi_{M_Z}=0$
    for all $\varphi\in\KMS([\Theta^\sharp],\beta)$.
    If instead $\vartheta|_{M_Z}=\mathrm{d}H$ for
    $H\in C^\infty(M_Z)$, then there exists
    $c=\{c_j\}\in[0,+\infty)^{M_{Z,\mathrm{c}}}$ such that
    \begin{align}
        \varphi_{M_Z}(f)
        =
        \sum_jc_j\int_{M_{Z,j}} fe^{-\beta H_j} \prod_k|\zeta_k|^{a_k\beta-1}\mu\,,
        \qquad
        \forall f\in C^\infty_{\mathrm{c}}(M_Z)\,.
    \end{align}
    where $H_j:=H|_{M_{Z,j}}$.  Notice that both the sum over $j$ as
    well as the product over $k$ are finite on account of the
    compactness of $\supp(f)$.
\end{remark}
\begin{example}
    Let $M=\mathbb{R}^2$ be the $b$-Poisson manifold defined by
    $\Pi:=\sin(z)\partial_z\wedge\partial_t$.  In this case $Z$ is
    made by infinitely many connected components $Z_k:=\{z=\pi k\}$,
    moreover, $^bH^1(M)\simeq\bigoplus\limits_{k\geq 0} H^0(Z_k)$.  It is
    tempting to factorize the defining function $\sin(z)$ by
    exploiting the product formula
    \begin{align*}
        \sin(z)=z\prod_{k\geq 1}\left[1-\frac{z^2}{(\pi k)^2}\right]\,.
    \end{align*}
    This way we decompose any $\Theta\in[\Theta]\in\,^bH^1(M)$ as
    \begin{align*}
        \Theta=-a_0\mathrm{d}\log|z|
        -\sum_{k\geq 1}a_k\mathrm{d}\log\left|1-\frac{z}{\pi k}\right|
        -\sum_{k\geq 1}a_{-k}\mathrm{d}\log\left|1+\frac{z}{\pi k}\right|\,,
    \end{align*}
    which leads to the representation of $\varphi\in\KMS([\Theta^\sharp],\beta)$ as
    \begin{align*}
        \varphi(f)=\int f(z,t) |z|^{a_0\beta-1}
        \prod_{k\geq 1}\left|1-\frac{z}{\pi k}\right|^{a_k\beta-1}
        \left|1+\frac{z}{\pi k}\right|^{a_{-k}\beta-1}
        \mathrm{d}z\mathrm{d}t\,,
    \end{align*}
    for $f\in C^\infty_{\mathrm{c}}(M_Z)$.

    Unfortunately, though these formulas are rather explicit, they do
    not fit with the previous setting.  In particular the sum defining
    $\Theta$ is not locally finite and the same holds true for the
    product appearing in the definition of $\varphi$.  Therefore, the
    formulas for $\Theta$ and $\varphi$ make sense only for particular
    choices of the sequence $\{a_k\}_k$ -- for example they converge
    whenever $\sum\limits_{k\geq 1}\dfrac{a_k}{k}$ converges.

    As a matter of fact, none of the functions in the product formula
    for $\sin z$ fulfils the property required in Remark \ref{Rmk:
      Z-connectedness assumption for Lemma}.  To cope with this
    problem we have to replace each function
    $z,1+\dfrac{z}{\pi},1-\dfrac{z}{\pi},\ldots$ with functions
    $\zeta_0,\zeta_1,\ldots$ with the properties described in Remark
    \ref{Rmk: Z-connectedness assumption for Lemma}.
\end{example}
\begin{example}
    \label{Ex: toy example non-exactness of vartheta}%
    Notice that item~\ref{Item: varphiMZ formula if varthetaMZ is
      trivial} of Lemma~\ref{Lem: constraint on the form of KMS on M
      minus Z} does not imply that $[\vartheta]=[0]$ as shown by the
    example in Section~\ref{Sec: toy example b-cohomology class}.
    Indeed, therein we have $M_Z=M_{Z,0}\cup M_{Z,\pi}$, where
    $M_{Z,0}=(0,\pi)\times\mathbb{R}$ while
    $M_{Z,\pi}=(\pi,2\pi)\times\mathbb{R}$.  Choosing
    $\Theta=\vartheta=\mathrm{d}\theta$ we clearly have
    $[\mathrm{d}\theta|_{M_{Z,0}}] =
    [\mathrm{d}\theta|_{M_{Z,\pi}}]=[0]$ while
    $[\mathrm{d}\theta]\neq[0]$.
\end{example}

Lemma~\ref{Lem: constraint on the form of KMS on M minus Z} implies
that $(\Theta^\sharp,\beta)$-KMS functionals cannot be supported on
$M_Z$ if $[\vartheta|_{M_Z}]\neq [0]$ -- \textit{cf.}
Equation~\eqref{Eq: decomposition of closed b-form of degree 1}.
Moreover, the restriction $\varphi_{M_Z}$ of any such functional is
rather explicit.  The next theorem characterizes when a given KMS
functional on $M_Z$ extends to a KMS functional on $M$.  This provides
a classification of the convex set of KMS functionals for the case
$[\vartheta|_{M_Z}]=[0]$.

For later convenience we shall prove the following technical lemma:
\begin{lemma}
    \label{Lem: non-existence of positive extension}%
    Let $d_\pm,\beta\geq 0$ and let
    $T\in C^\infty_{\mathrm{c}}(\mathbb{R}\setminus\{0\})'$ be the
    positive distribution defined by
    \begin{align*}
        T(f)
        =
        \int_{\mathbb{R} \setminus\{0\}}
        f(z)D(z)|z|^{-(1+\beta)}\mathrm{d}z\,,
    \end{align*}
    where
    \begin{equation*}
        D(z)
        =
        \begin{cases}
            d_+ & z > 0 \\
            d_- & z < 0
        \end{cases}
    \end{equation*}
    with constants $d_+, d_- \in \mathbb{R}$.  Then there is no
    extension $\widehat{T}\in C^\infty_{\mathrm{c}}(\mathbb{R})'$ of
    $T$ which is positive unless $d_+=d_-=0$.
\end{lemma}
\begin{proof}
    By contradiction, let assume that at least one among $d_+,d_-$,
    say $d_+$, is stricly positive and let
    $\widehat{T}\in C^\infty_{\mathrm{c}}(\mathbb{R})$ be a positive
    extension of $T$.  We obtain a more explicit form of
    $\widehat{T}$: Let
    $C^\infty_{\mathrm{c},\lfloor\beta\rfloor}(\mathbb{R})$ be the
    subspace defined by
    \begin{align*}
        C^\infty_{\mathrm{c},\lfloor\beta\rfloor}(\mathbb{R})
        :=
        \{
        f\in C^\infty_{\mathrm{c}}(\mathbb{R})
        \;|\;
        f^{(k)}(0)=0\;\forall k\leq\lfloor\beta\rfloor
        \}\,,
    \end{align*}
    where $^{(k)}$ is a short notation for
    $\dfrac{\mathrm{d}^k}{\mathrm{d}z^k}$ and
    $\lfloor\beta\rfloor:=\max\{n\in\mathbb{N}\,|\,n\leq\beta\}$.

    Notice that $T$ extends uniquely to a positive functional
    $T_{\lfloor\beta\rfloor}$ on
    $C^\infty_{\mathrm{c},\lfloor\beta\rfloor}(\mathbb{R})$.  Indeed,
    a direct inspection shows that $T(f)$ is actually well-defined for
    $f\in C^\infty_{\mathrm{c},\lfloor\beta\rfloor}(\mathbb{R})$.
    Moreover, given two different positive distributions
    $T_1,T_2\in C^\infty_{\mathrm{c}}(\mathbb{R})'$ extending $T$ we
    have $T_1-T_2=c\delta$ for $c\in\mathbb{R}$ and therefore
    $T_1(f)=T_2(f)$ for all
    $C^\infty_{\mathrm{c},\lfloor\beta\rfloor}(\mathbb{R})$.

    It follows that $\widehat{T}(f)=T_{\lfloor\beta\rfloor}(f)$ for
    all $f\in C^\infty_{\mathrm{c}, \lfloor\beta\rfloor}(\mathbb{R})$.
    Moreover, let $\psi\in C^\infty_{\mathrm{c}}(\mathbb{R})$ be such
    that $\psi^{(k)}(0)=\delta^k_0$ for all $k\in\mathbb{N}\cup\{0\}$.
    For $f\in C^\infty_{\mathrm{c}}(\mathbb{R})$ we consider
    $p_\psi(f)\in
    C^\infty_{\mathrm{c},\lfloor\beta\rfloor}(\mathbb{R})$ with
    \begin{align*}
        p_\psi(f)(z)
        :=f(z)
        -\sum_{k\leq\lfloor\beta\rfloor}\frac{z^k}{k!}\psi(z)f^{(k)}(0)\,.
    \end{align*}
    Uniqueness of $T_{\lfloor\beta\rfloor}$ ensures that
    \begin{align}\label{Eq: explicit form of widehatT}
        \widehat{T}(f)
        =T_{\lfloor\beta\rfloor}(p_\psi(f))
        +\sum_{k\leq\lfloor\beta\rfloor}\widehat{T}_kf^{(k)}(0)\,,
    \end{align}
    where $\widehat{T}_k=-\widehat{T}(\frac{1}{k!}z^k\psi)$.  In what
    follows we choose $\psi$ so that
    $\psi_k:=\sum\limits_{\ell=0}^{k}\dfrac{z^\ell}{\ell!}\psi\geq 0$ for all
    $k\leq\lfloor\beta\rfloor$ -- this can be achieved by shrinking
    enough $\supp(\psi)$.

    We now claim that at least one of the constants
    $\{\widehat{T}_k\}_{k\leq\lfloor\beta\rfloor}$ is non-vanishing.
    Indeed, let us assume that $\widehat{T}_k=0$ -- that is,
    $\widehat{T}(z^k\psi)=0$ -- for all $k\leq\lfloor\beta\rfloor$.
    It follows that for all $f\in C^\infty_{\mathrm{c}}(\mathbb{R})$
    with $0\leq f\leq 1$ we have $f\psi_k\geq0$ and
    \begin{align*}
        0\leq\widehat{T}(f\psi_k)=
        \widehat{T}((f-1)\psi_k)\leq 0\,.
    \end{align*}
    This implies that $\widehat{T}(\psi_k f)=0$ for all
    $f\in C^\infty_{\mathrm{c}}(\mathbb{R})$ with $f\geq 0$.
    Therefore $\widehat{T}\psi_k=0$ thus 
    $\supp(\psi_k)\cap\supp(\widehat{T})=0$.
    It follows that $0\notin\supp(\widehat{T})$ because $\psi_k(0)=1$.  However,
    this would imply $0\notin\supp(T)$ too, a contradiction.  Thus
    there is at least one non-vanishing constant $\widehat{T}_k$'s.

    Let now
    $\kappa:=\max\{k\,|\, \widehat{T}_\ell=0\;\forall\ell<
    k\}\leq\lfloor\beta\rfloor$ and let
    $\psi_{n,\kappa}:=g_n\psi_\kappa\in
    C^\infty_{\mathrm{c}}(\mathbb{R})$, where we choose
    $g\in C^\infty_{\mathrm{c}}(\mathbb{R})$ such that $0\leq g\leq 1$
    and $g=1$ in a neighbourhood of $z=0$ and set $g_n(z):=g(2^nz)$.
    Notice that $\psi_{n,\kappa}\geq 0$ as well
    as
    \begin{align*}
        \psi^{(\ell)}_{n,\kappa}(0)=
        \begin{cases}
            1 & \textrm{for } \ell\leq\kappa\\
            0 & \textrm{for } \ell>\kappa.
        \end{cases}
    \end{align*}
    It follows that $\widehat{T}(\psi_{n,\kappa})\geq 0$, which
    entails
    \begin{align*}
        \widehat{T}_\kappa
        \geq\sup_nT((1-g_n)\psi_\kappa)
        \geq 0\,,
    \end{align*}
    where we used Equation~\eqref{Eq: explicit form of widehatT}
    together with $p_\psi\psi_{n,\kappa}=(g_n-1)\psi_\kappa$.  We
    shall now prove that $\sup_nT((1-g_n)\psi_\kappa)$ diverges.  In
    what follows we shall indicate with $O(1)$ a contribution which
    stays bounded in the limit $n\to+\infty$.

    We need to distinguish between two cases: $\beta>\kappa$ or
    $\beta=\kappa$.  In the former case integration by parts leads to
    \begin{multline*}
        T((1-g_n)\psi_\kappa)
        =
        \int(1-g_n)\psi_\kappa D(z)\frac{\mathrm{d}z}{|z|^{1+\beta}}
        \\
        =\int(1-g_n) \psi_\kappa^{(\kappa+1)} D_{\kappa+1}(z)
        \frac{\mathrm{d}z}{|z|^{\beta-\kappa}}
        -
        \sum_{\substack{j+\ell=\kappa+1\\j\geq 1\,,\,\ell\geq 0}}
        \frac{(\kappa+1)!}{j!\ell!}
        2^{nj}\int g^{(j)}(2^nz)
        \psi_\kappa^{(\ell)}(z)
        D_{\kappa+1}(z)\frac{\mathrm{d}z}{|z|^{\beta-\kappa}}\,,
    \end{multline*}
    where we used that
    $(1-g_n)\psi_\kappa\in
    C^\infty_{\mathrm{c}}(\mathbb{R}\setminus\{0\})$.  Here
    $D_{\kappa+1}(z)=D_{+,\kappa}$ for $z>0$ while
    $D_{\kappa+1}(z)=D_{-,\kappa}$ for $z<0$ for suitable constants
    $D_{\pm,\kappa}>0$ whose exact values are not necessary in what
    follows.
    Notice that
    $(\psi_\kappa^{(\kappa+1)})^{(\ell)}(0)=0$ for all $\ell$,
    therefore the first contribution is $O(1)$.  We find
    \begin{align*}
        T((1-g_n)\psi_\kappa)=O(1)
        -
        \sum_{\substack{j+\ell=\kappa+1\\j\geq 1\,,\,\ell\geq 0}}
        \frac{(\kappa+1)!}{j!\ell!}
        2^{n(j-1+\beta-\kappa)}
        \int g^{(j)}(y)
        \psi_\kappa^{(\ell)}(2^{-n}y)
        D_{\kappa+1}(y)\frac{\mathrm{d}y}{|y|^{\beta-\kappa}}\,.
    \end{align*}
    Since $g^{(j)}\in C^\infty_{\mathrm{c}}(\mathbb{R}\setminus\{0\})$
    for all $j\geq 1$, it holds
    \begin{align*}
        \lim_{n\to+\infty}
        \int g^{(j)}(y) \psi_\kappa^{(\ell)}(2^{-n}y)
        D_{\kappa+1}(y)\frac{\mathrm{d}y}{|y|^{\beta-\kappa}}
        =
        \int g^{(j)}(y) D_{\kappa+1}(y)
        \frac{\mathrm{d}y}{|y|^{\beta-\kappa}}
        =
        O(1)\,,
    \end{align*}
    where we used that $\psi_\kappa^{(\ell)}(0)=1$ for all
    $\ell\leq\kappa$.  The previous integral is not vanishing for a
    suitable choice of $g$.  It follows that
    \begin{align*}
        T((1-g_n)\psi_\kappa)
        =O(1)
        +\alpha 2^{n(\beta-\kappa)}\,,
    \end{align*}
    for $\alpha>0$ because $T((1-g_n)\psi_\kappa)\geq 0$.  This shows
    that $\sup_n T((1-g_n)\psi_\kappa)=+\infty$ when $\beta-\kappa>0$.

    Whenever $\beta=\kappa>0$ the same procedure leads to
    \begin{align*}
        T((1-g_n)\psi_\kappa)
        &=
        O(1)
        -
        \int_{\mathbb{R}}
        \sum_{\substack{j+\ell=\kappa+1\\j\geq 1\,,\,\ell\geq 0}}
        \frac{(\kappa+1)!}{j!\ell!}
        2^{nj}g^{(j)}(2^nz)
        \psi_\kappa^{(\ell)}(z)
        C_{\kappa+1}(z)\log|z|\mathrm{d}z\\
        &=
        O(1)
        +
        \alpha_1 n+\alpha_2 2^{n}\,,
    \end{align*}
    where the linear contribution arises from the term corresponding
    to $j=1$ in the previous sum.  Similarly when $\beta=\kappa=0$ we
    have $T((1-g_n)\psi_\kappa)=O(1)+\alpha n$, therefore, in all
    cases $\sup_n T((1-g_n)\psi_\kappa)=+\infty$.  It follows that
    there are no positive extensions $\widehat{T}$ of $T$.
\end{proof}
\begin{theorem}
    \label{Thm: KMS functionals for vartheta Hamiltonian in Mz}%
    Let $[\Theta]\in\,^bH^1(M)$ and let $a\in\mathbb{R}$ and
    $[\vartheta]\in H^1(M)$ be as in Equation~\eqref{Eq: decomposition
      of closed b-form of degree 1} with
    $\vartheta|_{M_Z}=\mathrm{d}H$ and $H\in C^\infty(M_Z)$.  Let
    $\mu\in\Omega^{2n}(M)$ be the volume form
    associated to the chosen defining function $\zeta$ and let
    $M_Z^\pm$ be the two connected
    components of $M_Z$.  Then:
    \begin{enumerate}
    \item\label{Item: KMS functionals for modular class in b-Poisson
          manifolds} There is an isomorphism of convex cones
        \begin{align}
            \label{Eq: KMS functionals for modular class in b-Poisson manifolds}
            \KMS([\Theta^\sharp],\beta)
            \simeq
            \begin{cases}
                \{0\}
                & \textrm{for } a<0
                \\
                [0,+\infty)^2
                & \textrm{for } a>0,
            \end{cases}
        \end{align}
		In particular, if $a>0$ and for all $c_\pm\geq 0$, the unique KMS
        functional $\varphi\in\KMS([\Theta^\sharp],\beta)$
        associated with $c_\pm$ is given by
        \begin{align}
            \label{Eq: KMS associated with positive function on connected components of MZ}
            \varphi(f)
            =
            \sum_\pm c_\pm \int_{M_Z^\pm}
            fe^{-\beta H_\pm}|\zeta|^{a\beta-1}\mu\,,
            \qquad
            \forall f\in C^\infty_{\mathrm{c}}(M)\,,
        \end{align}
        where $H_\pm:=H|_{M_Z^\pm}$.
    \item\label{Item: support constraint for b-KMS for a=0}
    If $\varphi \in \KMS(M, \Pi)$ then $\supp(\varphi)\subseteq Z$.
    \end{enumerate}
\end{theorem}
\begin{proof}
    We consider the cases $a>0$ and $a\leq 0$ separately:
    \begin{description}
    \item[$\boxed{a > 0\colon}$]  Let $\varphi\in\KMS([\Theta^\sharp],\beta)$.
        On account of Lemma~\ref{Lem: constraint on the form of KMS on
          M minus Z} there exists a unique pair $(c_+,c_-)\in[0,+\infty)$ such that
        Equation~\eqref{Eq: constrained form of KMS on M minus Z}
        holds true.  This leads to a map of convex cones
        \begin{align*}
            \Phi\colon\KMS([\Theta^\sharp],\beta)\to[0,+\infty)^2\,,
            \qquad
            \varphi\mapsto(c_+,c_-)\,.
        \end{align*}
        We now show that $\Phi$ is bijective.  Indeed $\Phi$ is
        injective: if $\Phi(\varphi)=(0,0)$, then
        $\supp(\varphi)\subseteq Z$ and therefore $\varphi=0$ because
        of Lemma~\ref{Lem: non-existence of KMS functional supported
          on Z with dynamics containing modular class}.

        Moreover, $\Phi$ is surjective: for all
        $c_\pm\in[0,+\infty)$, let
        $\varphi_{M_Z}\colon C^\infty_{\mathrm{c}}(M_Z)\to\mathbb{C}$
        be the linear functional defined by Equation~\eqref{Eq:
          constrained form of KMS on M minus Z}.  Notice that
        $|\zeta|^{a\beta-1}\mu$ is locally integrable on $M$ as one
        can see by considering a Darboux chart -- \textit{cf.}
        Remark \ref{Rmk: results on b-Poisson manifolds}.
        Moreover also $e^{-\beta H}\mu$ is locally integrable, though
        $H$ may diverge at $Z$.  To see this we exploit a Darboux
        chart $U\ni(z,x) = (z,x^1,\ldots,x^{2n-1})$ near $Z$ --
        \textit{cf.} Remark \ref{Rmk: results on b-Poisson manifolds}.
        On account of Equation~\eqref{Eq: decomposition
          of b-cotangent bundle} and of Remark~\ref{Rmk: on
          decomposition of b-forms} we have that
        $\lim\limits_{z\to 0}z\partial_zH=0$.  Thus, for all
        $\varepsilon>0$ there exists $\delta>0$ such that
        $|z\partial_zH(z,x)|<\varepsilon$ if $|z|<\delta$ while $x$
        ranges in a compact set.  Considering $0<|z_0|<\delta$ we have
	\begin{align*}
            |H(z,x)|
            \leq
            |H(z_0,x)|
            +
            \int_{z_0}^z|z\partial_zH(z,x)|\frac{\mathrm{d}z}{|z|}
            \leq
            |H(z_0,x)|
            +
            \varepsilon\log\bigg|\frac{z}{z_0}\bigg|\,,
	\end{align*}
	for $0<|z|<\delta$.  Therefore $e^{-\beta H(z,x)}\leq 1$ for
        points where $H(z,x)\geq 0$ while
        $e^{-\beta H(z,x)} = e^{\beta|H(z,x)|}\leq C(x)z^\varepsilon$
        when $H(z,x)\leq 0$.  It follows that $e^{-\beta H}$ is
        locally integrable nearby $Z$.

	Overall we find that $\varphi_{M_Z}$ extends to a KMS
        functional $\varphi\in\KMS([\Theta^\sharp],\beta)$.  The extension
        is unique on account of Lemma~\ref{Lem: non-existence of KMS
          functional supported on Z with dynamics containing modular
          class}.
    \item[$\boxed{a\leq 0\colon}$] For simplicity, we shall set $H=0$
        and $a=-1$ in the forthcoming discussion, so that
        $\Theta^\sharp=-Y_\mu$.  Different values for $H$ and $a$
        would not spoil the final result.  Let
        $\varphi\in\KMS(-Y_\mu,\beta)$.  Notice that, by allowing
        $\beta=0$, the following argument will also prove item
        \ref{Item: support constraint for b-KMS for a=0}, that is, it also covers the case $a=0$.
        Lemma~\ref{Lem: constraint on the form of KMS on M minus Z}
        entails that there exists $c_+,c_-\in[0,+\infty)$ such that,
        for all $f\in C^\infty_{\mathrm{c}}(M_Z)$,
	\begin{align*}
            \varphi(f)=\sum_\pm c_\pm\int_{M_Z^\pm}f|\zeta|^{-\beta-1}\mu\,.
	\end{align*}
	We now show that there are no positive extensions of
        $\varphi_{M_Z}$ to a functional on $C^\infty_{\mathrm{c}}(M)$:
        this entails that $\varphi=0$.  By contradiction, let
        $\widehat{\varphi}\colon
        C^\infty_{\mathrm{c}}(M)\to\mathbb{C}$ be a positive
        functional such that $\widehat{\varphi}(f)=\varphi(f)$ for all
        $f\in C^\infty_{\mathrm{c}}(M)$.

	The idea is to localize $\widehat{\varphi}$ on a Darboux
        chart, reducing the problem to the extension of a positive
        distribution on $C^\infty_{\mathrm{c}}(\mathbb{R})$.  For
        that, let consider a (small enough) local Darboux chart
        $U\stackrel{\phi}{\to}\mathbb{R}^{2n}$ near $Z$ and consider
        the restriction of $\widehat{\varphi}$ on functions supported
        in that chart.  This leads to a functional
        $\widetilde{\varphi}\colon
        C^\infty_{\mathrm{c}}(\mathbb{R}^{2n})\to\mathbb{C}$ such that
	\begin{align*}
            \widetilde{\varphi}(f)
            =
            \int_{\mathbb{R}^{2n}}
            f(z,x) C(z)|z|^{-\beta-1} \rho(z,x)
            \mathrm{d}z \mathrm{d}^{2n-1}x\,,
	\end{align*}
	where $C(z)=c'_\pm$ for $\pm z>0$ being $c'_\pm>0$, while
        $\rho>0$ is smooth and strictly positive.  By a suitable
        redefinition of $\zeta$ -- and thus of $\mu$ -- we can assume
        that $\rho(z,x)=\rho(x)$ is constant in $z$.  Let now
        $h\in C^\infty_{\mathrm{c}}(\mathbb{R}^{2n-1})$ be positive
        and let us consider the positive distribution
        $\widehat{T}\in C^\infty_{\mathrm{c}}(\mathbb{R})'$ defined by
        $\widetilde{T}(f):=\widetilde{\varphi}(fh)$.  In particular we
        find that $\widetilde{T}\supseteq T$ where
        $T\in C^\infty_{\mathrm{c}}(\mathbb{R}\setminus\{0\})'$ is
        defined by
	\begin{align*}
            T(f)
            :=
            \int_{\mathbb{R}}
            f(z) D(z)|z|^{-\beta-1}\mathrm{d}z\,,
	\end{align*}
	where $D(z)=d_+1_{\mathbb{R}_+}(z)+d_-1_{\mathbb{R}_-}(z)$ for
        $d_\pm=c'_\pm\int_{\mathbb{R}^{2n-1}}h(x)\rho(x)\mathrm{d}^{2n-1}x$.
        We then apply Lemma~\ref{Lem: non-existence of positive
          extension} which entails that $c_+=c_-=0$.
    This proves that $\supp(\varphi)\subseteq Z$: Lemma~\ref{Lem:
          non-existence of KMS functional supported on Z with dynamics
          containing modular class} entails that $\varphi=0$.
    \end{description}
\end{proof}

\begin{remark}
    \label{Rmk: Z-connectedness assumption for Theorem}%
    With reference to Remark \ref{Rmk: Z-connectedness assumption for
      Lemma} we comment on the generalization of Theorem \ref{Thm: KMS
      functional space for b-Poisson manifolds} to the case of
    non-connected $Z$.  For each connected component $Z_k$ of $Z$ let
    $M_{Z,Z_k}\subseteq M_{Z,c}$ be the collection of all connected
    components of $M_Z$ whose closure in $M$ have non-empty
    intersection with $Z_k$.  Moreover, for any sequence $a=\{a_k\}_k$
    let $M_{Z,\{a_k\}}$ be defined by
	\begin{align*}
            M_{Z,\{a_k\}}
            :=
            \left\{
                N\in M_{Z,c}
                \,\big|\,
                \forall k\colon
                N\in M_{Z,Z_k}\Rightarrow a_k>0
            \right\}\,.
	\end{align*}
	The results of Theorem \ref{Thm: KMS functional space for
          b-Poisson manifolds} are then generalized as follows: Item
        \ref{Item: support constraint for b-KMS for a=0} remains
        uneffected whereas Item \ref{Item: KMS functionals for modular
          class in b-Poisson manifolds} is replaced by
	\begin{align*}
		\KMS([\Theta^\sharp],\beta)
		\simeq
		[0,+\infty)^{M_{Z,\{a_k\}}}\,.
	\end{align*}
\end{remark}

Out of Equation~\eqref{Eq: KMS functionals for modular class in
  b-Poisson manifolds} we can identify those KMS functionals which
leads to extreme rays -- \textit{cf.} Section \ref{Sec: KMS convex
  cone}.  The case $a\leq 0$ is trivial while for $a>0$ we have the
following result.
\begin{corollary}
    \label{Cor: extremal ray KMS functionals for vartheta Hamiltonian in Mz}%
    In the hypothesis of Theorem~\ref{Thm: KMS functionals for
      vartheta Hamiltonian in Mz} and assuming $a>0$ let
    $\varphi\in\KMS([\Theta^\sharp],\beta)$.
    Let $\varphi_+,\varphi_-\in\KMS([\Theta^\sharp],\beta)$ be defined by
    Equation \eqref{Eq: KMS associated with positive function on
      connected components of MZ} with the particular choice $(1,0)$
    and $(0,1)$.  Then $\varphi\in\KMS([\Theta^\sharp],\beta)$ is extremal
    if and only if either $\varphi\in\mathbb{R}_+\varphi_+$ or
    $\varphi\in\mathbb{R}_+\varphi_-$.
\end{corollary}
\begin{proof}
    We prove the two implications:
    Let $\varphi\in\KMS([\Theta^\sharp],\beta)$ be associated with
    $c_+,c_-\in [0,+\infty)$ -- \textit{cf.} Equation~\eqref{Eq: KMS
      associated with positive function on connected components of MZ}
    -- such that $c_+c_->0$.  We observe that
    $\varphi=c_+\varphi_++c_-\varphi_-$, where $\varphi_+$,
    $\varphi_-$ are defined by Equation \eqref{Eq: KMS associated with
      positive function on connected components of MZ} with the
    particular choice $(1,0)$ and $(0,1)$.  Then
    $\varphi_\pm\in\KMS([\Theta^\sharp],\beta)$ and
    $\varphi_\pm\notin\mathbb{R}_+\varphi$: it follows that $\varphi$
    is not extremal.

    Conversely we shall prove that $\varphi_+$ are extremal -- the
    proof for $\varphi_-$ is analogous.
    Let $\varphi_1,\varphi_2\in\KMS([\Theta^\sharp],\beta)$ be such that
    $\varphi_+=\varphi_1+\varphi_2$: We shall prove that
    $\varphi_1,\varphi_2\in\mathbb{R}_+\varphi$.  For all
    $f\in C^\infty_{\mathrm{c}}(M\setminus \overline{M_Z^+})$ we have
    \begin{align*}
        0
        =
        \varphi_+(f)
        =
        \varphi_1(f) + \varphi_2(f)\,.
    \end{align*}
    Positivity of $\varphi_1,\varphi_2$ entails that
    $\varphi_1(f)=\varphi_2(f)=0$ for all
    $f\in C^\infty_{\mathrm{c}}(M \setminus \overline{M_Z^+})$ so that
    $\supp(\varphi_1)\cup\supp(\varphi_2)\subseteq M_Z^+$.  This
    entails in particular that $\varphi_1,\varphi_2$ are
    $(\Theta^\sharp|_{M_Z^+},\beta)$-KMS on $M_Z^+$, where
    $\Theta:=-a\dfrac{\mathrm{d}\zeta}{\zeta}+\vartheta$.  As $M_Z^+$
    is a connected symplectic manifold, Remark~\ref{Rmk: KMS for
      symplectic manifolds} entails that there exists
    $\lambda_1,\lambda_2\in\mathbb{R}_+$ such that
    $\varphi_+=\lambda_1\varphi_1=\lambda_2\varphi_2$.
\end{proof}
\begin{remark}
	\label{Rmk: Z-connectedness assumption for extremality}%
	Corollary~\ref{Cor: extremal ray KMS functionals for vartheta
          Hamiltonian in Mz} can be proved also for the case of
        non-connected $Z$.  With reference to Remarks~\ref{Rmk:
          Z-connectedness assumption for Lemma} and \ref{Rmk:
          Z-connectedness assumption for Theorem} we find that the
        extremal elements in $\KMS([\Theta^\sharp],\beta)$ are those
        supported on a single $N\in M_{Z,\{a_k\}}$.
\end{remark}
\begin{example}
    \label{Ex: toy example of KMS for vartheta vanishing on Mz}%
    With reference to the example in Section~\ref{Sec: toy example
      b-cohomology class} we investigate the convex cone
    $\KMS([X],\beta)$, providing results in accordance with
    Theorem~\ref{Thm: KMS functionals for vartheta Hamiltonian in Mz}
    -- see also Remark~\ref{Rmk: Z-connectedness assumption for
      Theorem}.  This also discusses an example of KMS functionals for
    a $b$-Poisson manifold with non-connected $Z$.  As explained in
    that example any Poisson vector field $X$ is of the form, up to
    Hamiltonian vector fields,
    \begin{align*}
        X=a_\vartheta\mathrm{d}\theta^\sharp
        +a_0\psi_0(\theta)\partial_y
        +a_\pi\psi_\pi(\theta)\partial_y\,,
    \end{align*}
    for suitable constants $c_\vartheta,c_0,c_\pi\in\mathbb{R}$ --
    here $\psi_0,\psi_\pi$ have been defined in Equation~\eqref{Eq:
      psi0, psipi definition}.  It follows that, on $M_Z$, the vector
    field $X$ is Hamiltonian $X = X_H$, with
    \begin{align*}
        H(\theta)
        =
        \begin{cases}
            \alpha_0
            + a_\vartheta \theta
            - a_0\log|\sin(\theta/2)|
            - a_\pi\log|\cos(\theta/2)|
            & \textrm{for } 0<\theta<\pi
            \\
            \alpha_\pi
            + a_\vartheta \theta
            - a_0\log|\sin(\theta/2)|
            - a_\pi\log|\cos(\theta/2)|
            & \textrm{for } \pi<\theta<2\pi,
        \end{cases}
    \end{align*}
    where $\alpha_0, \alpha_\pi \in \mathbb{R}$.  This entails that
    the restriction $\varphi_{M_Z}=M_{Z,0}\cup M_{Z,\pi}$ to $M_Z$ of
    any $\varphi\in\KMS(X,\beta)$ can be written as
    \begin{align*}
        \varphi_{M_Z}(f)
        =
        \sum_{j\in\{0,\pi\}}c_j\int_{M_{Z,j}}
        f(\theta,y)e^{-\beta a_\vartheta \theta}
        |\sin(\theta/2)|^{a_0\beta-1}
        |\cos(\theta/2)|^{a_\pi\beta-1}
        \mathrm{d}\theta\mathrm{d}y\,,
    \end{align*}
    for all $f\in C^\infty_{\mathrm{c}}(M_Z)$, where $c_0, c_\pi > 0$
    -- the constants $\alpha_0, \alpha_\pi$ have been absorbed into
    the definition of $c_0, c_\pi$.  The latter functional admits a
    positive extension to $M$ if and only if $a_0, a_\pi > 0$.  This
    provides an isomorphism of convex cones
    $\KMS([X],\beta)\simeq [0,+\infty)^2$ as claimed in
    Theorem~\ref{Thm: KMS functionals for vartheta Hamiltonian in Mz}.
\end{example}

Lemma~\ref{Lem: constraint on the form of KMS on M minus Z} and
Theorem~\ref{Thm: KMS functionals for vartheta Hamiltonian in Mz}
provide a description of $\KMS([\Theta^\sharp],\beta)\simeq [0,+\infty)^2$ for all $a\neq 0$
and $[\vartheta|_{M_Z}]=[0]$.  It remains to discuss the case of
$[\vartheta|_{M_Z}]\neq [0]$ and $a=0$ -- \textit{cf.} Equation \eqref{Eq: decomposition of closed b-form of degree 1}.
However, Lemma~\ref{Lem: constraint on the form of KMS on M minus Z} entails that any
$\varphi\in\KMS([\vartheta^\sharp],\beta)$ is necessarily supported on $Z$.
Since $\vartheta^\sharp$ is tangent to $Z$ this reduces the problem to
the discussion of $(\vartheta^\sharp|_Z,\beta)$-KMS functionals on
$Z$, for which we may apply Theorem~\ref{Thm: KMS functional space for
  cosymplectic manifolds}.
This provides a classification of $\KMS([\Theta^\sharp],\beta)$ which is complete whenever $(Z,\Pi_Z)$ is compact with one (hence all) proper leaf.

We summarize our results in a single theorem.
\begin{theorem}
    \label{Thm: KMS functional space for b-Poisson manifolds}%
    Let $(M,\Pi)$ be a $b$-Poisson manifold with (connected) critical locus $Z$.
    Let $\vartheta\in[\vartheta]\in H^1(M)$ and $a\in\mathbb{R}$ and
    set
    $[\Theta]:=[-a\frac{\mathrm{d}\zeta}{\zeta}+\vartheta]\in\,^bH^1(M)$,
    where $\zeta$ is a $Z$-defining function.
    \begin{enumerate}
    \item \label{item:KMSbPoissonCases} Then the convex cone
        $\KMS([\Theta^\sharp],\beta)$ sits in one of the following cases:
	\begin{center}
            \begin{tabular}{c|c|c|c}
                & $a>0$
                & $a= 0$
                & $a<0$
                \\
                \hline
                $[\vartheta|_{M_Z}]=[0]$
                & $[0,+\infty)^2$
                & $\KMS(Z,\Pi_Z)$
                & $\{0\}$
                \\
                $[\vartheta|_{M_Z}]\neq [0]$
                & $\{0\}$
                & $\KMS([\vartheta^\sharp|_Z],\beta)$
                & $\{0\}$
            \end{tabular}
	\end{center}
    \item \label{item:UniqueRegularKMSbPoissonCase} There exists a
        unique (up to isomorphism) regular KMS functional, obtained
        for $a\beta=1$ and $[\vartheta] = [0]$: Such a functional is
        induced by the smooth density $\mu$ such that
        $\zeta := \iota_{\Pi^{\wedge n}}\mu$.
    \item \label{item:PoissonTracesbPoisson} Poisson traces are
        supported on $Z$ and there exists a unique Poisson trace which
        is invariant under the modular vector field $[Y_\Pi]$: The
        latter is induced by the smooth density on $Z$ given by
        $\alpha_Z\wedge\sigma_Z^{\wedge (n-1)}$ -- \textit{cf.}
        Remark~\ref{Rmk: results on b-Poisson manifolds}.
    \end{enumerate}
\end{theorem}
\begin{example}
    We consider the $b$-Poisson manifold
    \begin{align*}
        M
        =
        \mathbb{R}\times\mathbb{T}\times\mathbb{T}\times\mathbb{R}\,,
        \qquad
        \Pi
        =
        z\partial_z\wedge\partial_{\theta_1}
        +
        \partial_{\theta_2}\wedge\partial_{y}\,,
    \end{align*}
    where $\theta_1,\theta_2\in[0,2\pi)$ while $z,y\in\mathbb{R}$.
    Comparing with Remark~\ref{Rmk: results on b-Poisson manifolds} we
    have $\alpha=\mathrm{d}\theta_1$ and
    $\sigma=\mathrm{d}\theta_2\wedge\mathrm{d}y$.  Moreover, the
    isomorphism \eqref{Eq: characterization of b-Poisson cohomology}
    leads $H_\Pi^1(M)\simeq\mathbb{R}^3$: In fact, the generic Poisson
    vector field $X$ is, modulo Hamiltonian vector fields,
    \begin{align*}
        X
        =
        a_1\mathrm{d}\theta_1^\sharp
        +
        a_2\mathrm{d}\theta_2^\sharp
        -
        a\frac{dz}{z}^\sharp
        =
        a_1z\partial_z
        -
        a_2\partial_y
        +
        a\partial_{\theta_1}\,,
    \end{align*}
    for $a,a_1,a_2\in\mathbb{R}$.
    Theorem~\ref{Thm: KMS functional space for b-Poisson manifolds} and a direct
    inspection lead to the following classification:
    \begin{center}
        \begin{tabular}{c|c|c|c|}
            & $a>0$ & $a= 0$ & $a<0$ \\\hline
            $a_1=a_2=0$ & $[0,+\infty)^2$ & \multirow{2}{*}{$\mathcal{M}_+(\mathbb{T})$} & \multirow{4}{*}{$\{0\}$}\\\cline{1-2}
            $a_1\neq 0\;a_2=0$ & $\{0\}$ &&\\\cline{1-3}
            $a_1= 0\;a_2\neq0$ & \multicolumn{2}{|c}{\multirow{2}{*}{$\{0\}$}}&\\\cline{1-1}
            $a_1a_2>0$&\multicolumn{2}{|c}{\multirow{2}{*}{}}&\\\hline
        \end{tabular}
    \end{center}
    In fact for $a\neq 0$ we necessarily have $a_1=a_2=0$ and $a>0$,
    otherwise $\KMS([\Theta^\sharp],\beta)=\{0\}$ -- \textit{cf.}
    Lemma~\ref{Lem: constraint on the form of KMS on M minus Z}.  For
    the latter case $\KMS([a\partial_{\theta_1}],\beta)\simeq[0,+\infty)^2$ as any
    $(a\partial_{\theta_1},\beta)$-KMS functional is given by
    \begin{align*}
        \varphi(f)
        =
        \int_M
        f(z,\theta_1,\theta_2,y)c(z)
        |z|^{a\beta-1}
        \mathrm{d}z\mathrm{d}\theta_1\mathrm{d}\theta_2\mathrm{d}y\,,
        \qquad
        c(z)
        :=
        \begin{cases}
            c_+ & z > 0 \\
            c_- & z < 0,
        \end{cases}
    \end{align*}
    where $c_\pm\geq 0$.  If we focus on Poisson traces, \textit{i.e}
    $a=a_1=a_2=0$, we are reduced to the classification of Poisson
    traces for the cosymplectic manifold
    $Z=\mathbb{R}\times\mathbb{T}^2$ with $\eta=\mathrm{d}\theta_1$
    and $\omega=\mathrm{d}\theta_2\wedge\mathrm{d}y$.  The latter
    space is isomorphic to the convex cone $\mathcal{M}_+(\mathbb{T})$
    of positive measure on $\mathbb{T}$: Actually, for all
    $\varphi \in \KMS(M,\Pi) \simeq \KMS(Z,\Pi_Z)$ there exists
    $\nu_\varphi \in \mathcal{M}_+(\mathbb{T})$ such that
    \begin{align*}
        \varphi(f)
        =
        \int
        f(0,\theta_1,\theta_2,y)
        \mathrm{d}\nu_\varphi(\theta_1)\mathrm{d}\theta_2\mathrm{d}y\,,
        \qquad
        \forall f\in C^\infty_{\mathrm{c}}(M)\,.
    \end{align*}
    The proof of this result goes along the same line of
    Example~\ref{Ex: KMS functionals for plane Poisson manifold}.

    Finally if $a=0$ and $a_2=0$ we have
    $\KMS([a_1z\partial_z],\beta) = \KMS(M,\Pi) \simeq \mathcal{M}_+(\mathbb{R})$
    no matter the value of $a_1$: This is due to the fact that
    $a_1 \mathrm{d}\theta_1^\sharp|_{Z=0}=0$.  If $a=0$ and
    $a_2\neq 0$ then $\KMS([a_1z\partial_z-a_2\partial_y],\beta) = \{0\}$.
    Indeed, let assume for simplicity $a_2 = 1$: then for all $\varphi\in\KMS([a_1z\partial_z-\partial_y],\beta)$
    and $f\in C^\infty_{\mathrm{c}}(M)$ we have
    \begin{align*}
        \beta\varphi(f)
        =
        \beta\varphi(\partial_y(y1_f)f)
        =
        -\beta\varphi(\{f,y1_f\})
        =
        -\beta\varphi(\partial_{\theta_2}f)\,,
    \end{align*}
    where $1_f\in C^\infty_{\mathrm{c}}(M)$ is such that
    $1_f|_{\supp(f)}=1$.  It follows that
    $\varphi(f)=e^{-\beta t}\varphi(\Psi_t^*f)$ where $\Psi_t$ is the
    flow associated with $\partial_{\theta_2}$.  As
    $\varphi(\Psi^*_tf)$ is bounded in $t$ we have
    $\varphi(f)=\lim\limits_{t\to+\infty}e^{-\beta
      t}\varphi(\Psi^*_tf)=0$.
\end{example}

As expected, the structure of the convex cone $\KMS([X],\beta)$
for $b$-Poisson manifold is way richer than the corresponding convex
cone for symplectic manifolds.  In particular, the $b$-Poisson case
supports KMS functionals for dynamics $[X]\neq [Y_\Pi]$, though none
of these functionals is induced by smooth densities.  In fact for
$[X]\neq[Y_\Pi]$ the non-trivial cases include the KMS convex cone of
the cosymplectic manifold $Z\hookrightarrow M$, where the results of
Theorem~\ref{Thm: KMS functional space for cosymplectic manifolds}
apply.  Moreover, the behaviour of $\KMS([aY_\Pi],\beta)$ shows
a phase transition for $a=0$.  In fact $\KMS([aY_\Pi],\beta)$ is
empty for $a<0$ and isomorphic to $[0,+\infty)^2$ for $a>0$
while it reduces to Poisson traces on the coymplectic manifold $Z$ for
$a=0$.  It would be interesting to see whether or not the pattern
shown in the $b$-Poisson case holds true for $b^k$-Poisson manifolds
\cite{Scott-16}: in the latter case we expect a phase transition to
occur also in the parameter $\beta>0$ as shown by the following
example.
\begin{example}
    \label{Ex: bk-Poisson manifold}%
    Let consider the Poisson manifold defined by
    \begin{align*}
        M=\mathbb{R}^2\,,
        \qquad
        \Pi=z^k\partial_z\wedge\partial_y\,,
        \qquad
        k\in\mathbb{N}\,.
    \end{align*}
    According to \cite[Def.~2.8]{Scott-16} this Poisson manifold
    $(M,\Pi)$ is an example of $b^k$-Poisson manifold.  Focusing for
    simplicity on the modular class, we have
    $[Y_\Pi]=[kz^{k-1}\partial_y]\neq[0]$.  Setting $Z:=\{z=0\}$ and
    $M_Z:=M\setminus Z$ a direct inspection shows that the restriction
    to $M_Z$ of any $\varphi\in\KMS(kz^{k-1}\partial_y,\beta)$ is
    necessarily of the form
    \begin{align*}
        \varphi(f)
        =
        \int_{M_Z}
        f(z,y)c(z)|z|^{k(\beta-1)}\mathrm{d}z\mathrm{d}y\,,
        \qquad
        \forall f\in C^\infty_{\mathrm{c}}(M_Z)\,,
    \end{align*}
    where
    \begin{equation*}
        c(z)
        :=
        \begin{cases}
            c_+ & z > 0 \\
            c_- & z < 0,
        \end{cases}
    \end{equation*}
    with $c_\pm\geq 0$.  A slight modification of Lemma~\ref{Lem:
      non-existence of positive extension} shows that the positive
    functional defined on the right-hand side of the latter equation
    has no positive extension to $M$ if $\beta\leq 1-1/k$.  Therefore in
    the latter case we have $\supp(\varphi)\subseteq Z$ where both
    $\Pi$ and $Y_\Pi$ vanish.  This shows that
    \begin{align*}
        \KMS([Y_\Pi],\beta)
        \simeq
        \begin{cases}
            [0,+\infty)^2
            & \textrm{for } \beta>\frac{k-1}{k}
            \\
            \mathcal{M}_+(\mathbb{R})
            & \textrm{for } \beta\leq \frac{k-1}{k}\,,
        \end{cases}
    \end{align*}
    This shows the occurrence of a phase transition in the parameter
    $\beta>0$ at $\beta=\frac{k-1}{k}$.
\end{example}

\section{Examples}
\label{Sec: examples}

In this section we discuss in details two examples of regular
codimension $1$ Poisson manifolds $(M,\Pi)$ for which
$\KMS([X],\beta)$ can be computed explicitly for all
$[X]\in H_\Pi^1(M)$.  This investigation is mainly motivated by the
need of building some intuition on which properties of the geometry of
$(M,\Pi)$ are codified by $\KMS([X],\beta)$.

For example, on account of Remark \ref{Rmk: KMS for symplectic manifolds} and Proposition~\ref{Prop: extremal KMS supported on a proper leaf is extremal}, one may think that extremal elements
$\varphi$ of $\KMS([X],\beta)$ are necessarily supported on
leaves $L$ of $M$, which can then be interpreted as ``pure
thermodynamical phases'' of the system under consideration.  However,
this point of view disregards cases where the natural KMS state
defined on $C^\infty_{\mathrm{c}}(L)$ cannot be pushed forward to
well-defined functionals on $C^\infty_{\mathrm{c}}(M)$ -- this would
happen when $L$ is not a closed embedded submanifold in $M$.  In this situation
there may exist KMS functionals on $M$, though each leaf $L$ has none:
Section~\ref{Sec: 3-torus regular Poisson manifolds} provides an
example in this direction.

Thus it appears that $\KMS([X],\beta)$ is capable to capture
whether a leaf $L$ is a nice submanifold or not.  The example
considered in Section~\ref{Sec: spiral regular Poisson manifolds} is
even more specific: Therein, one can find $[X]\in H_\Pi^1(M)$ for
which there are KMS functionals on $M$ supported on leaves $L$ which
are not embedded submanifolds.  This is particularly interesting as it seems to provide some information about how badly $L$ fails to be an embedded
submanifold.

\subsection{The Torus $\mathbb{T}^3$ with a Cosymplectic Poisson
  Structure}
\label{Sec: 3-torus regular Poisson manifolds}

We now consider the example of a family of compact cosymplectic
Poisson manifolds $(M,\Pi_c)$ which depends on a parameter
$c\in\mathbb{R}$.  We shall analyse in details the convex set of KMS
states $\kms([X],\beta)$ for a given $[X]\in H_{\Pi_c}^1(M)$
and $\beta>0$.  As already discussed in Section~\ref{Sec: Main
  properties of classical KMS functionals} we shall focus on the set
$\kms_0([X],\beta)$ of extremal KMS states.  As we shall
see, the structure of $\kms_0([X],\beta)$ strongly depends on
whether $c$ is chosen to be rational or not.

Let $c\in\mathbb{R}$ and let $(M,\Pi_c)$ be the cosymplectic Poisson
manifold defined by
\begin{align*}
    M
    :=
    \mathbb{T}^3\,,
    \qquad
    \eta_c:=c\mathrm{d}\theta_2-\mathrm{d}\theta_3\,,
    \qquad
    \omega:=\mathrm{d}\theta_1\wedge\mathrm{d}\theta_2\,,
\end{align*}
whose associated Poisson vector by Equation~\eqref{Eq: Poisson tensor
  for cosymplectic manifolds} is given by
\begin{align}
    \label{Eq: 3-torus regular codim 1 Poisson manifold}
    \Pi_c
    :=
    \partial_{\theta_1}
    \wedge
    (\partial_{\theta_2}+c\partial_{\theta_3})\,,
\end{align}
where $\theta_1,\theta_2,\theta_3$ are the usual quasi-global
coordinates on $\mathbb{T}^3$.

As discussed in Section~\ref{Sec: cosymplectic manifolds}, the
symplectic foliation of $(M,\Pi_c)$ coincides with the foliation
associated with the $1$-form $\eta_c$, see Remark~\ref{Rmk: results on
  cosymplectic manifolds}.  In particular, if
$c = \frac{p}{q} \in \mathbb{Q}$ -- here $p, q\in\mathbb{Z}$ are such
that $\mathrm{GCD}(p,q) = 1$ -- then the symplectic leaf through
$(\bar{\theta}_1,\bar{\theta}_2,\bar{\theta}_3)$ is the proper
submanifold of $\mathbb{T}^3$ which in terms of quasi-global
coordinates is defined by
\begin{align}
    \label{Eq: 3-torus foliation for rational c}
    L_{(\bar{\theta}_1,\bar{\theta}_2,\bar{\theta}_3)}
    =
    L_{(0,\bar{\theta}_2,\bar{\theta}_3)}
    =
    \lbrace(u_1,\bar{\theta}_2+qu_2,\bar{\theta}_3+pu_2)
    \,|\,
    u_1,u_2\in[0,2\pi)
    \rbrace\,.
\end{align}
If instead $c\in\mathbb{R}\setminus\mathbb{Q}$ we have
\begin{align}
    \label{Eq: 3-torus foliation for irrational c}
    L_{(\bar{\theta}_1,\bar{\theta}_2,\bar{\theta}_3)}
    =
    L_{(0,\bar{\theta}_2,\bar{\theta}_3)}
    =
    \left\{
        \big(
        u_1,
        \Psi_t^{\partial_{\theta_2} + c\partial_{\theta_3}}
        (\bar{\theta_2}, \bar{\theta_3})
        \big)
        \,\Big|\,
	u_1\in[0,2\pi)\,,\,
	t \in \mathbb{R}
    \right\}\,,
\end{align}
where $\Psi_t^{\partial_{\theta_2}+c\partial_{\theta_3}}$ denotes the
flow along the vector field
$\partial_{\theta_2}+c\partial_{\theta_3}$.

We now compute explicitly $H_{\Pi_c}^1(M)$.  Notice that, using the
isomorphism \eqref{Eq: characterization of cosymplectic Poisson
  cohomology}, the first Poisson cohomology $H_{\Pi_c}^1(M)$ is
isomorphic to $H_{\eta_c}^1(M)\oplus H_{\eta_c}^0(M)$.  Since in
computing the latter spaces one faces the same complexity, we
preferred to proceed with a direct computation.  To deal with the
irrational case $c\in\mathbb{R}\setminus\mathbb{Q}$ the following
lemma will be helpful.
\begin{lemma}
    \label{Lem: magnitude of rationals approximating irrationals}%
    Let $c\in\mathbb{R}\setminus\mathbb{Q}$.  Then for all
    $n\in\mathbb{N}$ there exists a $k\in\mathbb{N}$ such that
    \begin{align*}
        \forall p,q\in\mathbb{Z}\colon
        |p-cq|<\frac{1}{kn}\,,\,
        1\leq q\leq kn
        \Longrightarrow q\geq n\,.
    \end{align*}
\end{lemma}
\begin{proof}
    Let $A_n\subseteq \mathbb{Q}$ be the set defined by
    \begin{align*}
        A_n
        :=
        \left\lbrace
            \tfrac{p}{q}\in\mathbb{Q}
            \,\Big|\,
            |p-cq|\leq \tfrac{1}{n}\,,\,
            1\leq q\leq n
        \right\rbrace\,.
    \end{align*}
    A pigeonhole principle argument shows that $A_n\neq\emptyset$.
    Indeed let us consider the $n+1$ elements
    $\{jc\}:=jc-\lfloor jc\rfloor\in(0,1)$ for $j\in\{0,\ldots,n\}$,
    where $\lfloor x\rfloor:=\sup\{m\in\mathbb{N}\,|\,m\leq x\}$.
    Considering the partition
    $(0,1)=\bigcup_{j=0}^{n-1}(\frac{j}{n},\frac{j+1}{n})$ of $(0,1)$
    into $n$ disjoint intervals, the pigeonhole principle entails that
    there exists $i,j,h\in\{0,\ldots,n\}$ such that $i<j$ and
    $\{ic\},\{jc\}\in(\frac{h}{n},\frac{h+1}{n})$.  This implies that
    \begin{align}
        |\lfloor jc\rfloor-\lfloor ic\rfloor-(j-i)c|<\frac{1}{n}\,,
        \qquad
        1\leq (j-i)\leq n\,,
    \end{align}
    so that $\frac{p}{q}\in A_n$ there
    $p=\lfloor ic\rfloor-\lfloor jc\rfloor$ and $q=j-i$.  The set
    $A_n$ is finite and therefore closed. Moreover, $c\notin A_n$ so
    that
    \begin{align*}
        0<d(A_n,c)
        =
        \min_{\frac{p}{q}\in A_n}\left|\frac{p}{q}-c\right|
        \leq
        \frac{1}{n}\,.
    \end{align*}
    Since $d(A_n,c)>0$ there exists $k\in\mathbb{N}$ such that
    $\frac{1}{kn}< d(A_n,c)\leq\frac{1}{n}$.  Let now
    $p,q\in\mathbb{Z}$ such that $|p-cq|<\frac{1}{kn}$ and
    $1\leq q\leq kn$.  It follows that $\frac{p}{q}\notin A_n$ and
    therefore $q\geq n$.
\end{proof}

\begin{lemma}
    \label{Lem: H1 3-torus}%
    For the Poisson manifold defined in Equation~\eqref{Eq: 3-torus
      regular codim 1 Poisson manifold} we have that
    \begin{align}
        H_{\Pi_c}^1(M)\simeq H_{\Pi_c}^0(M)^{\oplus 3}\,,
    \end{align}
    where
    $H_{\Pi_c}^0(M)=\{f\in
    C^\infty(M)\,|\,\partial_{\theta_1}f=0\,,\,(\partial_{\theta_2}+c\partial_{\theta_3})f=0\}$
    ---notice in particular that $H_{\Pi_c}^0(M)\simeq\mathbb{R}$ for
    $c\in\mathbb{R}\setminus\mathbb{Q}$.  Actually any
    $[X]\in H_{\Pi_c}^1(M)$ is of the form
    \begin{align}\label{Eq: 3torus Poisson vector field}
        [X]
        =
        [t_1\partial_{\theta_1}
        +
        t_2(\partial_{\theta_2}+c\partial_{\theta_3})
        +
        t_3\partial_{\theta_3}]\,,
        \qquad
        t_1,t_2,t_3\in H_{\Pi_c}^0(M)\,.
    \end{align}
\end{lemma}
\begin{proof}
    Let $X\in\Gamma(TM)$ be a Poisson vector field: We decompose $X$
    as
    \begin{align*}
        X
        =\tau_1\partial_{\theta_1}
        +\tau_2\partial_{\theta_2}
        +\tau_3\partial_{\theta_3}\,,
    \end{align*}
    and we impose the condition $\mathcal{L}_X(\Pi_c)=0$.  This leads
    to
    \begin{subequations}
        \begin{align}
            \label{Eq: tau12 relation}
            \partial_{\theta_1}\tau_1
            +(\partial_{\theta_2}+c\partial_{\theta_3})\tau_2
            &=0\\
            \label{Eq: tau13 relation}
            c\partial_{\theta_1}\tau_1
            +(\partial_{\theta_2}+c\partial_{\theta_3})\tau_3
            &=0\\
            \label{Eq: tau23 relation}
            c\partial_{\theta_1}\tau_2
            -\partial_{\theta_1}\tau_3
            &=0\,.
        \end{align}
    \end{subequations}
    It follows that
    \begin{align*}
        X
        =\tau_1\partial_{\theta_1}
        +\tau_2(\partial_{\theta_2}
        +c\partial_{\theta_3})
        +t_3\partial_{\theta_3}\,,
    \end{align*}
    where $\tau_1$, $\tau_2$ are related by Equation~\eqref{Eq: tau12
      relation} while $t_3\in H_{\Pi_c}^0(M)$.
    We shall now discuss the cases $c\in\mathbb{Q}$ and
    $c\notin\mathbb{Q}$ separately.
    \begin{description}
    \item[$\boxed{c\in\mathbb{Q}\colon}$] Let
        $c=\frac{p}{q}\in\mathbb{Q}$ where $p,q\in\mathbb{Z}_+$ are
        such that $\operatorname{GCD}(p,q)=1$.  We now look for
        $f\in C^\infty(M)$ such that
        $\tau_1\partial_{\theta_1}+\tau_2(\partial_{\theta_2}+c\partial_{\theta_3})=X_f$
        which implies
	\begin{align}\label{Eq: tau12 Hamiltonian PDE system}
            \tau_1
            =(f_{\theta_2}+cf_{\theta_3})\,,\qquad
            \tau_2
            =-f_{\theta_1}\,.
	\end{align}
	Any such function would satisfy
	\begin{multline}\label{Eq: tau12 Hamiltonian}
            f(\theta_1,\theta_2+qu,\theta_3+pu)
            :=f(0,\theta_2,\theta_3)
            -\int_0^{\theta_1}\tau_2(s_1,\theta_2,\theta_3)\mathrm{d}s_1
            \\+\int_0^{u}\tau_1(\theta_1,\theta_2+qs_2,\theta_3+ps_2)\mathrm{d}s_2\,.
	\end{multline}
	In particular, evaluation at $\theta_1=2\pi$ or $u=2\pi$
        implies the following consistency relations:
	\begin{subequations}\label{Eq: consistency relations for rational c}
            \begin{align}
                (\mathcal{P}_{\theta_1}\tau_2)(\theta_2,\theta_3)
                &:=\int_0^{2\pi}\tau_2(\theta_1,\theta_2,\theta_3)\mathrm{d}\theta_1
                =0\,,\\
                (\mathcal{P}_c\tau_1)(\theta_1,\theta_2,\theta_3)
                &:=\int_0^{2\pi}\tau_1(\theta_1,\theta_2+qu,\theta_3+pu)\mathrm{d}u
                =0\,.
            \end{align}
	\end{subequations}
	Notice that on account of equation \eqref{Eq: tau12 relation}
        we have
        $\mathcal{P}_{\theta_1}\tau_2,\mathcal{P}_c\tau_1\in
        H_{\Pi_c}^0(M)$.  Whenever Equation~\eqref{Eq: consistency
          relations for rational c} are satisfied, Equation~\eqref{Eq:
          tau12 Hamiltonian} can be used to define a smooth function
        $f\in C^\infty(M)$.  In fact it suffices to specify $f$ on,
        say, the submanifold
        $\{(0,\theta_2,\theta_3)\in[0,2\pi)^3\,|\,p\theta_2-q\theta_3=0\}$,
        and then use Equation~\eqref{Eq: tau12 Hamiltonian}. We write
	\begin{align*}
            [X]
            =[\mathcal{P}_c\tau_1\partial_{\theta_1}
            +\mathcal{P}_{\theta_1}\tau_2(\partial_{\theta_2}+c\partial_{\theta_3})
            +t_3\partial_{\theta_3}]\,,
	\end{align*}
	so that the linear map
	\begin{align*}
            H^1(M;\Pi_c)\ni[X]\mapsto
            (\mathcal{P}_c\tau_1,\mathcal{P}_{\theta_1},t_3)
            \in H_{\Pi_c}^0(M)^{\oplus 3}\,,
	\end{align*}
	completes the proof.

    \item[$\boxed{c\in\mathbb{R}\setminus\mathbb{Q}\colon}$] The
        strategy of the proof is similar to the case of rational $c$.
        In particular, we look for a solution $f\in C^\infty(M)$ of
        the system \eqref{Eq: tau12 Hamiltonian PDE system}.  This
        implies that
	\begin{subequations}\label{Eq: consistency relations for irrational c}
            \begin{align}
                (\mathcal{P}_{\theta_1}\tau_2)(\theta_2,\theta_3)
                &:=
                \int_0^{2\pi}
                \tau_2(\theta_1,\theta_2,\theta_3)
                \mathrm{d}\theta_1
                =0\,,
                \\
                (\mathcal{P}_{\theta_2}\mathcal{P}_{\theta_3}\tau_1)(\theta_1)
                &:=
                \int_0^{2\pi}\int_0^{2\pi}
                \tau_1(\theta_1,\theta_2,\theta_3)
                \mathrm{d}\theta_2\mathrm{d}\theta_3
                =0\,.
            \end{align}
	\end{subequations}
	We shall now show that the Equations~\eqref{Eq: consistency
          relations for irrational c} are sufficient for proving the
        existence of $f$.  In particular the condition
        $(\mathcal{P}_{\theta_1}\tau_2)(\theta_2,\theta_3)=0$ entails
        that we can define $f$ as
	\begin{align*}
            f(\theta_1,\theta_2,\theta_3)
            =
            f(0,\theta_2,\theta_3)
            -
            \int_0^{\theta_1}\tau_2(u,\theta_2,\theta_3)\mathrm{d}u\,,
	\end{align*}
	so that it suffices to determine
        $f_0(\theta_2,\theta_3)=f(0,\theta_2,\theta_3)$ in a way
        compatible with the first equation in \eqref{Eq: tau12
          Hamiltonian PDE system}.  We now show that the equation
	\begin{align*}
            (\partial_{\theta_2}+c\partial_{\theta_3})f_0
            =
            \tau_1\,,
	\end{align*}
	admits a solution which is unique up to a constant.  Regarding
        the uniqueness statement, let
        $g:=f_0-\hat{f}_0\in C^\infty(M)$ be the difference of two
        solutions.  Then
        $(\partial_{\theta_2}+c\partial_{\theta_3})g=0$ so that $g$ is
        constant on the leaves of $M$.  Since
        $c\in\mathbb{R}\setminus\mathbb{Q}$, any such leaf is dense in
        $M$ and therefore $g$ is constant on $M$.

	For the existence part we shall define $f_0$ using Fourier
        series.  In particular, as $\tau_1\in C^\infty(M)$ we can
        write
	\begin{align*}
            \tau_1(\theta_2,\theta_3)
            =
            \sum_{p,q\in\mathbb{Z}}\hat{\tau}_1(p,q)e^{-i(p\theta_2+q\theta_3)}\,,
	\end{align*}
	where $\hat{\tau}(p,q)$ is fast decreasing in $p,q$, that is,
        for all $k\in\mathbb{N}$ it holds
	\begin{align*}
            \lim_{(p,q)\to\infty}(1+|p|^2+|q|^2)^{\frac{k}{2}}\hat{\tau}_1(p,q)=0\,.
	\end{align*}
	Notice that
        $\hat{\tau}_1(0,0)=\mathcal{P}_{\theta_2}\mathcal{P}_{\theta_3}\tau_1=0$.
        We shall then define $f_0$ as
	\begin{align*}
            f_0(\theta_2,\theta_3)
            =
            \sum_{p,q\in\mathbb{Z}}
            \frac{\hat{\tau}(p,q)}{i(p+cq)}e^{-i(p\theta_2+q\theta_3)}\,.
	\end{align*}
	This definition defines a solution provided the sequence
        $\frac{\hat{\tau}_1(p,q)}{i(p+cq)}$ is fast decreasing in
        $p,q$.  As $\hat{\tau}_1(p,q)$ is fast decreasing we have only
        to discuss what happens when $p,q$ are such that $p+cq$ is
        small.  However, Lemma~\ref{Lem: magnitude of rationals
          approximating irrationals} implies that for all
        $N\in\mathbb{N}$ and $p,q\in\mathbb{Z}$ such that
        $|p-cq|<\frac{1}{N}$ we have $|q|=O(N)$ as well as $|p|=O(N)$.
        This implies that $\frac{\hat{\tau}_1(p,q)}{i(p+cq)}$ is fast
        decreasing.

	Thus, proceeding as in the case of rational $c$ we have
	\begin{align*}
            [X]
            =
            [\mathcal{P}_{\theta_2}\mathcal{P}_{\theta_3}\tau_1\partial_{\theta_1}
            +
            \mathcal{P}_{\theta_1}\tau_2(\partial_{\theta_2}+c\partial_{\theta_3})
            +
            t_3\partial_{\theta_3}]\,,
	\end{align*}
	so that the linear map
        $H^1(M;\Pi_c)\ni[X]\mapsto
        (\mathcal{P}_{\theta_2}\mathcal{P}_{\theta_3}\tau_1,\mathcal{P}_{\theta_1}\tau_2,t_3)
        \in H_{\Pi_c}^0(M)^{\oplus 3}$ completes the proof.  Notice
        that since $c\in\mathbb{R}\setminus\mathbb{Q}$ we have
        $H_{\Pi_c}^0(M)\simeq\mathbb{R}$ as any smooth function
        satisfying $f_{\theta_1}=0=f_{\theta_2}+cf_{\theta_3}$ is
        necessarily constant.
    \end{description}
\end{proof}

We now investigate the convex set $\kms([X],\beta)$ of KMS
states, focusing in particular on the subset
$\kms_0([X],\beta)$ of extremal KMS states.

The result will depend on whether $c\in\mathbb{Q}$ or not.  In
particular for $c\in\mathbb{Q}$, item \ref{Item: KMS for compact
  cosymplectic manifold} of Theorem~\ref{Thm: KMS functional space for
  cosymplectic manifolds} applies: The convex set $\kms(M,\Pi_c)$ is
therefore generated by Poisson traces supported on the leaves of $M$,
while there are no KMS states for non-trivial dynamics.  This
justifies the interpretation of the leaves as pure thermodynamical
phases.

Conversely, if $c$ is irrational, the scenario changes abruptly.  None
of the leaves is proper and we cannot rely on the results of
Theorem~\ref{Thm: KMS functional space for cosymplectic manifolds} --
specifically item \ref{Item: KMS for compact cosymplectic manifold}
does not apply.  It turns out that also in this situation there are no
KMS states for non-Hamiltonian dynamics.  Concerning Poisson traces,
their convex set is rather small.  Actually there exists a unique
trace, which necessarily coincides with the regular one induced by the
smooth density $\mu$ such that $Y_\mu=0$ -- recall that cosymplectic
manifolds are unimodular \textit{cf.} Theorem~\ref{Thm: KMS functional
  space for cosymplectic manifolds}.  In this situation interpreting
leaves as thermodynamical phases seems not plausible.

Notice that the existence of a unique Poisson trace can be expected
recalling item~\ref{Item: KMS support properties} of
Proposition~\ref{Prop: properties of KMS functionals}.  In fact, if
$x\in\supp(\varphi)$ then $L_x\subseteq\supp(\varphi)$ and therefore
$\supp(\varphi)=M$ as all leaves are dense in $M$.
\begin{proposition}
    \label{Prop: 3-torus KMS functional convex set}%
    Let $(M,\Pi_c)$ be the Poisson structure defined in \eqref{Eq:
      3-torus regular codim 1 Poisson manifold}.
    \begin{enumerate}
    \item\label{Item: KMS for rational c} If
        $c=\frac{p}{q}\in\mathbb{Q}$ then
        $\kms([X],\beta)=\{0\}$ unless $[X]=0$.  Moreover,
        extremal $\Pi_c$-traces corresponds to the canonical traces
        obtained on each leaf $L$ of $M$.
    \item\label{Item: KMS for irrational c} If
        $c\in\mathbb{R}\setminus\mathbb{Q}$ then
        $\kms([X],\beta)=\{0\}$ unless $[X]=0$.  Moreover
        $\kms(M,\Pi_c)=\{I_\mu\}$ is a singleton where
        \begin{align*}
            I_\mu(f)
            =
            \frac{1}{(2\pi)^3}\int_{[0,2\pi)^3}
            f(\theta_1,\theta_2,\theta_3)
            \mathrm{d}\theta_1\mathrm{d}\theta_2\mathrm{d}\theta_3\,.
        \end{align*}
    \end{enumerate}
\end{proposition}
\begin{proof}
    Let $X\in [X]$ be as in Equation~\eqref{Eq: 3torus Poisson vector
      field}.
    \begin{description}
    \item[\ref{Item: KMS for rational c}] Since $M$ is compact with a
        proper leaf, we may apply item \ref{Item: KMS for compact
          cosymplectic manifold} of Theorem~\ref{Thm: KMS functional
          space for cosymplectic manifolds}.  This proves the claim on
        Poisson traces, moreover, it entails that :(a) if $\kms_0([X],\beta)\neq\{0\}$ then $X$ has to be tangent to the leaves of $M$; (b) $\varphi\in\kms_0([X],\beta)$ is supported on a leaf
        $L$ such that $[X|_L]$ is Hamiltonian.  However, $X|_L$ is not
        Hamiltonian for any leaf unless $X=0$, therefore
        $\kms([X],\beta)=\{0\}$ for non trivial $[X]$.
    \item[\ref{Item: KMS for irrational c}] Let
        $X = t_1\partial_{\theta_1} + t_2(\partial_{\theta_2} +
        c\partial_{\theta_3}) + t_3\partial_{\theta_3}\in[X]$, where
        $t_1,t_2,t_3\in\mathbb{R}$, and let
        $\varphi\in\kms([X],\beta)$.  We shall show that
        $\varphi\neq 0$ implies that $t_1=t_2=t_3=0$.  Lemma~\ref{Lem:
          non-existence of KMS functional supported on Z with dynamics
          containing modular class} entails that
        $\kms([X],\beta)=\{0\}$ unless $t_3=0$.  We now prove
        that $t_2=0$.  For all $f\in C^\infty_{\mathrm{c}}(M)$ we
        shall consider Equation~\eqref{Eq: global classical KMS
          condition}.  Let $\chi_1$, $\chi_2$ be a partition of unity
        associated with the covering
        $(0,2\pi) =
        [(0,\frac{\pi}{2})\cup(\frac{3\pi}{4},2\pi)]\cup(\frac{\pi}{4},\pi)$.
        We may then compute
	\begin{align*}
            \varphi(\partial_{\theta_1}f)
            &=
            \varphi(\partial_{\theta_1}(f\chi_1))
            +
            \varphi(\partial_{\theta_1}(f\chi_2))
            \\
            &=
            \varphi(\{f\chi_1,\theta_2 1_f\})
            +
            \varphi(\{f\chi_2,\theta_2 1_f\})
            \\
            &=
            -\beta\varphi(X(\theta_2)f\chi_1)
            -\beta\varphi(X(\theta_2)f\chi_2)
            \\
            &=
            -\beta t_2\varphi(f)\,,
	\end{align*}
	where $1_f\in C^\infty_{\mathrm{c}}(M)$ is such that
        $1_f|_{\supp(f)}=1$.  Overall we find
        $\varphi(f)=e^{\beta t_2s}\varphi(\Psi^*_sf)$ for all
        $s\in\mathbb{R}$, where
        $\Psi_s=\Psi^{(\partial_{\theta_1})}_s$ denotes the flow
        associated with $\partial_{\theta_1}$.  As
        $|\varphi(\Psi^*_sf)|$ is bounded in $s\in\mathbb{R}$ we find
	\begin{align*}
            \varphi(f)
            =
            \begin{cases}
                \lim\limits_{s\to+\infty}
                e^{\beta t_2 s}\varphi(\Psi^*_sf) = 0
                & \textrm{for } t_2<0
                \\
                \lim\limits_{s\to-\infty}
                e^{\beta t_2 s}\varphi(\Psi^*_sf) = 0
                & \textrm{for } t_2>0\,,
            \end{cases}
        \end{align*}
        which shows that $t_2=0$ if $\varphi\neq 0$.
	With a similar argument we find
	\begin{align*}
            \varphi(f)
            =
            e^{\beta t_1s}\varphi(\Psi_s^*f)\,,
            \qquad
            \forall s\in\mathbb{R}\,,
	\end{align*}
	where $\Psi_t$ denotes the flow associated with
        $\partial_{\theta_2}+c\partial_{\theta_3}$.  Again,
        $|\varphi(\Psi_s^*f)|$ is bounded in $s$ and therefore
	\begin{align*}
            \varphi(f)
            =
            \begin{cases}
		\lim\limits_{s\to-\infty}
                e^{\beta t_1s}\varphi(\Psi^*_s f) = 0
                & \textrm{for } t_1>0
                \\
		\lim\limits_{s\to+\infty}
                e^{\beta t_1s}\varphi(\Psi^*_s f) = 0
                & \textrm{for } t_1<0\,.
            \end{cases}
        \end{align*}
	It follows that $\varphi\neq0$ entails $t_1=0$ as
        well.

	It remains to prove that there exists a unique trace.  Let
        $\varphi\in\KMS(M,\Pi_c)$ and let $0\leq\chi_j\in C^\infty(M)$
        such that $\mathcal{P}_{\theta_j}\chi_j=1$ for all
        $j\in\{1,2,3\}$.  If follows that for all
        $f\in C^\infty_{\mathrm{c}}(M)$
	\begin{align*}
            f
            &=
            f-(\mathcal{P}_{\theta_1}f)\chi_1
            +
            [(\mathcal{P}_{\theta_1}f)\chi_1
            -
            (\mathcal{P}_{\theta_1}\mathcal{P}_{\theta_2}\mathcal{P}_{\theta_3}f)
            \chi_1\chi_2\chi_3]
            +
            (\mathcal{P}_{\theta_1}\mathcal{P}_{\theta_2}\mathcal{P}_{\theta_3}f)
            \chi_1\chi_2\chi_3
            \\
            &=
            \partial_{\theta_1}g_1
            +
            (\partial_{\theta_2}+c\partial_{\theta_3})g_2
            +
            (\mathcal{P}_{\theta_1}\mathcal{P}_{\theta_2}\mathcal{P}_{\theta_3}f)
            \chi_1\chi_2\chi_3\,,
	\end{align*}
	where $g_1,g_2\in C^\infty_{\mathrm{c}}(M)$ solves
	\begin{align*}
            \partial_{\theta_1}g_1
            =
            f-(\mathcal{P}_{\theta_1}f)\chi_1
            =:
            f_1\,,
            \qquad
            (\partial_{\theta_2}+c\partial_{\theta_3})g_2
            =
            (\mathcal{P}_{\theta_1}f)\chi_1
            -
            (\mathcal{P}_{\theta_1}\mathcal{P}_{\theta_2}\mathcal{P}_{\theta_3}f)\chi_1\chi_2\chi_3
            =:
            f_2\,.
	\end{align*}
	Notice that $g_1,g_2$ exist on account of the conditions
        $\mathcal{P}_{\theta_1}f_1=0$ and
        $\mathcal{P}_{\theta_2}\mathcal{P}_{\theta_3}f_2=0$ --
        \textit{cf.} the proof of Lemma~\ref{Lem: H1 3-torus}.  Since
        $\partial_{\theta_1}$,
        $\partial_{\theta_2}+c\partial_{\theta_3}$ are locally
        Hamiltonian a partition of unity argument leads to
	\begin{align*}
            \varphi(f)
            =
            (\mathcal{P}_{\theta_1}\mathcal{P}_{\theta_2}\mathcal{P}_{\theta_3}f)
            \varphi(\chi_1\chi_2\chi_3)
            =
            c\int_{\mathbb{T}^3}
            f(\theta_1,\theta_2,\theta_3)
            \mathrm{d}\theta_1\mathrm{d}\theta_2\mathrm{d}\theta_3\,,
	\end{align*}
	where $c=\frac{1}{(2\pi)^3}$ on account of the condition
        $\varphi(1)=1$.
    \end{description}
\end{proof}

%
% A Regular Poisson Structure on $\mathbb{R}^2\times\mathbb{T}$
%

\subsection{A Regular Poisson Structure on
  $\mathbb{R}^2\times\mathbb{T}$}
\label{Sec: spiral regular Poisson manifolds}

We now provide an example of a non-compact non-unimodular regular
Poisson manifold $(M,\Pi)$.  In particular, the symplectic foliation
of $M$ is made by leaves which are not embedded submanifolds of $M$ -- with one
exception.  Nevertheless, the structure of the KMS convex cone is
rather rich -- \textit{cf.} Proposition~\ref{Prop: spiral KMS
  functional convex set}.  In fact, there exist KMS functionals for
dynamics $[X]$ which differs from the modular class $[Y_\Pi]$.
Moreover, we can still find KMS functionals supported on leaves.  This
is a hint of the fact that these leaves fail to be submanifolds of $M$
in a rather controlled way.

We consider
\begin{align}
    \label{Eq: spiral regular codim 1 Poisson manifold}
    M
    =
    \mathbb{R}^2\times\mathbb{T}\,,
    \qquad
    \Pi
    =
    (x\partial_x+\partial_\theta)\wedge\partial_y\,,
\end{align}
where $x, y \in \mathbb{R}$ while $\theta$ is the usual quasi-global
coordinate on $\mathbb{T}$.  The symplectic foliation of $(M,\Pi)$ is
easily computed to be the collections of codimension $1$ submanifolds
\begin{align}
    \label{Eq: spiral symplectic foliation}
    L_{(x_0,\theta_0,y_0)}
    :=
    \begin{cases}
	\left\lbrace(x_0e^u,\theta_0+u,y)\,|\,u,y\in\mathbb{R}\right\rbrace
        & \textrm{for }	x_0\neq 0
        \\
	\left\lbrace(0,\theta,y)\,|\,y\in\mathbb{R}\,,\,\theta\in[0,2\pi)\right\rbrace
        & \textrm{for } x_0= 0.
    \end{cases}
\end{align}
As $L_{(x_0,\theta_0,y_0)}=L_{(x_0e^{\theta_0},0,0)}$, in what follows
we shall denote the leaves as $L_x:=L_{(x,0,0)}$.  Notice in
particular that $L_{xe^{2\pi}}=L_x$, therefore, the leaf space is the
collection
\begin{align*}
    \left\{
        L_x
        \,\big|\,
        x \in (-e^{-2\pi},-1] \cup \{0\} \cup[1,e^{2\pi})
    \right\}\,.
\end{align*}
Notice that none of these leaves is a proper submanifold of $M$ except
for the case $x_0 = 0$ -- \textit{cf.} Figure~\ref{Fig: regular codim
  1 symplectic foliation}.
\begin{remark}
    Denoting $\eta:=x\mathrm{d}\theta-\mathrm{d}x$ we have that
    $\eta\in\Omega^1(M)$ is no-where vanishing while the symplectic
    foliation \eqref{Eq: spiral symplectic foliation} coincides with
    the one induced by $\eta$.  Moreover, if we consider
    $\omega:=\mathrm{d}y\wedge\mathrm{d}\theta\in\Omega^2(M)$ we have
    that $\eta\wedge\omega$ is never vanishing.  In this situation the
    Poisson tensor $\Pi$ of $M$ is shown to coincide with the one
    defined in \eqref{Eq: Poisson tensor for cosymplectic manifolds}.
    However, notice that $\eta$ is not closed and therefore $M$ is not
    cosymplectic. In fact, $M$ is called an \defterm{almost cosymplectic}
    manifold, see
    \cite[Def.~2.1]{CappellettiMontano-DeNicola-Yudin-13}.  This
    entails in particular that this example is not covered by the
    results discussed in the previous sections.
\end{remark}

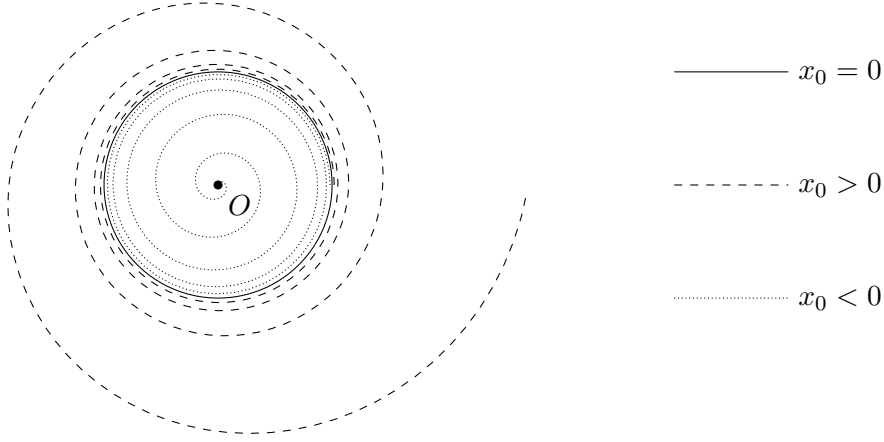
\begin{figure}[h!]
	\centering
	\begin{tikzpicture}[scale=1.5]
		\draw[-] (4,1) -- (5,1) node[right] {$x_0=0$};
		\draw[dashed] (4,0) -- (5,0) node[right] {$x_0>0$};
		\draw[densely dotted] (4,-1) -- (5,-1) node[right] {$x_0<0$};
		\filldraw (0,0) circle (1pt) node[below right] {$O$};
		\draw[variable=\t,domain=0:360,samples=100]
		plot ({cos(\t)},{sin(\t)});
		\draw[dashed,variable=\t,domain=-4:0,samples=300]
		plot ({exp(exp(\t))*cos(360*\t)},{exp(exp(\t))*sin(360*\t)});
		\draw[densely dotted,variable=\t,domain=-4:3,samples=300]
		plot ({exp(-exp(\t))*cos(360*\t)},{exp(-exp(\t))*sin(360*\t)});
	\end{tikzpicture}
	\caption{The regular foliation of $(M,\Pi)$ visualized in the
		plane -- barring the $y$-direction -- where the radial
		coordinate $r$ is identified by $\log(r)=x_0e^u$ ---\textit{cf.}
		Equation \eqref{Eq: spiral symplectic foliation}.  Symplectic
		leaves with $x_0>0$ are spirals with flows anticlockwise at
		infinity, starting (at $-\infty$) at the disk $\mathbb{D}_1$ of
		radius $r=1$.
		For $x_0<0$ we get a spiral which flows anticlockwise from the disk $\mathbb{D}_1$ to the origin.
		Finally $x_0=0$ is $\mathbb{D}_1$.}
	\label{Fig: regular codim 1 symplectic foliation}
\end{figure}

In order to study KMS functionals on $(M,\Pi)$ we first compute the
first Poisson cohomology group $H_\Pi^1(M)$ of $(M,\Pi)$.
\begin{lemma}
    \label{Lem: H1 spiral}%
    For the Poisson manifold given by Equation~\eqref{Eq: spiral
      regular codim 1 Poisson manifold} one has
    \begin{align}
        H_\Pi^1(M)
        \simeq
        \mathbb{R}^2\,.
    \end{align}
    More precisely, any element $[X]\in H_\Pi^1(M)$ is a linear
    combination of $[\partial_\theta]$ and $[\partial_y]$.  Finally,
    the modular class $Y_\Pi$ of $(M,\Pi)$ corresponds to
    $[\partial_y]$.
\end{lemma}
\begin{proof}
    Let $X=\xi\partial_x+\eta\partial_y+\tau\partial_\theta$ be a
    Poisson vector field on $(M,\Pi)$ where
    $\xi, \eta, \tau \in C^\infty(M)$.  The condition
    $\mathcal{L}_X(\Pi) = 0$ implies the following equations:
    \begin{subequations}
        \begin{align}
            \label{Eq: DxDy Poisson condition}
            \xi-x\xi_x-\xi_\theta-x\eta_y
            &=0
            \\\label{Eq: DthetaDy Poisson condition}
            x\tau_x+\tau_\theta+\eta_y
            &=0
            \\\label{Eq: DthetaDx Poisson condition}
            \xi_y-x\tau_y
            &=0\,.
        \end{align}
    \end{subequations}
    Evaluation of Equation~\eqref{Eq: DxDy Poisson condition} at $x=0$
    ensures that $\xi(0,y,\theta)=0$, that is,
    $\xi(x,y,\theta)=x\alpha(x,y,\theta)$ for $\alpha\in C^\infty(M)$.
    Equations~\eqref{Eq: DxDy Poisson condition} and \eqref{Eq:
      DthetaDx Poisson condition} simplify to
    \begin{align}
        x\alpha_x+\alpha_\theta+\eta_y
        =
        0\,,
        \qquad
        \alpha_y=\tau_y\,.
    \end{align}
    Combining the found equations with \eqref{Eq: DthetaDy Poisson
      condition} it follows that
    \begin{align*}
        X
        =
        \alpha(x\partial_x+\partial_\theta)
        +
        \eta\partial_y
        +
        A\partial_\theta\,,
    \end{align*}
    where
    \begin{align}
        \eta_y
        =
        -(x\alpha_x
        +\alpha_\theta)\,,
        \qquad
        A_y=0\,,
        \qquad
        xA_x+A_\theta=0\,.
    \end{align}
    The equations for $A$ imply that $A$ is a function of $x$ and
    $\theta$ only. Moreover, $A(0,\theta)$ is constant.  Furthermore
    one has
    \begin{align*}
        A(x,\theta)
        =
        A(xe^s,\theta+s)
        \stackrel{s=-2\pi}{=}A(xe^{-2\pi},\theta)
        =
        \lim_{n\to +\infty} A(xe^{-2\pi n},\theta)
        =
        A(0,\theta)\,,
    \end{align*}
    which proves that, in fact, $A\in\mathbb{R}$ is constant.

    We now consider the cohomology class $[X]\in H_\Pi^1(M)$ of $X$.
    We observe that, for all $f\in C^\infty(M)$,
    \begin{align}\label{Eq: regular codim 1 Hamiltonian vector field}
        X_f
        =
        f_y(x\partial_x+\partial_\theta)
        -
        (xf_x+f_\theta)\partial_y\,.
    \end{align}
    Let $f\in C^\infty(M)$ be such that $f_y=\alpha$.
    Equation~\eqref{Eq: DthetaDy Poisson condition} implies that
    $\eta=b-(xf_x+f_\theta)$, being $b\in C^\infty(M)$ with $b_y=0$.
    Comparing with \eqref{Eq: regular codim 1 Hamiltonian vector
      field} we have
    \begin{align*}
        [X]
        =
        [b\partial_y+A\partial_\theta]\,.
    \end{align*}
    We now consider the task of finding $f\in C^\infty(M)$ such that
    $X_f=b\partial_y$: Equation~\eqref{Eq: regular codim 1 Hamiltonian
      vector field} entails
    \begin{align*}
        xf_x+f_\theta=-b\,,
        \qquad
        f_y=0\,.
    \end{align*}
    These last equations give
    \begin{align}\label{Eq: solution of xDx+Dtheta PDE with source}
        f(xe^s,\theta+s)
        =
        f(x,\theta)
        -
        \int_0^sb(xe^t,\theta+t)\mathrm{d}t\,,
    \end{align}
    which in particular implies the condition
    \begin{align}\label{Eq: b-condition}
        \int_0^{2\pi}b(0,\theta)\mathrm{d}\theta=0\,.
    \end{align}
    The latter condition is necessary and sufficient to determine $f$
    up to a constant, as we shall see now.  Indeed, by assuming
    Equation \eqref{Eq: b-condition} and letting $s=-2\pi n$,
    $n\in\mathbb{N}$, in Equation \eqref{Eq: solution of xDx+Dtheta
      PDE with source} we find
    \begin{align*}
        f(xe^{-2\pi n},\theta)
        =
        f(x,\theta)
        -
        \int_0^{-2\pi n}b(xe^t,\theta+t)\mathrm{d}t
        =
        f(x,\theta)
        -
        \int_0^{-2\pi n}[b(xe^t,\theta+t)-b(0,\theta+t)]\mathrm{d}t\,,
    \end{align*}
    where we exploited the assumption on $b$.  We also observe that
    $f(0,\theta)$ is determined up to a constant by $b(0,\theta)$ as
    one can see by evaluating at $x=0$ the differential equation
    $xf_x+f_\theta=-b$ ---this is another way to recover the necessary
    condition in Equation \eqref{Eq: b-condition}.  Therefore,
    considering the limit $n\to+\infty$ of the previous equation we
    obtain
    \begin{align}\label{Eq: explicit solution of xDx+Dtheta PDE with source}
        f(x,\theta)=f(0,0)
        -\int_0^\theta b(0,s)\mathrm{d}s
        -\int_{-\infty}^0[b(xe^t,\theta+t)-b(0,\theta+t)]\mathrm{d}t\,.
    \end{align}
    Notice that integral on $(-\infty,0)$ converges because the
    function $c(x,t):=b(x,\theta)-b(0,\theta)$ vanishes at $x=0$,
    therefore, we can find $\tilde{c}\in C^\infty(M)$ such that
    $c(x,\theta)=x\tilde{c}(x,\theta)$.  This entails that
    $|c(xe^t,\theta+t)|=|xe^t\tilde{c}(xe^t,\theta+t)|\leq C|x|e^t$, where $C:=\sup_{x\in(0,1),\theta\in[0,2\pi)}|\tilde{c}(x,\theta)|$, which proves the convergence of the integral.

    Equation \eqref{Eq: explicit solution of xDx+Dtheta PDE with
      source} provides an explicit formula for $f$ under the assumption
    of Equation \eqref{Eq: b-condition}.  As a matter of fact an
    explicit computation shows that Equation \eqref{Eq: explicit
      solution of xDx+Dtheta PDE with source}, together with the
    choice of the constant $f(0,0)\in\mathbb{R}$, defines a smooth
    function $f\in C^\infty(M)$ which solves $xf_x+f_\theta=-b$.
    Indeed, smoothness of $f$ as per Equation \eqref{Eq: explicit
      solution of xDx+Dtheta PDE with source} is clear ---notice that
    $\theta$-periodicity is ensured by the periodicity condition of
    $b$ as well as Equation \eqref{Eq: b-condition}--- whereas by
    direct inspection we have
    \begin{align*}
        (x\partial_x+\partial_\theta)f(x,\theta)
        =-b(0,\theta)
        -\int_{-\infty}^0\frac{\mathrm{d}}{\mathrm{d}t}\left[b(xe^t,\theta+t)-b(0,\theta+t)\right]\mathrm{d}t
        =-b(x,\theta)\,.
    \end{align*}
    Overall we have that, setting
    $B:=\frac{1}{2\pi}\int_0^{2\pi}b(0,\theta)\mathrm{d}\theta$,
    \begin{align*}
        [X]=[(b-B)\partial_y+B\partial_y+A\partial_\theta]
        =[B\partial_y+A\partial_\theta]\,.
    \end{align*}
    It remains to show that $[B\partial_y]\neq[0]$ if $B\neq 0$: This
    has already been shown in \cite{Weinstein-97}.  Indeed,
    $B\partial_y=X_f$ for $f\in C^\infty(M)$ would imply
    \begin{align*}
        xf_x+f_\theta=B\,,\qquad f_y=0\,.
    \end{align*}
    Considering $F(x):=\int_0^{2\pi}f(x,\theta)\mathrm{d}\theta$, the
    first equation implies $xF'(x)=B$ which leads to a contradiction
    as $F\in C^\infty(\mathbb{R})$.

    The modular class $Y_\Pi$ can be computed considering
    $\mu=\mathrm{d}x\mathrm{d}y\mathrm{d}\theta$, thus proving the
    rest of the assertion.
\end{proof}

We now investigate the convex cone $\KMS([X],\beta)$ for all
$[X]\in H_\Pi^1(M)$.  To this avail the following lemma with be
relevant.
\begin{lemma}
    \label{Lem: range of Dbeta}%
    Let $D_\beta$ be the differential operator defined by
    \begin{align}
        \label{Eq: Dbeta differential operator}
        D_\beta f
        :=
        xf_x
        +
        f_\theta
        +
        \beta f\,.
    \end{align}
    Moreover, for all $\theta$ and
    $x\in(-e^{2\pi},-1]\cup[1,e^{2\pi})$ let $_x\psi_\theta^\beta$ be
    the linear functional
    \begin{align}
        \label{Eq: xpsitheta functional}
        _x\psi_\theta^\beta\colon
        C^\infty_{\mathrm{c}}(\mathbb{R}\times\mathbb{T})\to\mathbb{C}\,,
        \qquad
        _x\psi_\theta^\beta(f)
        :=
        \int_{\mathbb{R}}f(x e^u,\theta+u)e^{\beta u}\mathrm{d}u\,.
    \end{align}
    Let $f\in C^\infty_{\mathrm{c}}(\mathbb{R}\times \mathbb{T})$.
    Then the following two conditions are equivalent:
    \begin{enumerate}
    \item \label{item:ExistenceFunctiong} There exists a function
        $g\in C^\infty_{\mathrm{c}}(\mathbb{R}\times\mathbb{T})$ such
        that $D_\beta g=f$.
    \item \label{item:PsixThetaNull} One has
        $_x\psi_\theta^\beta(f)=0$ for all $\theta$ and
        $x\in(-e^{2\pi},-1]\cup[1,e^{2\pi})$.
    \end{enumerate}
\end{lemma}
\begin{proof}
    We first notice that, for all $\theta$ and all
    $x\in(-e^{2\pi},-1]\cup[1,e^{2\pi})$ as well as
    $f\in C^\infty_{\mathrm{c}}(\mathbb{R}\times\mathbb{T})$,
    $_x\psi_\theta^\beta(f)$ is well-defined on account of the
    exponential decay of $e^{\beta u}$ as $u\to-\infty$ and thanks to
    the compact support of $f$ for $u\to+\infty$.  We prove the two
    implications separately.
    \begin{description}
    \item[$\boxed{\ref{item:ExistenceFunctiong}\Rightarrow\ref{item:PsixThetaNull}}$] Let
        $g\in C^\infty_{\mathrm{c}}(\mathbb{R}\times\mathbb{T})$ be
        such that $D_\beta g=f$.  For all $\theta\in[0,2\pi)$ and
        $x\in(-e^{2\pi},-1]\cup[1,e^{2\pi})$ a direct computation
        leads to
	\begin{align*}
            _x\psi_\theta^\beta(f)
            =
            \int_{\mathbb{R}}
            (D_\beta g)(xe^u,\theta+u)e^{\beta u}
            \mathrm{d}u
            =
            \int_{\mathbb{R}}
            \frac{\mathrm{d}}{\mathrm{d}u}
            [g(xe^u,\theta+u)e^{\beta u}]
            \mathrm{d}u
            =
            0\,.
	\end{align*}
    \item[$\boxed{\ref{item:PsixThetaNull}\Rightarrow\ref{item:ExistenceFunctiong}}$] Let
        $f\in C^\infty_{\mathrm{c}}(\mathbb{R}\times\mathbb{T})$ be
        such that $_x\psi_\theta^\beta(f)=0$ for all
        $\theta\in[0,2\pi)$ and $x\in(-e^{2\pi},-1]\cup[1,e^{2\pi})$.
        Any smooth solution
        $g\in C^\infty_{\mathrm{c}}(\mathbb{R}\times\mathbb{T})$ of
        $D_\beta g=f$ would satisfy
	\begin{align*}
            g(x,\theta)
            =
            e^{\beta u}g(xe^u,\theta+u)
            -
            \int_0^u
            f(xe^v,\theta+v)e^{\beta v}
            \mathrm{d}v\,,
            \qquad
            \forall u\in\mathbb{R}\,.
	\end{align*}
	Letting $u=-2\pi n$, $n\in\mathbb{N}$, and considering the
        limit $n\to+\infty$ the previous equation leads to
	\begin{align}\label{Eq: D_beta g=f solution}
            g(x,\theta)
            =
            \int_{-\infty}^0f(xe^v,\theta+v)e^{\beta v}\mathrm{d}v\,.
	\end{align}
	On the other hand, Equation~\eqref{Eq: D_beta g=f solution}
        defines a smooth function
        $g\in C^\infty(\mathbb{R}\times\mathbb{T})$ which satisfies
        $D_\beta g=f$ ---notice that the $\theta$-periodicity
        condition on $g$ is ensured by the one on $f$ together with
        the fact that the above integral is uniformly bounded by
        $\|f\|_\infty\int_{-\infty}^0 e^{\beta u}\mathrm{d}u$.  We
        shall now prove that in fact
        $g\in C^\infty_{\mathrm{c}}(\mathbb{R}\times\mathbb{T})$.  If
        $\supp (f) \subseteq (-\infty, x_0] \times \mathbb{T}$,
        $x_0>0$, then for all $x>x_0$
        \begin{equation*}
            g(x,\theta)
            =
            \int_{-\infty}^0
            f(xe^u,\theta+u)e^{\beta u}
            \mathrm{d}u
            =
            \int_{-\infty}^\infty
            f(xe^u,\theta+u)e^{\beta u}
            \mathrm{d}u
            =
            { }_x\psi_\theta^\beta(f)
            =
            0,
        \end{equation*}
        since the integration for positive $u$ will not contribute.
        Conversely, for $x_0 < 0$ such that
        $\supp(f) \subseteq [x_0, \infty) \times \mathbb{T}$ we can
        again extend the integration from $(-\infty, 0]$ to all of
        $\mathbb{R}$ by the same argument, proving that
        $g(x,\theta)=0$ for all $x<x_0$.  Since $f$ has compact
        support, we find a large enough $x > 0$ with
        $\supp(f) \subseteq [-x, x] \times \mathbb{T}$ and then the
        function $g$ has support in this compact subset as well.
    \end{description}
\end{proof}
\begin{remark}
    Notice that, for all $s\in\mathbb{R}$, it holds
    \begin{align*}
        _{xe^s}\psi_{\theta+s}^\beta(f)
        =
        e^{-\beta s}\,_x\psi_\theta^\beta(f)\,.
    \end{align*}
    This entails that condition \ref{item:PsixThetaNull} of Lemma~\ref{Lem: range
      of Dbeta} is equivalent to
    \begin{align*}
        _{\pm 1}\psi_\theta^\beta(f)=0\,,
        \qquad
        \forall \theta\in[0,2\pi)\,.
    \end{align*}
\end{remark}

We now have the following complete characterization for $\KMS([A\partial_\theta+B\partial_y],\beta)$.
This example is particularly interesting because the convex cone of
KMS functionals is non-trivial for several dynamics.  First of all
there exists a unique (up to a prefactor) Poisson trace, induced by
the unique proper leaf of $M$.  The existence of such a leaf makes
also possible to have a unique (up to a prefactor) KMS functional
associated with $[\partial_\theta]$.  The latter provides an example
of a convex cone of KMS functionals which is non-trivial though it
does not contain regular functionals as per Definition~\ref{Def:
  regular classical KMS functional}.

KMS functionals associated with multiples of the modular class
$[Y_\mu]=[\partial_y]$ lead to more interesting results.  In fact
$\KMS([B\partial_y],\beta)$ is non-empty for all $B>0$ and
$\beta>0$.  While for $\beta=1$ this follows from Remark \ref{Rmk:
  properties of the modular class}, all other values of $\beta$ are
not for free and the rather large class of functionals found is
somehow surprising.  In fact, $\KMS([\partial_y],\beta)$ is ``generated'' by KMS
functionals supported on non-compact leaves of $M$, where $Y_\Pi$ is
Hamiltonian.  Conversely, for $B<0$ the convex cone
$\KMS([B\partial_y],\beta)$ reduces to $\{0\}$ for all
$\beta>0$.  This is interesting, because it shows that the restriction
of $-Y_\Pi$ on non-compact leaves does not lead to a suitable
Hamiltonian, providing indirect information on how badly those leaves
embed in $M$.  Overall $\KMS([B\partial_y],\beta)$ shows a phase transition in the
parameter $B\in\mathbb{R}$.  In fact $\KMS([B\partial_y],\beta)$ is isomorphic to
a measure space for $B>0$, a singleton for $B=0$ and $\{0\}$ for
$B<0$.
\begin{proposition}
    \label{Prop: spiral KMS functional convex set}%
    Let $(M,\Pi)$ be the Poisson manifold $(M,\Pi)$ defined by
    Equation~\eqref{Eq: spiral regular codim 1 Poisson manifold} and
    let $[X]=[A\partial_\theta+B\partial_y]\in H_\Pi^1(M)$ where
    $A,B\in\mathbb{R}$.  Then
    $\KMS([A\partial_\theta+B\partial_y],\beta)$
    fits in one of the following cases:
    \begin{center}
        \begin{tabular}{c|c|c|c}
            & $B<0$ & $B=0$ & $B>0$ \\
            \hline
            $A\neq 0$&\multirow{2}{*}{$\{0\}$}&\multirow{2}{*}{$[0,+\infty)$}&$\{0\}$
            \\\cline{1-1}\cline{4-4}
            $A= 0$&&&$\mathcal{M}_+(\mathbb{T})^2$
        \end{tabular}
    \end{center}
    where $\mathcal{M}_+(\mathbb{T})$ denotes the set of positive
    Borel measures on $\mathbb{T}$.  In particular:
    \begin{enumerate}
    \item \label{item:KMSANullBeta} The unique (up to multiplicative
        constant) element $\varphi\in\KMS([A\partial_\theta],\beta)$ is given by
        \begin{align*}
            \varphi(f)
            :=
            \int_0^{2\pi}\int_{\mathbb{R}}f(0,\theta,y)e^{-A\beta y}\mathrm{d}\theta\mathrm{d}y\,.
        \end{align*}
    \item \label{item:KMSPiTraceUnique} The unique (up to
        multiplicative constant) $\Pi$-trace $\varphi\in\KMS(M,\Pi)$
        is given by
        \begin{align*}
            \varphi(f)
            :=
            \int_0^{2\pi}\int_{\mathbb{R}}f(0,\theta,y)\mathrm{d}\theta\mathrm{d}y\,.
        \end{align*}
    \item \label{item:KMSNullOneBeta} For all
        $\varphi\in\KMS([\partial_y],\beta)$ there exists positive measures
        $\nu_\varphi^\pm\in \mathcal{M}_+(\mathbb{T})$ such that
        \begin{align}
            \label{Eq: spiral B=1 KMS}
            \varphi(f)
            =
            \sum_\pm\int_0^{2\pi}\,_1\varphi_\theta^\beta(f)\mathrm{d}\nu_\varphi^\pm(\theta)\,,
        \end{align}
        where for all $x\neq 0$ and $\theta\in[0,2\pi)$,
        $_x\varphi_\theta^\beta$ is the KMS functional defined by
        \begin{align}
            \label{Eq: spiral B=1 extremal KMS}
            _x\varphi_\theta^\beta(f)
            :=
            \int_{\mathbb{R}^2}f(xe^u,\theta+u,y)e^{\beta u}\mathrm{d}u\mathrm{d}y\,.
        \end{align}
        It follows that $\varphi\in\KMS([\partial_y],\beta)$ is extremal if and
        only if $\varphi=\,_1\varphi_\theta^\beta$ for
        $\theta\in[0,2\pi)$.
    \end{enumerate}
\end{proposition}
\begin{proof}
    Let $X=A\partial_\theta+B\partial_y\in[X]$.
    \begin{description}
    \item[$\boxed{AB\neq 0\colon}$] Let $Z:=M\cap\{x=0\}$ and consider
        the regular codimension $1$ Poisson manifold
        $(M_Z:=M\setminus Z,\Pi|_{M_Z})$.  Since $A\neq 0$ it follows
        that $X|_{M_Z}$ is transverse to the symplectic foliation of
        $M_Z$.  Moreover if $\varphi\in\KMS(X,\beta)$ then
        $\varphi|_{C^\infty_{\mathrm{c}}(M_Z)}$ is a $(X|_{M_Z},\beta)$-KMS functional on $(M_Z,\Pi|_{M_Z})$ and, by Lemma~\ref{Lem: non-existence criterion for regular
          codim 1 Poisson manifolds}, we have
        $\varphi|_{C^\infty_{\mathrm{c}}(M_Z)}=0$.

	Therefore $\supp(\varphi)\subseteq Z$ for all
        $\varphi\in\KMS([A\partial_\theta+B\partial_y],\beta)$.
        This entails that $\varphi(f)=\psi(f|_Z)$ for all
        $f\in C^\infty_{\mathrm{c}}(M)$ where
        $\psi\colon C^\infty_{\mathrm{c}}(Z)\to\mathbb{C}$ is a
        positive linear functional.  Moreover,
        $\psi$ is a $(X|_Z,\beta)$-KMS functional on $(Z,\Pi|_Z)$, because $X$ is tangent to
        $Z$.  But $(Z,\Pi_Z)$ is symplectic and therefore
        Remark~\ref{Rmk: KMS for symplectic manifolds} ensures that
        $\psi\neq 0$ would imply
	\begin{align*}
            [0]=[X|_Z]=[B\partial_y|_Z]=[B\mathrm{d}\theta^\sharp|_Z]\,,
	\end{align*}
	a contradiction as $B\neq 0$.  It follows that $\varphi=0$.

    \item[$\boxed{A\neq B=0\colon}$] Proceeding as in the previous
        case we have that $\varphi\in\KMS([A\partial_\theta],\beta)$ implies that
        $\supp(\varphi)\subseteq Z$.  Therefore $\varphi(f)=\psi(f_Z)$
        for all $f\in C^\infty_{\mathrm{c}}(M)$, where
        $\psi$ is a $(A\partial_\theta,\beta)$-KMS functional on $(Z,\Pi|_Z)$.  Remark
        \ref{Rmk: KMS for symplectic manifolds} implies that there
        exists $c>0$ such that
	\begin{align*}
            \varphi(f)
            =
            c \int_0^{2\pi} \int_{\mathbb{R}}
            f(0,\theta,y)e^{-\beta Ay}
            \mathrm{d}\theta\mathrm{d}y\,.
	\end{align*}
    \item[$\boxed{A=0\,,\,B\leq 0\colon}$] Let
        $\varphi\in\KMS([B\partial_y],\beta)$ where $B\leq 0$.
	For all $f\in C^\infty_{\mathrm{c}}(M)$ we apply
        Equation~\eqref{Eq: global classical KMS condition} for
        $g=y1_f$ where $1_f\in C^\infty_{\mathrm{c}}(M)$ is such that
        $1_f|_{\supp(f)}=1$.  It follows that, denoting by
        $\Psi_t=\Psi^{x\partial_x+\partial_\theta}_t$ the flow
        associated with the vector field
        $X_y=x\partial_x+\partial_\theta$, it holds
	\begin{align*}
            \varphi(f)
            =
            e^{\beta Bt}\varphi(\Psi_t^*f)\,.
	\end{align*}
	The above argument shows that $\varphi$ has to be supported on
        $Z$: indeed let $f\in C^\infty_{\mathrm{c}}(M_Z)$.  Then, for
        all $t_0\in\mathbb{R}$ there exists $K_{t_0}$ such that
        $\supp(\Psi^*_tf)\subseteq K_{t_0}$ for all $t\geq t_0$.
        Since $\varphi$ is positive we have the bound
        $|\varphi(\Psi^*_tf)|\leq C_{t_0}\|f\|_\infty$ for all
        $t\geq t_0$ with $C_{t_0}>0$.  It then follows that
	\begin{align*}
            \varphi(f)
            =
            \lim_{t\to+\infty}e^{B\beta t}\varphi(\Psi^*_tf)
            =
            0\,,
	\end{align*}
	where we used that $B\leq 0$.  Notice that for $B=0$ we have
        $\lim\limits_{t\to+\infty}\varphi(\Psi_t^*f)=0$ by dominated
        convergence: Indeed $\varphi$ is positive, while
        $\lim\limits_{t\to+\infty}\Psi_t^*f=0$ pointwise and
        $|\Psi_t^*f|\leq \|f\|_\infty \varrho_{K_{t_0}}$ for
        $t\geq t_0$ ---here $\varrho_{K_{t_0}}$ denotes the
        characteristic function over $K_{t_0}$.

	As $Z$ is symplectic Remark~\ref{Rmk: KMS for symplectic
          manifolds} entails that $\varphi\neq0$ would imply
        $0=[B\partial_y]=[B\mathrm{d}\theta^\sharp]$.  It follows that
        $B<0$ implies $\varphi=0$, moreover, if $B=0$ then there is a
        constant $c>0$ such that
	\begin{align*}
            \varphi(f)
            =
            c\int_0^{2\pi}\int_{\mathbb{R}}
            f(0,\theta,y)
            \mathrm{d}y\mathrm{d}\theta\,.
	\end{align*}
    \item[$\boxed{A=0\,,\,B>0\colon}$] As
        $\KMS([B\partial_y],\beta)=\KMS([\partial_y],B\beta)$ we shall restrict to
        $B=1$.  We first observe that the functionals
        $_x\varphi_\theta^\beta$ defined in Equation~\eqref{Eq: spiral
          B=1 extremal KMS} are $(\partial_y,\beta)$-KMS functionals:
        in fact $_x\varphi_\theta^\beta$ is the pull-back on $M$ of
        the unique (up to multiplicative constant) KMS functional on
        $L_x$.  Notice that, recalling Equation~\eqref{Eq: xpsitheta
          functional}, we have
	\begin{align*}
            _x\varphi_\theta^\beta(f)
            =
            \int_{\mathbb{R}}\,_x\psi_\theta^\beta(f(\cdot,\cdot,y))\mathrm{d}y
            =
            [I_y\circ\,_x\psi_\theta^\beta](f)\,,
	\end{align*}
	where $(I_y
        f)(x,\theta):=\int_{\mathbb{R}}f(x,\theta,y)\mathrm{d}y$.  Let
        $\varphi\in\KMS([\partial_y],\beta)$ and let
        $\eta,\chi_\pm\in C^\infty_{\mathrm{c}}(\mathbb{R})$ be such
        that
	\begin{align}\label{Eq: chi, eta conditions}
            \eta,\chi_\pm\geq0\,,
            \qquad
            \int_{\mathbb{R}}\eta(y)\mathrm{d}y
            =1\,,
            \qquad
            \int_0^{\pm\infty}\chi_\pm(z)\frac{\mathrm{d}z}{z}
            =1\,,
            \qquad
            \supp(\chi_\pm)\subseteq(0,\pm\infty)\,.
	\end{align}
	For all $f\in
        C^\infty_{\mathrm{c}}(\mathbb{R}\times\mathbb{T})$ we set
	\begin{align*}
            f_\chi \in
            C^\infty_{\mathrm{c}}(\mathbb{R}\times\mathbb{T})\,,
            \qquad
            f_\chi(x,\theta)
            :=
            f(x,\theta)
            -
            \sum_\pm\,_x\psi_\theta^\beta(f)\chi_\pm(x)\,.
	\end{align*}
	Notice that the well-definedness of $f_\chi$ is due to the
        support properties of $\chi$, in particular
        $0\notin\supp(\chi)$.  A direct computation shows that, for
        all $x\neq0$ and $\theta\in[0,2\pi)$,
	\begin{align*}
            _x\psi_\theta^\beta(f_\chi)
            &=
            \,_x\psi_\theta^\beta(f)
            -
            \sum_\pm\int_{\mathbb{R}}
            \,_{xe^u}\psi_{\theta+u}(f)\chi_\pm(xe^u)e^{\beta u}
            \mathrm{d}u
            \\
            &=
            \,_x\psi_\theta^\beta(f)
            \left[
		1 -
                \sum_\pm\int_0^{\operatorname{sgn}(x)\infty}\chi_\pm(z)
                \frac{\mathrm{d}z}{z}
            \right]
            \\
            &=
            0\,,
	\end{align*}
	where we used that
        $_{xe^u}\psi_{\theta+u}(f)e^{\beta
          u}=\,_x\psi_\theta^\beta(f)$.  Lemma~\ref{Lem: range of
          Dbeta} entails that $f_\chi=D_\beta g$ for
        $g\in C^\infty_{\mathrm{c}}(\mathbb{R}\times\mathbb{T})$.

	Moreover a direct computation shows that, for all
        $f\in C^\infty_{\mathrm{c}}(M)$,
	\begin{align*}
            \varphi(xf_x+f_\theta)
            =
            \varphi(\{f,y1_f\})
            =
            -\beta\varphi(f)
            \implies
            \varphi\circ D_\beta
            =
            0\,.
	\end{align*}
	With these preliminaries we now come to the characterization
        of $\KMS([\partial_y],\beta)$.  At first we shall define two positive
        Borel measures $\nu_\varphi^\pm\in \mathcal{M}_+(\mathbb{T})$ associated
        to $\varphi$.  For that let
        $F\colon\mathbb{R}\setminus\{0\}\times\mathbb{T}\to\mathbb{T}$
        be defined by $F(x,e^{i\theta}):=e^{i(\theta-\log |x|)}$.
        For all open subsets $U\subseteq\mathbb{T}$ we then set
	\begin{align*}
            \nu_\varphi^\pm(U)
            :=
            \varphi(\varrho_{F^{-1}(U)}\eta x^{-\beta}\chi_\pm)\,,
	\end{align*}
	where $\varrho_{F^{-1}(U)}$ denotes the characteristic function over
        $F^{-1}(U)$ while we regarded
        	$\varphi$ as a positive measure on $M$ -- \textit{cf.}
        	Remark \ref{Rmk: on definition of KMS functional}.  Notice that
        $x^{-\beta}\chi_\pm\in C^\infty_{\mathrm{c}}(\mathbb{R})$ on
        account of the hypothesis on $\chi_\pm$.

	The last equation defines two positive Borel measures
        $\nu_\varphi^\pm\in \mathcal{M}_+(\mathbb{T})$ which do not depend on
        the choice of neither $\eta$ nor $\chi_\pm$.  Indeed, for any
        other $\hat{\eta}$ and $\hat{\chi}_\pm$ which satisfies
        Equation~\eqref{Eq: chi, eta conditions} we have
        $\eta-\hat{\eta}=h_y$ and $\chi_\pm-\hat{\chi}_\pm=x g^\pm_x$
        for $h,g\in C^\infty_{\mathrm{c}}(\mathbb{R})$.  We therefore
        have
	\begin{align*}
            \varphi(\varrho_{F^{-1}(U)}(\eta-\hat{\eta}) x^{-\beta}\chi_\pm)
            &=
            (\varphi\circ\partial_y)(\varrho_{F^{-1}(U)}h
            x^{-\beta}\chi_\pm)
            =0\,,
            \\
            \varphi(\varrho_{F^{-1}(U)}\eta x^{-\beta}(\chi_\pm-\hat{\chi}_\pm))
            &=
            (\varphi\circ D_\beta)(\varrho_{F^{-1}(U)}\eta x^{-\beta}g^\pm)
            =0\,,
	\end{align*}
	where we used $x^{-\beta+1}g_x^\pm=D_\beta(x^{-\beta}g^\pm)$
        as well as the fact that $D_\beta f\circ F^{-1}=0$ for all
        $f\in C^\infty(\mathbb{R}\times\mathbb{T})$.

	We now prove Equation~\eqref{Eq: spiral B=1 KMS}.  For any
        $f\in C^\infty_{\mathrm{c}}(M)$ we have
	\begin{align*}
            f
            =
            \underbrace{f-I_y(f)\eta}_{=h_y}
            +
            \underbrace{[I_y(f)\eta]_\chi}_{=D_\beta g}
            +
            \sum_\pm\,_{\cdot}\varphi_{\cdot}^\beta(f)\chi_\pm\eta\,,
	\end{align*}
	where $h,g\in C^\infty_{\mathrm{c}}(M)$.  Since
        $_x\varphi_\theta^\beta(f)=x^{-\beta}\,_1\varphi_{\theta-\log
          x}^\beta(f)$ we
        find
	\begin{align*}
            \varphi(f)
            =
            \sum_\pm
            \varphi(\,_\cdot\varphi_\cdot^\beta(f)\chi_\pm\eta)
            =
            \sum_\pm\int_0^{2\pi}
            \,_1\varphi_\theta^\beta(f)\mathrm{d}\nu_\varphi^\pm(\theta)\,.
	\end{align*}
    \end{description}
\end{proof}
\begin{example}[Symmetry breaking]
    As an example of Equation~\eqref{Eq: spiral B=1 KMS} we consider
    the case $\mathrm{d}\nu_\varphi^+(\theta)=\mathrm{d}\theta$,
    $\nu_\varphi^-=0$.  This leads to
    \begin{align*}
        \varphi(f)
        =
        \int_0^{+\infty}\int_{0}^{2\pi}\int_{\mathbb{R}}
        f(x,\theta,y)x^{\beta-1}
        \mathrm{d}x\mathrm{d}\theta\mathrm{d}y\,.
    \end{align*}
    Notice that the KMS functional obtained with such a choice is
    invariant under the flow of $\partial_\theta$.  A direct
    inspection shows that $(\partial_y,\beta)$-KMS functionals which
    are $\partial_\theta$-invariant correspond to the choice
    $\mathrm{d}\nu_\varphi^\pm(\theta)=c_\pm\mathrm{d}\theta$ with
    $c_\pm\in[0,+\infty)$.  On account of Theorem~\ref{Prop: spiral
      KMS functional convex set} none of these functionals is extremal
    in $\KMS([\partial_y],\beta)$.  Stated differently, any
    $\varphi\in\KMS([\partial_y],\beta)$ is either extremal or
    $\partial_\theta$-invariant but not both.  In physics jargon such
    a phenomenon is known with the name of symmetry breaking.
\end{example}

\paragraph{Acknowledgements}
It is a pleasure to thank R. Brunetti, G. Canepa, C. Dappiaggi,
V. Moretti, S. Murro and K.-H. Neeb for helpful discussions on the
subject.  N.D. is supported by a Postdoctoral Fellowship of the
Alexander von Humboldt Foundation (Germany).

\end{document}